\pdfoutput=1
\documentclass[a4paper,UKenglish,autoref, thm-restate,numberwithinsect,final]{lipics-v2021}

\usepackage[capitalise]{cleveref} \usepackage{tikz}
\usepackage{xspace}
\usetikzlibrary{calc}
\usetikzlibrary{positioning}
\usetikzlibrary{arrows,automata,shapes}

\RequirePackage{centernot}
\RequirePackage[Symbol]{upgreek}

\RequirePackage{etoolbox}
\RequirePackage{mathtools}

\let\originalleft\left
\let\originalright\right
\renewcommand{\left}{\mathopen{}\mathclose\bgroup\originalleft}
\renewcommand{\right}{\aftergroup\egroup\originalright}

\newcommand{\bigO}[1]{\ensuremath{O(#1)}}
\newcommand{\EXPSPACE}{\ensuremath{\operatorname{EXPSPACE}}\xspace}
\newcommand{\PSPACE}{\ensuremath{\operatorname{PSPACE}}\xspace}

\newcommand{\Nat}{\mathbb{N}}

\newcommand{\Procs}{\ensuremath{\mathcal{P}}}

\definecolor{roleColor}{rgb}{0.1, 0.3, 0.1}\newcommand{\roleCol}[1]{{\color{roleColor}#1}}\newcommand{\roleFmt}[1]{\boldsymbol{\roleCol{\mathtt{#1}}}}\newcommand{\roleP}[1][]{\ifempty{#1}{{\color{roleColor}\roleFmt{p}}}{{\color{roleColor}\roleFmt{p}_{#1}}}}
\newcommand{\roleVar}[1][]{\ifempty{#1}{{\color{roleColor}\roleFmt{x}}}{{\color{roleColor}\roleFmt{x}_{#1}}}}

\newcommand{\buyerA}{{\color{roleColor}\roleFmt{a}}}
\newcommand{\buyerB}{{\color{roleColor}\roleFmt{b}}}
\newcommand{\seller}{{\color{roleColor}\roleFmt{s}}}

\newcommand{\procA}{{\color{roleColor}\roleFmt{p}}}
\newcommand{\procB}{{\color{roleColor}\roleFmt{q}}}
\newcommand{\procC}{{\color{roleColor}\roleFmt{r}}}
\newcommand{\procD}{{\color{roleColor}\roleFmt{s}}}

\newcommand{\procAunc}{\boldsymbol{\mathtt{p}}}
\newcommand{\procBunc}{\boldsymbol{\mathtt{q}}}

\newcommand{\RcvEvs}{\ensuremath{R}}
\newcommand{\SndEvs}{\ensuremath{S}}
\newcommand{\eventnodes}{\ensuremath{N}}
\newcommand{\edges}{\ensuremath{E}}
\newcommand{\val}{\ensuremath{m}}
\newcommand{\vertexA}{\ensuremath{v}}
\newcommand{\vertexB}{\ensuremath{u}}

\newcommand{\BMSCs}{\ensuremath{\mathcal{M}}}

\newcommand{\Alphabet}{\Sigma}

\newcommand{\MsgVals}{\ensuremath{\mathcal{V}}}

\newcommand{\CSM}[1]{\ensuremath{\{\!\!\{#1_\procA\}\!\!\}_{\procA \in \Procs}}}
\newcommand{\CSMl}[1]{\ensuremath{\{\!\!\{{#1}\}\!\!\}_{\procA \in \Procs}}}

\newcommand{\emptystring}{\varepsilon}

\newcommand{\set}[1]{\{#1\}}
\newcommand{\lang}{\mathcal{L}}

\newcommand{\interswaplang}{\mathcal{C}^{\interswap}}
\newcommand{\intraswaplang}{\mathcal{C}^{\intraswap}}
\newcommand{\indeprellang}{\mathcal{C}^{\equiv_{\mathcal{I}}}}
\newcommand{\SyncToAsync}{\ensuremath{\operatorname{split}}}

\newcommand{\channels}{\ensuremath{\mathsf{Chan}}}
\newcommand{\channel}[2]{\ensuremath{\langle#1,#2\rangle}}

\newcommand{\trace}{\operatorname{trace}}

\newcommand{\GG}{\mathbf{G}}
\newcommand{\getMu}{\mathit{get\mu}}

\newcommand{\semglobal}{\ensuremath{\mathsf{GAut}}}
\newcommand{\semglobalsync}{\ensuremath{\mathsf{GAut}}}

\newcommand{\AlphSync}{\ensuremath{Σ_{\mathit{sync}}}}
\newcommand{\AlphAsync}{\ensuremath{Σ_{\mathit{async}}}}
\newcommand{\semlocal}{\ensuremath{\mathsf{LAut}}}

\newcommand{\fin}{\operatorname{fin}}

\newcommand{\infi}{\operatorname{inf}}

\newcommand{\swap}{\ensuremath{\circeq}}
\newcommand{\interswap}{\ensuremath{\sim}}
\newcommand{\intraswap}{\ensuremath{\approx}}

\DeclareMathOperator*{\ExtCh}{\&}
\DeclareMathOperator*{\IntCh}{⊕}
\newcommand{\merge}{\sqcap}
\newcommand\plainmerge{\mathrel{\ooalign{\hss$\merge$\hss\cr \kern0.36ex\raise0.35ex\hbox{\scalebox{0.7}{$p$}}}}}
\newcommand\semifullmerge{\mathrel{\ooalign{\hss$\merge$\hss\cr \kern0.36ex\raise0.3ex\hbox{\scalebox{0.7}{$s$}}}}}
\newcommand\fullmerge{\mathrel{\ooalign{\hss$\merge$\hss\cr \kern0.3ex\raise0.03ex\hbox{\scalebox{0.7}{$f$}}}}}
\newcommand\availmerge{\mathrel{\ooalign{\hss$\merge$\hss\cr \kern0.36ex\raise0.3ex\hbox{\scalebox{0.7}{$a$}}}}}

\DeclareMathOperator*{\Merge}{\merge}

\def \ifempty#1{\def\temp{#1} \ifx\temp\empty }

\newcommand{\snd}[3]{\ifempty{#1} #2!#3 \else #1\triangleright#2!#3 \fi}
\newcommand{\rcv}[3]{\ifempty{#2} #1?#3 \else #2\triangleleft#1?#3 \fi}
\newcommand{\msgFromTo}[3]{#1\!\to\!#2\!:\!#3}

\newcommand{\pref}{\operatorname{pref}}

\def\revalpha{\rotatebox[origin=c]{180}{$\alpha$}}

\newcommand{\levelfunc}{\ensuremath{\operatorname{lvl}}}

\newcommand{\proj}{{\ensuremath{\upharpoonright}}}
\newcommand{\tproj}{{\ensuremath{\upharpoonright}}}
\newcommand{\wproj}{{\ensuremath{\Downarrow}}}

\def\mmerge{\mathrel{\ThisStyle{\stretchrel*{\ooalign{\raise0.2\LMex\hbox{$\SavedStyle\sqcap$}\cr \raise-0.2\LMex\hbox{$\SavedStyle\sqcap$}}}{\sqcap}}}}

\tikzstyle{hmscarrow}=[->, thick]
\tikzstyle{bmscbox}=[rounded corners, opacity=0.5]

\newcommand{\powerset}{\ensuremath{\mathscr{P}}}

\newcommand{\union}{\cup}
\newcommand{\inters}{\cap}
\newcommand{\Union}{\bigcup}

\newcommand{\dunion}{\uplus}
\newcommand{\Dunion}{\biguplus}
\DeclarePairedDelimiter\card{\lvert}{\rvert}

\providecommand{\implies}{\Rightarrow}

\providecommand{\coloneq}{\mathrel{\mathop:}=} \providecommand{\Coloneqq}{\mathrel{\mathop{::}}=} \newcommand{\is}{\coloneq}

\def\congr{\equiv}

\newcommand{\from}{\colon}

\newcommand{\inv}[1]{#1^{-1}}

 \def\st{\colon}  
\newcommand{\ie}{i.e.~}

\def\grammOr{\hspace{3pt}\mid\hspace{3pt}}
\def\grammIs{\Coloneqq}

\begingroup
\catcode`\|=\active \gdef\@grammar@bar{\catcode`\|=\active \def|{\grammOr}}
\endgroup

\newcommand{\gramm}[1]{\begingroup
  \def\is{\grammIs}\@grammar@bar #1\endgroup }

\newenvironment{grammar}{\begin{equation*}\def\is{& \grammIs }\@grammar@bar \aligned }
{\endaligned \end{equation*}\aftergroup\ignorespaces }

\newcommand{\hole}{\hbox{-}}

\DeclarePairedDelimiterXPP\aenc[2]{\constr{a}}{(}{)}{_{#2}}{\strip@parens#1}
\DeclarePairedDelimiterXPP\pub[1]{\constr{p}}{(}{)}{}{\strip@parens#1}

\newcommand{\Parallel}{\@ifstar{\prod}{{\textstyle\prod}}}
\newcommand{\Alt}{\@ifstar{\sum}{{\textstyle\sum}}}

\newcommand{\bang}[1]{{!}#1}

\tikzset{
  sem/.style={
    transform shape,
	 node distance = 1cm,
     every state/.style = {semnode},
	 every edge/.style = {semarrow}
  }
}

\tikzset{
  semnode/.style={
    shape=circle,
    minimum size = 5mm,
    inner sep = 1pt,
    font=\tiny,
    draw
  },
  semarrow/.style={
    ->,
    shorten >=1pt,
    >=stealth',
    auto,
    font=\scriptsize,
    draw,
}
}

\newcommand{\TBPOR}{\operatorname{2BP}}
\newcommand{\TBPWS}{\operatorname{2BPWS}}
\newcommand{\TBPIR}{\operatorname{2BPIR}}
\newcommand{\MPCP}{\operatorname{MPCP}}

\newcommand{\query}{\textit{query}}
\newcommand{\price}{\textit{price}}
\newcommand{\cancel}{\textit{cancel}}
\newcommand{\no}{\textit{no}}
\newcommand{\splitmsg}{\textit{split}}
\newcommand{\yes}{\textit{yes}}
\newcommand{\buy}{\textit{buy}}
\newcommand{\done}{\textit{done}}
\newcommand{\login}{\textit{login}}
\newcommand{\subscribe}{\textit{subscribe}}
\newcommand{\subscribed}{\textit{subscribed}}

\newcommand{\queryshort}{q}
\newcommand{\priceshort}{p}
\newcommand{\cancelshort}{c}
\newcommand{\noshort}{n}
\newcommand{\splitshort}{s}
\newcommand{\yesshort}{y}
\newcommand{\buyshort}{b}
\newcommand{\doneshort}{d}

\newcommand{\erased}[1]{\ensuremath{[#1]\rightsquigarrow\emptystring}}

\tikzstyle{sndstate}=[state, diamond]
\tikzstyle{rcvstate}=[state, rectangle]
\tikzstyle{recstate}=[state, densely dashed]
\tikzstyle{varstate}=[state, densely dashdotted]
\tikzstyle{finalstate}=[state, accepting]
 \usepackage{newunicodechar}
\newunicodechar{∃}{\ensuremath{\exists}}
\newunicodechar{∀}{\ensuremath{\forall}}
\newunicodechar{θ}{\ensuremath{\theta}}
\newunicodechar{τ}{\ensuremath{\tau}}
\newunicodechar{φ}{\ensuremath{\varphi}}
\newunicodechar{ξ}{\ensuremath{\xi}}
\newunicodechar{ζ}{\ensuremath{\zeta}}
\newunicodechar{ψ}{\ensuremath{\psi}}
\newunicodechar{π}{\ensuremath{\pi}}
\newunicodechar{α}{\ensuremath{\alpha}}
\newunicodechar{β}{\ensuremath{\beta}}
\newunicodechar{γ}{\ensuremath{\gamma}}
\newunicodechar{δ}{\ensuremath{\delta}}
\newunicodechar{ε}{\ensuremath{\varepsilon}}
\newunicodechar{κ}{\ensuremath{\kappa}}
\newunicodechar{λ}{\ensuremath{\lambda}}
\newunicodechar{μ}{\ensuremath{\mu}}
\newunicodechar{ρ}{\ensuremath{\rho}}
\newunicodechar{σ}{\ensuremath{\sigma}}
\newunicodechar{ω}{\ensuremath{\omega}}
\newunicodechar{Γ}{\ensuremath{\Gamma}}
\newunicodechar{Φ}{\ensuremath{\Phi}}
\newunicodechar{Δ}{\ensuremath{\Delta}}
\newunicodechar{Σ}{\ensuremath{\Sigma}}
\newunicodechar{Π}{\ensuremath{\Pi}}
\newunicodechar{∑}{\ensuremath{\Sigma}}
\newunicodechar{∏}{\ensuremath{\Pi}}
\newunicodechar{Θ}{\ensuremath{\Theta}}
\newunicodechar{Ω}{\ensuremath{\Omega}}
\newunicodechar{⇒}{\ensuremath{\Rightarrow}}
\newunicodechar{⇐}{\ensuremath{\Leftarrow}}
\newunicodechar{⇔}{\ensuremath{\Leftrightarrow}}
\newunicodechar{→}{\ensuremath{\rightarrow}}
\newunicodechar{←}{\ensuremath{\leftarrow}}
\newunicodechar{↔}{\ensuremath{\leftrightarrow}}
\newunicodechar{¬}{\ensuremath{\neg}}
\newunicodechar{∧}{\ensuremath{\land}}
\newunicodechar{∨}{\ensuremath{\lor}}
\newunicodechar{≠}{\ensuremath{\neq}}
\newunicodechar{≡}{\ensuremath{\equiv}}
\newunicodechar{∼}{\ensuremath{\sim}}
\newunicodechar{≈}{\ensuremath{\approx}}
\newunicodechar{≥}{\ensuremath{\geq}}
\newunicodechar{≤}{\ensuremath{\leq}}
\newunicodechar{≫}{\ensuremath{\gg}}
\newunicodechar{≪}{\ensuremath{\ll}}
\newunicodechar{∅}{\ensuremath{\emptyset}}
\newunicodechar{⊆}{\ensuremath{\subseteq}}
\newunicodechar{⊂}{\ensuremath{\subset}}
\newunicodechar{∩}{\ensuremath{\cap}}
\newunicodechar{⋂}{\ensuremath{\cap}}
\newunicodechar{∪}{\ensuremath{\cup}}
\newunicodechar{⋃}{\ensuremath{\cup}}
\newunicodechar{⊎}{\ensuremath{\uplus}}
\newunicodechar{∈}{\ensuremath{\in}}
\newunicodechar{∉}{\ensuremath{\not\in}}
\newunicodechar{⊤}{\ensuremath{\top}}
\newunicodechar{⊥}{\ensuremath{\bot}}
\newunicodechar{₀}{\ensuremath{_0}}
\newunicodechar{₁}{\ensuremath{_1}}
\newunicodechar{₂}{\ensuremath{_2}}
\newunicodechar{₃}{\ensuremath{_3}}
\newunicodechar{₄}{\ensuremath{_4}}
\newunicodechar{₅}{\ensuremath{_5}}
\newunicodechar{₆}{\ensuremath{_6}}
\newunicodechar{₇}{\ensuremath{_7}}
\newunicodechar{₈}{\ensuremath{_8}}
\newunicodechar{₉}{\ensuremath{_9}}
\newunicodechar{⁰}{\ensuremath{^0}}
\newunicodechar{¹}{\ensuremath{^1}}
\newunicodechar{²}{\ensuremath{^2}}
\newunicodechar{³}{\ensuremath{^3}}
\newunicodechar{⁴}{\ensuremath{^4}}
\newunicodechar{⁵}{\ensuremath{^5}}
\newunicodechar{⁶}{\ensuremath{^6}}
\newunicodechar{⁷}{\ensuremath{^7}}
\newunicodechar{⁸}{\ensuremath{^8}}
\newunicodechar{⁹}{\ensuremath{^9}}
\newunicodechar{𝔹}{\ensuremath{\mathbb{B}}}
\newunicodechar{ℝ}{\ensuremath{\mathbb{R}}}
\newunicodechar{ℕ}{\ensuremath{\mathbb{N}}}
\newunicodechar{ℂ}{\ensuremath{\mathbb{C}}}
\newunicodechar{ℚ}{\ensuremath{\mathbb{Q}}}
\newunicodechar{𝕋}{\ensuremath{\mathbb{T}}}
\newunicodechar{𝕏}{\ensuremath{\mathbb{X}}}
\newunicodechar{ℤ}{\ensuremath{\mathbb{Z}}}
\newunicodechar{✓}{\checkmark}
\newunicodechar{✗}{\ensuremath{\times}}
\newunicodechar{◊}{\ensuremath{\lozenge}}
\newunicodechar{□}{\ensuremath{\square}}
\newunicodechar{𝓐}{\ensuremath{\mathcal{A}}}
\newunicodechar{𝓑}{\ensuremath{\mathcal{B}}}
\newunicodechar{𝓒}{\ensuremath{\mathcal{C}}}
\newunicodechar{𝓓}{\ensuremath{\mathcal{D}}}
\newunicodechar{𝓔}{\ensuremath{\mathcal{E}}}
\newunicodechar{𝓕}{\ensuremath{\mathcal{F}}}
\newunicodechar{𝓖}{\ensuremath{\mathcal{G}}}
\newunicodechar{𝓗}{\ensuremath{\mathcal{H}}}
\newunicodechar{𝓘}{\ensuremath{\mathcal{I}}}
\newunicodechar{𝓙}{\ensuremath{\mathcal{J}}}
\newunicodechar{𝓚}{\ensuremath{\mathcal{K}}}
\newunicodechar{𝓛}{\ensuremath{\mathcal{L}}}
\newunicodechar{𝓜}{\ensuremath{\mathcal{M}}}
\newunicodechar{𝓝}{\ensuremath{\mathcal{N}}}
\newunicodechar{𝓞}{\ensuremath{\mathcal{O}}}
\newunicodechar{𝓟}{\ensuremath{\mathcal{P}}}
\newunicodechar{𝓠}{\ensuremath{\mathcal{Q}}}
\newunicodechar{𝓡}{\ensuremath{\mathcal{R}}}
\newunicodechar{𝓢}{\ensuremath{\mathcal{S}}}
\newunicodechar{𝓣}{\ensuremath{\mathcal{T}}}
\newunicodechar{𝓤}{\ensuremath{\mathcal{U}}}
\newunicodechar{𝓥}{\ensuremath{\mathcal{V}}}
\newunicodechar{𝓦}{\ensuremath{\mathcal{W}}}
\newunicodechar{𝓧}{\ensuremath{\mathcal{X}}}
\newunicodechar{𝓨}{\ensuremath{\mathcal{Y}}}
\newunicodechar{𝓩}{\ensuremath{\mathcal{Z}}}
\newunicodechar{…}{\ensuremath{\ldots}}
\newunicodechar{∗}{\ensuremath{\ast}}
\newunicodechar{⊢}{\ensuremath{\vdash}}
\newunicodechar{⊧}{\ensuremath{\models}}
\newunicodechar{′}{\ensuremath{'}}
\newunicodechar{″}{\ensuremath{''}}
\newunicodechar{‴}{\ensuremath{'''}}
\newunicodechar{∥}{\ensuremath{\|}}
\newunicodechar{⊕}{\ensuremath{\oplus}}
\newunicodechar{⁺}{\ensuremath{^+}}
\newunicodechar{⊇}{\ensuremath{\supseteq}}
\newunicodechar{∘}{\ensuremath{\circ}}
\newunicodechar{∙}{\ensuremath{\cdot}}
\newunicodechar{⋅}{\ensuremath{\cdot}}
\newunicodechar{≈}{\ensuremath{\approx}}
\newunicodechar{×}{\ensuremath{\times}}
\newunicodechar{∞}{\ensuremath{\infty}}
\newunicodechar{⊑}{\ensuremath{\sqsubseteq}}
 \usepackage{mathpartir}
\usepackage{booktabs}
\usepackage{ifdraft}
\usepackage[shortlabels]{enumitem}

\usepackage{subcaption}
\usepackage{bm}

\newcounter{assumptioncounter}
\newtheorem*{assumption*}{Assumption}
\newcounter{algorithmcounter}
\newtheorem{algorithm}[algorithmcounter]{Algorithm}

\newlength{\savedbelowdisplayskip}

\definecolor{colorblind1}{RGB}{216, 27, 96} \definecolor{colorblind2}{RGB}{30, 136, 229} \definecolor{colorblind3}{RGB}{255, 193, 7} \definecolor{colorblind4}{RGB}{0, 77, 64}

\newtoggle{arxiv}
\toggletrue{arxiv}

\iftoggle{arxiv}
{
\hideLIPIcs  }
{
}

\title{Asynchronous Multiparty Session Type \\ Implementability is Decidable -- \\ Lessons Learned from Message Sequence Charts}

\titlerunning{Asynchronous MST Implementability is Decidable -- Lessons Learned from MSCs}

\author{Felix Stutz}
{MPI-SWS, Kaiserslautern, Germany}
{fstutz@mpi-sws.org}
{https://orcid.org/0000-0003-3638-4096}
{}

\authorrunning{F.\ Stutz}

\Copyright{Felix Stutz} 



\category{} 

\iftoggle{arxiv}
{
}
{
\relatedversion{}
\relatedversiondetails{Extended Version}{https://arxiv.org/abs/2302.11272}
}

\funding{This research was sponsored in part by the Deutsche Forschungsgemeinschaft project 389792660 TRR 248—CPEC.}

\acknowledgements{The author thanks Damien Zufferey, Emanuele D'Osualdo, Ashwani Anand, Rupak Majumdar, and the anonymous reviewers for their valuable feedback.}

\nolinenumbers

\EventEditors{Karim Ali and Guido Salvaneschi}
\EventNoEds{2}
\EventLongTitle{37th European Conference on Object-Oriented Programming (ECOOP 2023)}
\EventShortTitle{ECOOP 2023}
\EventAcronym{ECOOP}
\EventYear{2023}
\EventDate{July 17--21, 2023}
\EventLocation{Seattle, Washington, United States}
\EventLogo{}
\SeriesVolume{263}
\ArticleNo{26}

\begin{document}

\setlength{\savedbelowdisplayskip}{\belowdisplayskip}
\setlength{\belowdisplayskip}{3pt}

\maketitle

\begin{abstract}
Multiparty session types (MSTs) provide efficient means to specify and verify asynchronous message-passing systems. For a global type, which specifies all interactions between roles in a system, the implementability problem asks whether there are local specifications for all roles such that their composition is deadlock-free and generates precisely the specified executions.
Decidability of the implementability problem is an open question.
We answer it positively for global types with sender-driven choice, which allow a sender to send to different receivers upon branching and a receiver to receive from different senders.
To achieve this, we generalise results from the domain of high-level message sequence charts (HMSCs). This connection also allows us to comprehensively investigate how HMSC techniques can be adapted to the MST setting.
This comprises techniques to make the problem algorithmically more tractable as well as a variant of implementability that may open new design space for MSTs.
Inspired by potential performance benefits, we introduce a generalisation of the implementability problem that we, unfortunately, prove to be~undecidable.  \end{abstract}

\section{Introduction}

Distributed message-passing systems are omnipresent and, therefore, designing and implementing them correctly is very important.
However, this is a very difficult task at the same time.
In fact, it is well-known that verifying such systems is algorithmically undecidable in general due to the combination of asynchrony (messages are buffered) and concurrency \cite{DBLP:journals/jacm/BrandZ83}.

Multiparty Session Type (MST) frameworks \cite{DBLP:conf/popl/HondaYC08,DBLP:journals/jacm/HondaYC16} provide efficient means to specify and verify such distributed message-passing systems (e.g., see the survey~\cite{DBLP:journals/ftpl/AnconaBB0CDGGGH16}).
They have also been applied to various other domains like
cyber-physical systems~\cite{DBLP:conf/ecoop/MajumdarPYZ19},
    timed systems~\cite{DBLP:conf/esop/BocchiMVY19},
web services~\cite{DBLP:conf/tgc/YoshidaHNN13}, and
    smart contracts~\cite{DBLP:conf/csfw/DasB0PS21}.
In~MST frameworks, global types are global specifications, which comprise all interactions between roles in a protocol.
From a design perspective, it makes sense to start with such a global protocol specification --- instead of a system with arbitrary communication between roles and a specification to satisfy.

Let us consider a variant of the well-known two buyer protocol from the MST literature,
e.g.,~\cite[Fig.~4(2)]{DBLP:journals/pacmpl/ScalasY19}.
Two buyers $\buyerA$ and~$\buyerB$ purchase a sequence of items from seller $\seller$.
We informally describe the protocol and \emph{emphasise} the interactions.
At the start and after every purchase (attempt), buyer $\buyerA$ can decide whether to buy the next item or whether they are \emph{done}.
For each item, buyer $\buyerA$ \emph{queries} its price and the seller $\seller$ replies with the \emph{price}.
Subsequently, buyer $\buyerA$ decides whether to \emph{cancel} the purchase process for the current item or proposes to \emph{split} to buyer $\buyerB$ that can \emph{accept} or \emph{reject}.
In both cases, buyer $\buyerA$ notifies the seller~$\seller$ if they want to \emph{buy} the item or \emph{not}.
This protocol is specified with the following global type:

{ \vspace{-5ex}
\scriptsize \[
\GG_{\TBPOR}
    \quad \is \quad
\mu t. \,
+
\begin{cases}
\msgFromTo{\buyerA}{\seller}{\query}. \,
\msgFromTo{\seller}{\buyerA}{\price}. \,
+
    \begin{cases}
    \msgFromTo{\buyerA}{\buyerB}{\splitmsg}. \,
        (
        \msgFromTo{\buyerB}{\buyerA}{\yes}. \,
        \msgFromTo{\buyerA}{\seller}{\buy}. \, t
        +
        \msgFromTo{\buyerB}{\buyerA}{\no}. \,
        \msgFromTo{\buyerA}{\seller}{\no}. \, t
        )
    \\
    \msgFromTo{\buyerA}{\buyerB}{\cancel}. \,
    \msgFromTo{\buyerA}{\seller}{\no}. \, t
    \end{cases}
\\
\msgFromTo{\buyerA}{\seller}{\done}. \,
\msgFromTo{\buyerA}{\buyerB}{\done}. \, 0
\end{cases}
\hspace{-5ex}. \] }

\vspace{-1ex}
\noindent The first term $\mu t$ binds the recursion variable $t$ which is used at the end of the first two lines and allows the protocol to recurse back to this point.
Subsequently, $+$ and the curly bracket indicate a choice by buyer $\buyerA$ as it is the sender for the next interaction, e.g.,
$\msgFromTo{\buyerA}{\seller}{\query}$.
For our asynchronous setting, this term jointly specifies the send event
$\snd{\buyerA}{\seller}{\query}$ for buyer~$\buyerA$ and its corresponding receive event
$\rcv{\buyerA}{\seller}{\query}$ for seller~$\seller$, which may happen with arbitrary delay.
The state machine in \cref{fig:fsm-semantics-2BPOR} illustrates its semantics with abbreviated message~labels.

\smallskip\noindent
{\sffamily\bfseries
The Implementability Problem for Global Types and the MST Approach
}

\noindent
A global type provides a global view of the intended protocol.
However, when implementing a protocol in a distributed setting, one needs a local specification for each role.
The \emph{implementability problem} for a global type asks whether there are local specifications for all roles such that, when complying with their local specifications, their composition never gets stuck and exposes the same executions as specified by the global type.
This is a challenging problem because roles can only partially observe the execution of a system:
each role only knows the messages it sent and received and,
in an asynchronous setting, a role does not know when one of its messages will be received by another role.
In contrast, in a synchronous setting, there are no channels, yielding finite state systems.
Still, we could not find a reference that precisely settles the decidability of synchronous implementability.
We sketch a proof in \cref{sec:related}.
In this work, we solely deal with the asynchronous setting.

In general, one distinguishes between a role in a protocol and the process which implements the local specification of a role in a system.
We use the local specifications directly as implementations so
the difference is not essential and we use the term role instead of~process.

\begin{figure}[t]

\begin{subfigure}[b]{0.33\textwidth}
\centering
\resizebox{0.87\textwidth}{!}{
    \begin{tikzpicture}[sem , node distance=0.5cm and 0.7cm]
    \node[recstate, initial above, initial text = ] (q0) {$q_0$};
    \node[state, below = of q0, yshift=1mm] (q1) {$q_1$};
    \node[state, below = of q1, xshift=-10mm, yshift=1mm] (q2) {$q_2$};
    \node[state, below = of q2, yshift=1mm] (q3) {$q_3$};
    \node[state, below = of q3, xshift=-8mm, yshift=1mm] (q4) {$q_4$};
    \node[state, below = of q4, xshift=-4mm, yshift=1mm] (q5) {$q_5$};
    \node[varstate, below = of q5, yshift=1mm] (q6) {$q_6$};
    \node[state, below = of q4, xshift=4mm, yshift=1mm] (q7) {$q_7$};
    \node[varstate, below = of q7, yshift=1mm] (q8) {$q_{8}$};
    \node[state, below = of q3, xshift=8mm, yshift=1mm] (q9) {$q_9$};
    \node[varstate, below = of q9, yshift=1mm] (q10) {$q_{10}$};
    \node[state, below = of q1, xshift=10mm, yshift=1mm] (q11) {$q_{11}$};
    \node[finalstate, below = of q11, yshift=1mm] (q12) {$q_{12}$};
\path (q0) edge node[right, yshift=.5mm] {$\emptystring$} (q1);
    \path (q1) edge node[left, yshift=1mm] {$\msgFromTo{\buyerA}{\seller}{\queryshort}$} (q2);
    \path (q2) edge node[left] {$\msgFromTo{\seller}{\buyerA}{\priceshort}$} (q3);
    \path (q3) edge node[right, yshift=1mm] {$\msgFromTo{\buyerA}{\buyerB}{\cancelshort}$} (q9);
    \path (q9) edge node[right, yshift=.5mm] {$\msgFromTo{\buyerA}{\seller}{\noshort}$} (q10);
    \path (q3) edge node[left, yshift=1mm] {$\msgFromTo{\buyerA}{\buyerB}{\splitshort}$} (q4);
    \path (q4) edge node[left, yshift=1mm] {$\msgFromTo{\buyerB}{\buyerA}{\yesshort}$} (q5);
    \path (q5) edge node[left, yshift=1mm] {$\msgFromTo{\buyerA}{\seller}{\buyshort}$} (q6);
    \path (q4) edge node[right, yshift=1mm] {$\msgFromTo{\buyerB}{\buyerA}{\noshort}$} (q7);
    \path (q7) edge node[right, yshift=.5mm] {$\msgFromTo{\buyerA}{\seller}{\noshort}$} (q8);
    \path (q1) edge node[right, yshift=1mm] {$\msgFromTo{\buyerA}{\seller}{\doneshort}$} (q11);
    \path (q11) edge node[right, yshift=.5mm] {$\msgFromTo{\buyerA}{\buyerB}{\doneshort}$} (q12);

    \draw[semarrow] (q10.south) -- node {} ++(0,-0.15) |- ++(2.3cm,0) |- (q0) node[] {} node[] {};
    \draw[semarrow] (q6.south) -- node {} ++(0,-0.15) |- ++(-1.1cm,0) |- (q0) node[] {} node[] {};
    \draw[semarrow] (q8.south) -- node {} ++(0,-0.15) |- ++(3.7cm,0) |- (q0) node[] {} node[] {};
\end{tikzpicture}
 }
\caption{State machine for \\ \hspace{4ex} semantics of $\GG_{\TBPOR}$}
\label{fig:fsm-semantics-2BPOR}
\end{subfigure}
\hfill
\begin{subfigure}[b]{0.33\textwidth}
\centering
\resizebox{0.87\textwidth}{!}{
    \begin{tikzpicture}[sem , node distance=0.5cm and 0.7cm]
    \node[recstate, initial above, initial text = ] (q0) {$q_0$};
    \node[rcvstate, below = of q0, yshift=1mm] (q1) {$q_1$};
    \node[sndstate, below = of q1, xshift=-10mm, yshift=1mm] (q2) {$q_2$};
    \node[state, below = of q2, yshift=1mm] (q3) {$q_3$};
    \node[state, below = of q3, xshift=-8mm, yshift=1mm] (q4) {$q_4$};
    \node[rcvstate, below = of q4, xshift=-4mm, yshift=1mm] (q5) {$q_5$};
    \node[varstate, below = of q5, yshift=1mm] (q6) {$q_6$};
    \node[rcvstate, below = of q4, xshift=4mm, yshift=1mm] (q7) {$q_7$};
    \node[varstate, below = of q7, yshift=1mm] (q8) {$q_{8}$};
    \node[rcvstate, below = of q3, xshift=8mm, yshift=1mm] (q9) {$q_9$};
    \node[varstate, below = of q9, yshift=1mm] (q10) {$q_{10}$};
    \node[finalstate, below = of q1, xshift=10mm, yshift=1mm] (q11) {$q'_{11}$};
\path (q0) edge node[right, yshift=.5mm] {$\emptystring$} (q1);
    \path (q1) edge node[left, yshift=1mm] {$\rcv{\buyerA}{\seller}{\queryshort}$} (q2);
    \path (q2) edge node[left, yshift=1mm] {$\snd{\seller}{\buyerA}{\priceshort}$} (q3);
    \path (q3) edge node[right, yshift=1mm] {$\emptystring$} (q9);
    \path (q9) edge node[right, yshift=.5mm] {$\rcv{\buyerA}{\seller}{\noshort}$} (q10);
    \path (q3) edge node[left, yshift=1mm] {$\emptystring$} (q4);
    \path (q4) edge node[left, yshift=1mm] {$\emptystring$} (q5);
    \path (q5) edge node[left, yshift=.5mm] {$\rcv{\buyerA}{\seller}{\buyshort}$} (q6);
    \path (q4) edge node[right, yshift=1mm] {$\emptystring$} (q7);
    \path (q7) edge node[right, yshift=.5mm] {$\rcv{\buyerA}{\seller}{\noshort}$} (q8);
    \path (q1) edge node[right, yshift=1mm] {$\rcv{\buyerA}{\seller}{\doneshort}$} (q11);

    \draw[semarrow] (q10.south) -- node {} ++(0,-0.15) |- ++(2.3cm,0) |- (q0) node[] {} node[] {};
    \draw[semarrow] (q6.south) -- node {} ++(0,-0.15) |- ++(-1.1cm,0) |- (q0) node[] {} node[] {};
    \draw[semarrow] (q8.south) -- node {} ++(0,-0.15) |- ++(3.7cm,0) |- (q0) node[] {} node[] {};
\end{tikzpicture}
 }
\caption{Projection of $\GG_{\TBPOR}$ onto $\seller$ \\ without merge}
\label{fig:projection-seller-2BPOR}
\end{subfigure}
\hfill
\begin{subfigure}[b]{0.32\textwidth}
\centering
\def\svgwidth{0.95\textwidth}
\begingroup \makeatletter \providecommand\color[2][]{\errmessage{(Inkscape) Color is used for the text in Inkscape, but the package 'color.sty' is not loaded}\renewcommand\color[2][]{}}\providecommand\transparent[1]{\errmessage{(Inkscape) Transparency is used (non-zero) for the text in Inkscape, but the package 'transparent.sty' is not loaded}\renewcommand\transparent[1]{}}\providecommand\rotatebox[2]{#2}\newcommand*\fsize{\dimexpr\f@size pt\relax}\newcommand*\lineheight[1]{\fontsize{\fsize}{#1\fsize}\selectfont}\ifx\svgwidth\undefined \setlength{\unitlength}{329.10428265bp}\ifx\svgscale\undefined \relax \else \setlength{\unitlength}{\unitlength * \real{\svgscale}}\fi \else \setlength{\unitlength}{\svgwidth}\fi \global\let\svgwidth\undefined \global\let\svgscale\undefined \makeatother \begin{picture}(1,1.01170208)\lineheight{1}\setlength\tabcolsep{0pt}\put(0,0){\includegraphics[width=\unitlength,page=1]{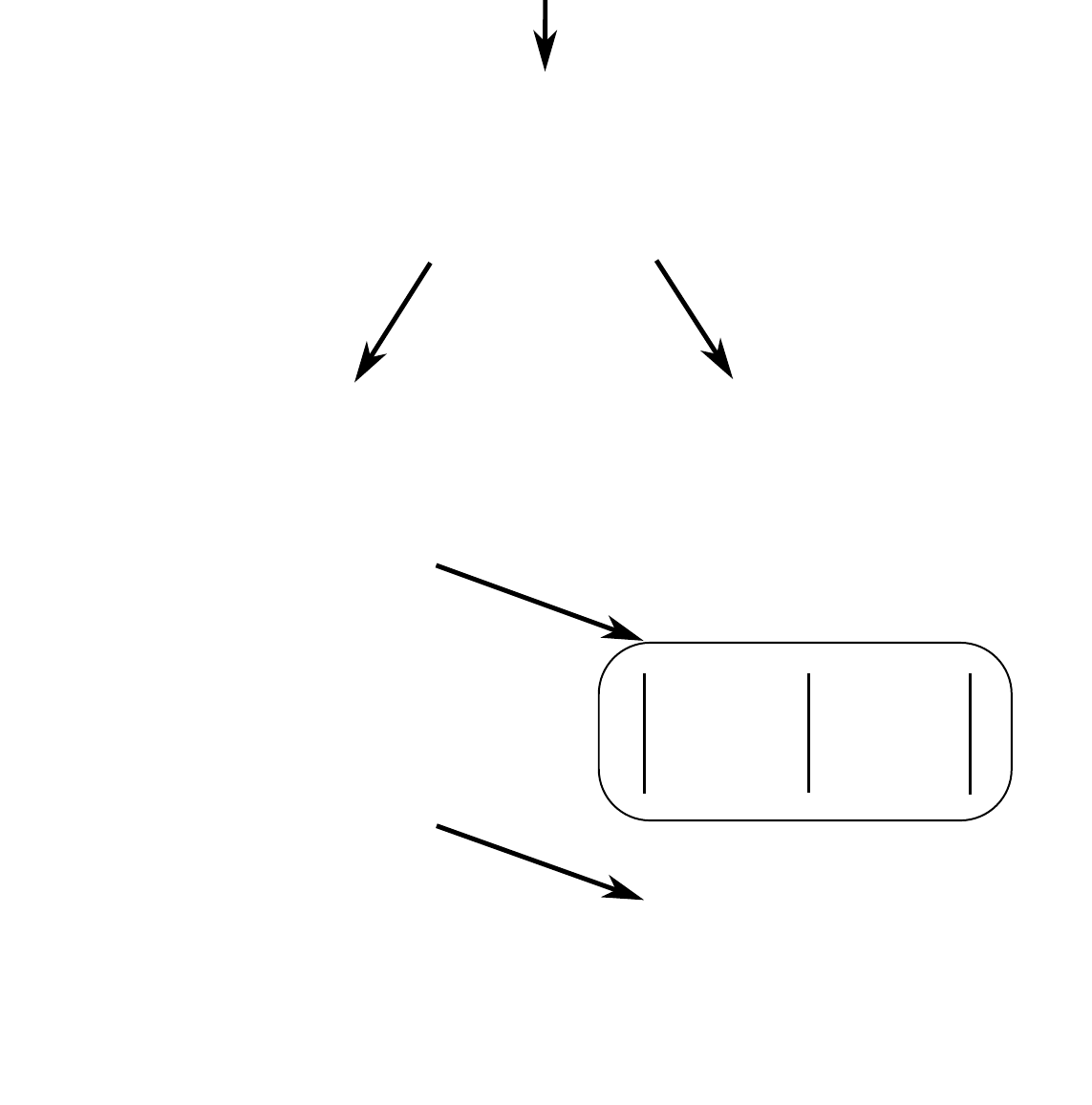}}\put(0.33496768,0.89237395){\color[rgb]{0,0,0}\makebox(0,0)[lt]{\lineheight{1.25}\smash{\begin{tabular}[t]{l}{\scriptsize $\seller$}\end{tabular}}}}\put(0,0){\includegraphics[width=\unitlength,page=2]{figs/hmsc-2BPOR.pdf}}\put(0.74482909,0.36127617){\color[rgb]{0,0,0}\makebox(0,0)[lt]{\lineheight{0.25}\smash{\begin{tabular}[t]{l}{\tiny \cancel}\end{tabular}}}}\put(0.65242462,0.31701227){\color[rgb]{0,0,0}\makebox(0,0)[lt]{\lineheight{1.25}\smash{\begin{tabular}[t]{l}{\tiny \no}\end{tabular}}}}\put(0,0){\includegraphics[width=\unitlength,page=3]{figs/hmsc-2BPOR.pdf}}\put(0.80159254,0.1240988){\color[rgb]{0,0,0}\makebox(0,0)[lt]{\lineheight{1.25}\smash{\begin{tabular}[t]{l}{\tiny \no}\end{tabular}}}}\put(0.65360937,0.07983359){\color[rgb]{0,0,0}\makebox(0,0)[lt]{\lineheight{1.25}\smash{\begin{tabular}[t]{l}{\tiny \no}\end{tabular}}}}\put(0,0){\includegraphics[width=\unitlength,page=4]{figs/hmsc-2BPOR.pdf}}\put(0.74960085,0.5524435){\color[rgb]{0,0,0}\makebox(0,0)[lt]{\lineheight{1.25}\smash{\begin{tabular}[t]{l}{\tiny \done}\end{tabular}}}}\put(0.63599412,0.59732208){\color[rgb]{0,0,0}\makebox(0,0)[lt]{\lineheight{1.25}\smash{\begin{tabular}[t]{l}{\tiny \done}\end{tabular}}}}\put(0,0){\includegraphics[width=\unitlength,page=5]{figs/hmsc-2BPOR.pdf}}\put(0.26891244,0.35984519){\color[rgb]{0,0,0}\makebox(0,0)[lt]{\lineheight{1.25}\smash{\begin{tabular}[t]{l}{\tiny \splitmsg}\end{tabular}}}}\put(0,0){\includegraphics[width=\unitlength,page=6]{figs/hmsc-2BPOR.pdf}}\put(0.31111812,0.12267387){\color[rgb]{0,0,0}\makebox(0,0)[lt]{\lineheight{1.25}\smash{\begin{tabular}[t]{l}{\tiny \yes}\end{tabular}}}}\put(0.16313489,0.07840865){\color[rgb]{0,0,0}\makebox(0,0)[lt]{\lineheight{1.25}\smash{\begin{tabular}[t]{l}{\tiny \buy}\end{tabular}}}}\put(0,0){\includegraphics[width=\unitlength,page=7]{figs/hmsc-2BPOR.pdf}}\put(0.13585148,0.60675369){\color[rgb]{0,0,0}\makebox(0,0)[lt]{\lineheight{1.00}\smash{\begin{tabular}[t]{l}{\tiny \query}\end{tabular}}}}\put(0.10817269,0.55486688){\color[rgb]{0,0,0}\makebox(0,0)[lt]{\lineheight{1.25}\smash{\begin{tabular}[t]{l}{\tiny \price}\end{tabular}}}}\put(0.48325701,0.89281857){\color[rgb]{0,0,0}\makebox(0,0)[lt]{\lineheight{1.25}\smash{\begin{tabular}[t]{l}{\scriptsize $\buyerA$}\end{tabular}}}}\put(0.63079467,0.89245544){\color[rgb]{0,0,0}\makebox(0,0)[lt]{\lineheight{1.25}\smash{\begin{tabular}[t]{l}{\scriptsize $\buyerB$}\end{tabular}}}}\put(0,0){\includegraphics[width=\unitlength,page=8]{figs/hmsc-2BPOR.pdf}}\end{picture}\endgroup  \caption{HMSC $H_{\TBPOR}$ \\ \phantom{sth}}
\label{fig:hmsc-2BPOR}
\end{subfigure}

\vspace{-1.5ex}
\caption{Two Buyer Protocol: the finite state machine for the semantics of $\GG_{\TBPOR}$ on the left, the first step of projection in the middle, and as HMSC on the right;
a transition label $\msgFromTo{\buyerA}{\seller}{\queryshort}$ jointly specifies a send event $\snd{\buyerA}{\seller}{\queryshort}$ for buyer~$\buyerA$ and a receive event $\rcv{\buyerA}{\seller}{\queryshort}$ for seller~$\seller$; styles of states indicate their kind, e.g., recursion states (dashed lines) while final states have double~lines}
\vspace{-3ex}
\end{figure}

\emph{Classical} MST frameworks employ a partial \emph{projection operator} with an in-built \emph{merge operator} to solve the implementability problem.
For each role, the projection operator takes the global type and removes all interactions the role is not involved in.
\Cref{fig:fsm-semantics-2BPOR} illustrates the semantics of $\GG_{\TBPOR}$ while \cref{fig:projection-seller-2BPOR} gives the projection onto seller $\seller$ before the merge operator is applied ---
in both, messages are abbreviated with their first letter.
It is easy to see that this procedure introduces non-determinism, e.g., in $q_3$ and $q_4$, which shall be resolved by the merge operator.
Most merge operators can resolve the non-determinism in \cref{fig:projection-seller-2BPOR}.
A~merge operator checks whether it is safe to merge the states and it might fail so it is a~partial operation.
For instance, every kind of state, indicated by a state's style in \cref{fig:projection-seller-2BPOR}, can only be merged with states of the same kind or states of circular shape. For a role, the result of the projection, if defined, is a local type.
They act as local specifications and their syntax is similar to the one of global~types.

Classical projection operators are a best-effort technique.
This yields good (mostly linear) worst-case complexity but comes at the price of rejecting implementable global types.
Intuitively, classical projection operators consider a limited search space for local types.
They bail out early when encountering difficulties and do not unfold recursion.
In addition, most MST frameworks do effectively not allow a role to send to different receivers or receive from different senders upon branching.
This restriction is called \emph{directed choice} --- in contrast to \emph{sender-driven choice} which is more permissive and allows these patterns.
Among the classical projection operators, the one by Majumdar et al.~\cite{DBLP:conf/concur/MajumdarMSZ21} is the only to handle global types with sender-driven choice but it suffers from the shortcomings of a classical projection approach.
We define different merge operators from the literature and visually explain their supported features by example.
We show that the presented projection/merge operators fail to project implementable variations of the two buyer protocol,
showcasing the sources of incompleteness for the classical approach.
For non-classical approaches, we refer to \cref{sec:related}.

As a best-effort technique, it is natural to focus on efficiency rather than completeness.
The work by Castagna et al.~\cite{DBLP:journals/corr/abs-1203-0780} is a notable exception.
Their notion of completeness~\cite[Def.~4.1]{DBLP:journals/corr/abs-1203-0780} is not as strict as the one considered in this work and only a restricted version of their characterisation is algorithmically checkable.
In general, it is not known whether the implementability problem for global types, with directed or sender-driven choice, is decidable.
We answer this open question positively for global types with sender-driven choice.
To this end, we relate the implementability problem for global types with the safe realisability problem for high-level message sequence charts and generalise results for the latter.

\smallskip\noindent
{\sffamily\bfseries
Lessons Learned from Message Sequence Charts
}

\noindent
The two buyer protocol $\GG_{\TBPOR}$ can also be specified as high-level message sequence chart (HMSC) \cite{DBLP:conf/sdl/MauwR97}, as illustrated in \cref{fig:hmsc-2BPOR}.
Each block is a basic message sequence chart (BMSC) which intuitively corresponds to straight-line code.
In each of those, time flows from top to bottom and each role is represented by a vertical line.
We only give the names in the initial block, which is marked by an incoming arrow at the top.
An arrow between two role lines specifies sending and receiving a message with its label.
The graph structure adds branching, which corresponds to choice in global types, and control flow.
Top branches from the global type are on the left in the HMSC while bottom branches are on the~right.

While research on MSTs and HMSCs has been pursued quite independently, the MST literature frequently uses HMSC-like visualisations for global types, e.g., \cite[Fig.~1]{Carbone2005ATB} and \cite[Figs.~1 and 2]{DBLP:journals/jacm/HondaYC16}.
The first formal connection was recently established by Stutz and Zufferey~\cite{DBLP:journals/corr/abs-2209-10328}.

The HMSC approach to the implementability problem, studied as safe realisability, differs from the MST approach of checking conditions during the projection.
For an HMSC, it is known that there is a candidate implementation~\cite{DBLP:journals/tse/AlurEY03}, which implements the HMSC if it is implementable.
Intuitively, one takes the HMSC and removes all interactions a role is not involved in and determinises the result.
We generalise this result to infinite executions.\footnote{For this, we impose a mild assumption: all protocols can (but do not need to) terminate.}

Hence, algorithms and conditions center around checking implementability of HMSCs.
In general, this problem is undecidable~\cite{DBLP:journals/tcs/Lohrey03}.
For \emph{globally-cooperative} HMSCs~\cite{DBLP:journals/jcss/GenestMSZ06},
Lohrey~\cite{DBLP:journals/tcs/Lohrey03} proved it to be \EXPSPACE-complete.
We show that any implementable global type belongs to this class of HMSCs.\footnote{For this, we also impose the mild assumption that all protocols can (but do not need to) terminate.}These results give rise to the following algorithm to check implementability of a global type.
One can check whether a global type is globally-cooperative (which is equivalent to checking its HMSC encoding).
If it is not globally-cooperative, it cannot be implementable.
If it is globally-cooperative, we apply the algorithm by Lohrey~\cite{DBLP:journals/tcs/Lohrey03} to check whether its HMSC encoding is implementable.
If it is, we use its candidate implementation and know that it generalises to infinite~executions.

While this algorithm shows decidability, the complexity might not be tractable.
Based on our results, we show how more tractable but still permissive approaches to check implementability of HMSCs can be adapted to the MST setting. In addition, we consider \emph{payload implementability}, which allows to add payload to messages of existing interactions and checks agreement when the additional payload is ignored.
We present a sufficient condition for global types that implies payload implementability.
These techniques can be used if the previous algorithms are not tractable or reject a global type.

Furthermore, we introduce a generalisation of the implementability problem.
A network may reorder messages from different senders for the same receiver but the implementability problem still requires the receiver to receive them in the specified order.
Our generalisation allows to consider such reorderings of arrival and can yield performance gains.
In addition, it also renders global types implementable that are not implementable in the standard setting.
Unfortunately, we prove this generalisation to be undecidable in general.

\smallskip\noindent
{\sffamily\bfseries
Contributions and Outline
}

\noindent
We introduce our MST framework in \cref{sec:mst} while \cref{sec:related} covers related work.
In the other sections, we introduce the necessary concepts to establish our main \emph{contributions}:\vspace{-1ex}
 \begin{itemize}
  \item We give a visual explanation of the classical projection operator with different merge operators and exemplify its shortcomings (\cref{sec:from-global-to-local}).
\item We prove decidability of the implementability problem for global types with sender-driven choice (\cref{sec:standard-implementability-decidable}) --- provided that protocols can (but do not need to) terminate.
  \item We comprehensively investigate how MSC techniques can be applied to the MST setting, including algorithmics with better complexity for subclasses as well as an interesting variant of the implementability problem
(\cref{sec:msc-techniques-for-mst}).
\item Lastly, we introduce a new variant of the implementability problem with a more relaxed role message ordering, which is closer to the network ordering, and prove it to be undecidable in general (\cref{intraswap-implementability-undecidable}).
 \end{itemize}

 \section{Multiparty Session Types}
\label{sec:mst}

In this section, we formally introduce our Multiparty Session Type (MST) framework.
We define the syntax of global and local types and their semantics. Subsequently, we recall the implementability problem for global types which asks if there is a deadlock-free communicating state machine that admits the same language (without additional synchronisation).

\smallskip\noindent\textbf{\sffamily Finite and Infinite Words.}
Let $\Sigma$ be an alphabet.
We denote the set of finite words over~$\Sigma$ by $\Sigma^*$ and the set of infinite words by $\Sigma^\omega\negmedspace$.
Their union is denoted by $\Sigma^\infty\negmedspace$.
For two strings $u \in \Sigma^*$ and $v \in \Sigma^\infty\negmedspace$,
we say that $u$ is a \emph{prefix} of $v$ if there is some $w\in\Sigma^\infty$ such that $u \cdot w = v$ and denote this with $u \leq v$ while $\pref(v)$ denotes all prefixes of $v$ and is lifted to languages as expected.
For a language $L \subseteq \Alphabet^{\infty}$, we distinguish between the language of finite words $L_{\fin} \is L \inters \Alphabet^*$ and the language of infinite words $L_{\inf} \is L \inters \Alphabet^\omega$.

\smallskip\noindent\textbf{\sffamily Message Alphabet.}
We fix a finite set of messages $\MsgVals$ and a finite set of roles $\Procs\negthinspace$, ranged over with $\procA$, $\procB$, $\procC$, and $\procD$.
With $Σ_{\mathit{sync}} = \set{ \msgFromTo{\procA}{\procB}{\val} \mid \procA,\procB ∈ \Procs \text{ and } \val ∈ \MsgVals}$, we denote the set of interactions where sending and receiving a message is specified at the same time.
For our asynchronous setting, we also define individual send and receive events:
    $Σ_{\procA} = \set{ \snd{\procA}{\procB}{\val}, \rcv{\procB}{\procA}{\val} \mid \procB \in \Procs\negthinspace,\; \val \in \MsgVals }$
for a role $\procA$.
For both send events $\snd{\procA}{\procB}{\val}$ and receive events $\rcv{\procB}{\procA}{\val}$, the first role is \emph{active}, i.e., the sender in the first event and the receiver in the second one.
The union for all roles yields all (asynchronous) events:
$\Alphabet = \Union_{\procA \in \Procs} \Alphabet_{\procA}$.
For the rest of this work, we fix the set of roles $\Procs\negthinspace$, the messages $\MsgVals\negthinspace$, and both sets $\AlphSync$ and~$\Alphabet$.
We may also use the term $\AlphAsync$ for $\Alphabet$.
We define an operator that splits events from $\AlphSync$,
$
\SyncToAsync(\msgFromTo{\procA}{\procB}{\val})
    \is
\snd{\procA}{\procB}{\val}. \,
\rcv{\procA}{\procB}{\val}
$,
which is lifted to sequences and languages as expected.
Given a word, we might also project it to all letters of a certain shape.
For instance, $w\wproj_{\snd{\procA}{\procB}{\_}}$ is the subsequence of $w$ with all of its send events where $\procA$ sends any message to $\procB$.
If we want to select all messages of $w$, we write $\MsgVals(w)$.

\paragraph*{Global and Local Types -- Syntax}

We give the syntax of global and local types following work by
Majumdar et al.~\cite{DBLP:conf/concur/MajumdarMSZ21}.
In this work, we consider global types as specifications for message-passing concurrency and omit features like delegation.

\begin{definition}[Syntax of global types]
\emph{Global types for MSTs} are defined by the grammar:
\vspace{-1ex}
    \begin{grammar}
     G \is
       0
     | \sum_{i ∈ I} \msgFromTo{\procA}{\procB_i}{\val_i.G_i}
     | μ t. G
     | t
    \end{grammar}
The term $0$ explicitly represents termination.
A term $\msgFromTo{\procA}{\procB_i}{\val_i}$ indicates an interaction where $\procA$ sends message $\val_i$ to $\procB_i$.
In our asynchronous semantics, it is split into a send event $\snd{\procA}{\procB_i}{\val_i}$ and a receive event $\rcv{\procA}{\procB_i}{\val_i}$.
In a choice
$
\sum_{i ∈ I} \msgFromTo{\procA}{\procB_i}{\val_i.G_i}
$,
the sender $\procA$ chooses the branch.
We require choices to be unique, i.e.,
$∀ i,j ∈ I.\, i≠j ⇒ \procB_i \neq \procB_j \lor \val_i ≠ \val_j$.
If~$\card{I} = 1$, which means there is no actual choice, we omit the sum operator.
The operators $\mu t$ and $t$ allow to encode loops.
We require them to be guarded, i.e., there must be at least one interaction between the binding $\mu t$ and the use of the recursion variable~$t$.
Without loss of generality, all occurrences of recursion variables $t$ are bound and~distinct.
\end{definition}

\noindent Our global types admit \emph{sender-driven choice} as
$\procA$ can send to different receivers upon branching: $\sum_{i ∈ I} \msgFromTo{\procA}{\procB_i}{\val_i.G_i}$.
This is also called generalised choice by Majumdar et al.~\cite{DBLP:conf/concur/MajumdarMSZ21}.
In contrast, \emph{directed choice} requires a sender to send to a single receiver, i.e.,
$∀ i,j ∈ I. \, \procB_i = \procB_j$.

\begin{example}[Global types]
The two buyer protocol $\GG_{\TBPOR}$ from the introduction is a global~type.
Instead of $\sum$, we use $+$ with curly brackets.
\end{example}

\begin{definition}[Syntax of local types]
\label{def:local-type}
For a role $\procA$, the \emph{local types} are defined as follows:
\vspace{-1ex}
    \begin{grammar}
     L \is 0
         | \IntCh_{i ∈ I} \snd{}{\procB_i}{\val_i}.L_i
         | \ExtCh_{i ∈ I} \rcv{\procB_{i}}{}{\val_i}.L_i
         | μ t. L
         | t
    \end{grammar}
We call
$\IntCh_{i ∈ I} \snd{}{\procB_i}{\val_i}$
an internal choice while
$\ExtCh_{i ∈ I} \rcv{\procB_{i}}{}{\val_i}$
is an external choice.
For both, we require the choice to be unique, i.e.,
$∀ i,j ∈ I.\, i≠j ⇒ \procB_i \neq \procB_j \lor \val_i ≠ \val_j$.
Similarly to global types,
we may omit $\IntCh$ or $\ExtCh$ if there is no actual choice and
we require recursion to be guarded as well as recursion variables to be bound and distinct.
\end{definition}

\begin{example}[Local type]
For the global type $\GG_{\TBPOR}$, a local type for seller $\seller$ is

\vspace{-3ex}
{ \scriptsize
\[
\mu t. \,
\ExtCh
\begin{cases}
\rcv{\buyerA}{}{query}. \,
\snd{}{\buyerA}{price}. \,
(
    \rcv{\buyerA}{}{buy}. \, t
    \; \ExtCh \; \rcv{\buyerA}{}{no}. \, t
    )
\\
\rcv{\buyerA}{}{done}. \, 0
\end{cases}
.
\]
}
\end{example}

\paragraph*{Implementing in a Distributed Setting}

Global types can be thought of as global protocol specifications.
Thus, a natural question and a main concern in MST theory is whether a global type can be implemented in a distributed setting.
We present communicating state machines, which are built from finite state machines, as the standard implementation model.

\begin{definition}[State machines] A \emph{state machine} $A = (Q, \Delta, \delta, q_{0}, F)$ is a $5$-tuple with
a finite set of states $Q$,
an alphabet $\Delta$,
a transition relation $\delta \subseteq Q \times (\Delta \union \set{\emptystring}) \times Q$,
an initial state $q_{0}\in Q$ from the set of states, and
a set of final states $F$ with $F \subseteq Q$.
If $(q, a, q')\in\delta$, we also write $q \xrightarrow{a} q'\negthinspace$.
A sequence $q_0\xrightarrow{w_0} q_1 \xrightarrow{w_1} \ldots$,
with $q_i~\in~Q$ and $w_i\in \Delta \union \set{\emptystring}$ for $i\geq 0$,
such that $q_0$ is the initial state, and for each $i\geq 0$, it holds that $(q_i, w_i, q_{i+1})\in\delta$, is called a \emph{run} in $A$ with its \emph{trace} $w_0w_1\ldots\in \Delta\negmedspace^\infty\negmedspace$.
A run is \emph{maximal} if it ends in a final state or is infinite.
The \emph{language} $\lang(A)$ of $A$ is the set of traces of all maximal runs.
If $Q$ is finite, we say $A$ is a \emph{finite state machine} (FSM).
\end{definition}

\begin{definition}[Communicating state machines] We call $\CSM{A}$ a \emph{communicating state machine} (CSM)  over $\Procs$ and~$\MsgVals$ if
${A}_\procA$
is a finite state machine
with alphabet~$\Sigma_\procA$ for every $\procA\in\Procs\negthinspace$.
The state machine for $\procA$ is denoted by $(Q_\procA, \Sigma_\procA, \delta_\procA, q_{0, \procA}, F_\procA)$.
Intuitively, a CSM allows a set of state machines, one for each role in $\Procs\negthinspace$,
to communicate by sending and receiving messages.
For this, each pair of roles $\procA, \procB\in \Procs\negthinspace$, $\procA \neq\procB$, is connected by two directed \emph{message channels}.
A transition $q_{\procA} \xrightarrow{\snd{\procA}{\procB}{\val}} q'_{\procA}$ in the state machine of $\procA$ denotes that $\procA$ sends message $\val$ to $\procB$ if $\procA$ is in the state~$q_\procA$ and changes its local state to~$q'_\procA$.
The channel $\channel{\procA}{\procB}$ is appended by message~$\val$.
For receptions, a transition $q_{\procB} \xrightarrow{\rcv{\procA}{\procB}{\val}} q'_{\procB}$ in the state machine of $\procB$
corresponds to $\procB$ retrieving the message $\val$ from the head of the channel when its local state is $q_\procB$ which is updated
to $q'_\procB$.
The run of a CSM always starts with empty channels and each finite state machine is in its respective initial state.
A deadlock of $\CSM{A}$ is a reachable configuration without outgoing transitions such that there is a non-empty channel or some participant is in a non-final local state.
The formalisation of this intuition is standard and can be found in\iftoggle{arxiv}
{
\cref{app:semantics-csm}.
}
{
the technical report \cite{arxiv-version}.
}

\end{definition}

\noindent A global type always specifies send and receive events together.
In a CSM execution, there may be independent events that can occur between a send and its respective receive~event.

\begin{example}[Motivation for indistinguishability relation $\interswap$]
Let us consider the following global type which is a part of the two buyer protocol:
$
    \msgFromTo{\buyerA}{\buyerB}{\cancel}. \,
    \msgFromTo{\buyerA}{\seller}{\no}. \, 0
$.
This is one of its traces:
$
    \snd{\buyerA}{\buyerB}{\cancel}. \,
    \rcv{\buyerA}{\buyerB}{\cancel}. \,
    \snd{\buyerA}{\seller}{\no}. \,
    \rcv{\buyerA}{\seller}{\no}
$.
Because the active roles in
$
    \rcv{\buyerA}{\buyerB}{\cancel}
$
and
$
    \snd{\buyerA}{\seller}{\no}
$
are different and we do not reorder a receive event in front of its respective send event, any CSM that accepts the previous trace also accepts the following trace:
$
    \snd{\buyerA}{\buyerB}{\cancel}. \,
    \snd{\buyerA}{\seller}{\no}. \,
    \rcv{\buyerA}{\buyerB}{\cancel}. \,
    \rcv{\buyerA}{\seller}{\no}
$.
\end{example}

Majumdar et al.~\cite{DBLP:conf/concur/MajumdarMSZ21} introduced the following relation to capture this phenomenon.

\begin{definition}[Indistinguishability relation $\interswap$ \cite{DBLP:conf/concur/MajumdarMSZ21}]
We define a family of \emph{indistinguishability relations}
${\interswap_i} \subseteq \Sigma^* \times \Sigma^*$, for $i\geq 0$.
For $w\in\Sigma^*$, we have $w \interswap_0 w$.
For $i=1$, we~define:
\vspace{-2ex}
\begin{enumerate}[label=\textnormal{\arabic*.}]
\item
If $\procA ≠ \procC$, then
$
 w.\snd{\procA}{\procB}{\val}.\snd{\procC}{\procD}{\val'}.u
 \; \interswap_{1} \;
 w.\snd{\procC}{\procD}{\val'}.\snd{\procA}{\procB}{\val}.u
$.

\item
If $\procB ≠ \procD$, then
$
 w.\rcv{\procA}{\procB}{\val}.\rcv{\procC}{\procD}{\val'}.u
 \; \interswap_{1} \;
 w.\rcv{\procC}{\procD}{\val'}.\rcv{\procA}{\procB}{\val}.u
$.

\item
If $\procA ≠ \procD \land (\procA ≠ \procC ∨ \procB ≠ \procD)
$, then
$
 w.\snd{\procA}{\procB}{\val}.\rcv{\procC}{\procD}{\val'}.u
 \; \interswap_{1} \;
 w.\rcv{\procC}{\procD}{\val'}.\snd{\procA}{\procB}{\val}.u
$.
\item
If $\card{w \wproj_{\snd{\procA}{\procB}{\_}}} >
    \card{w \wproj_{\rcv{\procA}{\procB}{\_}}}$,
then
$
 w.\snd{\procA}{\procB}{\val}.\rcv{\procA}{\procB}{\val'}.u
 \; \interswap_{1} \;
 w.\rcv{\procA}{\procB}{\val'}.\snd{\procA}{\procB}{\val}.u
$.
\end{enumerate}
\vspace{-1ex}
Let $w$, $w'$, and $w''$ be words s.t.~$w \interswap_1 w'$ and $w' \interswap_i w''$ for some~$i$.
Then, $w \interswap_{i+1} w''$.
We define $w \interswap u$ if $w \interswap_n u$ for some $n$.
It is straightforward that $\interswap$ is an equivalence relation.
Define $u \preceq_\interswap v$ if there is $w\in\Sigma^*$ such that $u.w \interswap v$.
Observe that $u \interswap v$ iff
$u \preceq_\interswap v$ and $v \preceq_\interswap u$.
For infinite words $u, v\in\Sigma^\omega$, we define $u \preceq_\interswap^\omega v$
if for each finite prefix $u'\negthinspace$ of $u$, there is a finite prefix~$v'$ of~$v$ such that
$u' \preceq_\interswap v'$.
Define $u \interswap v$ iff $u \preceq_\interswap^\omega v$ and $v\preceq_\interswap^\omega u$.

We lift the equivalence relation $\interswap$ on words to languages:

For a language $L$, we define
{
\small
$
  \interswaplang(L) = \left\{ w' \mid \bigvee
    \begin{array}{l}
    w' \in \Alphabet^* \land ∃ w ∈ \Alphabet^*. \; w \in L \text{ and } w' \interswap w \\
    w' ∈ \Alphabet^ω \land \exists w \in \Alphabet^\omega. \; w \in
    L \text{ and } w' \preceq_\interswap^\omega w
  \end{array} \right\}.
$
}
\end{definition}

This relation characterises what can be achieved in a distributed setting using~CSMs.

\begin{lemma}[L.~21 \cite{DBLP:conf/concur/MajumdarMSZ21}]
\label{lm:csm-closed-interswap}
Let $\CSM{A}$ be a CSM.
Then, $\lang(\CSM{A}) = \interswaplang(\lang(\CSM{A}))$.
\end{lemma}

\paragraph*{Global and Local Types -- Semantics}

Hence, we define the semantics of global types using the indistinguishability relation~$\interswap$.

\begin{definition}[Semantics of global types]
\label{def:language-global-mst}
We construct a state machine $\semglobal(\GG)$ to obtain the semantics of a global type $\GG$.
We index every syntactic subterm of $\GG$ with a unique index to distinguish common syntactic subterms, denoted with $[G, k]$ for syntactic subterm $G$ and index $k$.
Without loss of generality, the index for $\GG$ is $1$\emph{:} $[\GG, 1]$.
For clarity, we do not quantify indices.
We define
$\semglobalsync(\GG) = (Q_{\semglobalsync(\GG)}, \AlphSync, δ_{\semglobalsync(\GG)}, q_{0, \semglobalsync(\GG)}, F_{\semglobalsync(\GG)})$ where\vspace{-1ex}
\begin{itemize}
\item $Q_{\semglobalsync(\GG)}$ is the set of all indexed syntactic subterms $[G, k]$ of $\GG$
\item $δ_{\semglobalsync(\GG)}$ is the smallest set containing
            $(
            [\sum_{i ∈ I} \msgFromTo{\procA}{\procB_i}{\val_i.[G_i, k_i]}, k],
            \msgFromTo{\procA}{\procB_i}{\val_i},
            [G_i, k_i]
            )$ for~each $i ∈ I$,
            and
            $([μ t. [G',k'_2], k'_1], ε, [G', k'_2])$ and $([t, k'_3], ε, [μ t. [G', k'_2], k'_1])$, \item $q_{0, \semglobalsync(\GG)} = [\GG, 1]$, and
$F_{\semglobalsync(\GG)} = \set{[0, k] \mid k \text{ is an index for subterm } 0}$.
\end{itemize}
We consider asynchronous communication so each interaction is split into its send and receive event.
In addition, we consider CSMs as implementation model for global types and, from \cref{lm:csm-closed-interswap}, we know that CSM languages are always closed under the indistinguishability relation $\interswap$.
Thus, we also apply its closure to obtain the semantics of $\GG$\emph{:}
$
 \lang(\GG)
    \is
 \interswaplang(\SyncToAsync(\lang(\semglobalsync(\GG))))
$.
\end{definition}

The closure $\interswaplang(\hole)$ corresponds to similar reordering rules in standard MST developments, e.g.,
\cite[Def.~3.2 and 5.3]{DBLP:journals/jacm/HondaYC16}.

\begin{example}
\cref{fig:fsm-semantics-2BPOR} (p.\pageref{fig:fsm-semantics-2BPOR}) illustrates the FSM $\semglobalsync(\GG_{\TBPOR})$.
In the following global type, $\procA$ sends a list of book titles to~$\procB$:
{ $
\mu t. \,
(
    \msgFromTo{\procA}{\procB}{title}. \, t \,
    +
    \msgFromTo{\procA}{\procB}{done}. \, 0
).
$ }
Its semantics is the union of two cases: if the list of book titles is infinite, i.e.,
$
\interswaplang(
 (\snd{\procA}{\procB}{title}. \,
  \rcv{\procA}{\procB}{title} )^\omega
)
$;
and the one if the list is finite, i.e.,
$
\interswaplang(
 (\snd{\procA}{\procB}{title}. \,
  \rcv{\procA}{\procB}{title} )^* . \,
  \snd{\procA}{\procB}{done}. \,
  \rcv{\procA}{\procB}{done}
)
$.
Here, there are only two roles so $\interswaplang(\hole)$ can solely delay receive events (Rule 4 of $\interswap$).
\end{example}

We distinguish states depending on which subterm they correspond to:
\emph{binder states} with their dashed line correspond to a recursion variable binder,
while \emph{recursion states} with their dash-dotted lines indicate the use of a recursion variable.
We omit $\emptystring$ for transitions from recursion to binder states.

\subparagraph*{Local Types.} For the semantics of local types, we analogously construct a state machine $\semlocal(\hole)$.
In constrast, we omit the closure $\interswaplang(\hole)$ because languages of roles are closed under~$\interswap$ (cf.~\cite[Lm.\ 22]{DBLP:conf/concur/MajumdarMSZ21}).
For the full definition, we refer to\iftoggle{arxiv}
{
\cref{app:semantics-local-types}.
}
{
the technical report \cite{arxiv-version}.
}
Compared to global types, we distinguish two more kinds of states for local types:
a~\emph{send state} (internal choice) has a diamond shape while a \emph{receive state} (external choice) has a rectangular shape.
For states with $\emptystring$ as next action, we keep the circular shape and call them \emph{neutral states}.
Because of the $\emptystring$-transitions, \cref{fig:projection-seller-2BPOR}~(p.\pageref{fig:projection-seller-2BPOR}) does not represent the state machine for any local type but illustrates the use of different styles for different kinds of~states.

\paragraph*{The Implementability Problem for Global Types}

The implementability problem for global types asks whether a global type can be implemented in a distributed setting.
The projection operator takes the intermediate representation of local types as local specifications for roles.
We define implementability directly on the implementation model of CSMs.
Intuitively, every collection of local types constitutes a CSM through their semantics.

\begin{definition}[Implementability \cite{DBLP:conf/concur/MajumdarMSZ21}]
A global type $\GG$ is said to be \emph{implementable} if there exists a \emph{deadlock-free} CSM $\CSM{A}$ such that
 their languages are the same \emph{(protocol fidelity)}, i.e.,
$\lang(\GG) = \lang(\CSM{A})$.
We say that $\CSM{A}$ implements $\GG$.
\end{definition}
 \section{Projection -- From Global to Local Types}
\label{sec:from-global-to-local}

In this section, we define and visually explain a typical approach to the implementability problem:
the \emph{classical projection operator}.
It tries to translate global types to local types and, while doing so, checks if this is safe.
Behind the scenes, these checks are conducted by a partial merge operator.
We consider different variants of the merge operator from the literature and exemplify the features they support.
We provide visual explanations of the classical projection operator with these merge operators on the state machines of global types by example.
\iftoggle{arxiv}
{
In \cref{sec:additional-explanation},
}
{
In the technical report~\cite{arxiv-version},
}
we give general descriptions but they are not essential for our observations.
Lastly, we summarise the shortcomings of the full merge operator and exemplify them with variants of the two buyer protocol from the introduction.

\paragraph*{Classical Projection Operator with Parametric Merge}

\begin{definition}[Projection operator]
For a merge operator $\merge$, the \emph{projection} of a global type $\GG$ onto a role $\procC \in \Procs$ is a local type that is defined as follows\emph{:}\footnote{The case split for the recursion binder changes slightly across different definitions. We chose a simple but also the least restrictive condition. We simply check whether the recursion is vacuous, i.e.\ $\mu t. t$, and omit it in this case. We also require to omit $\mu t$ if $t$ is never used in the result.}
\vspace{.5ex}
\begin{footnotesize}
$
 0 \tproj^\merge_\procC
 \is
 0
 \hfill
 t \tproj^\merge_\procC
 \is
 t
$
\end{footnotesize}
\\
\begin{footnotesize}
$
 \left( \sum_{i ∈ I} \msgFromTo{\procA}{\procB_i}{\val_i.G_i} \right)
 \tproj^\merge_\procC \is
 \begin{cases}
        \IntCh_{i ∈ I} \snd{}{\procB_i}{\val_i}.(G_i \tproj^\merge_\procC)
            \hfill \text{ if } \procC = \procA
        \\[1mm]
            \ExtCh_{i \in I}
            \rcv{\procA}{}{\val_i}.(G_i \tproj^\merge_\procC)
            \hfill \text{ if } \procC = \procB
        \\
        \Merge_{i \in I}
            G_i \tproj^\merge_\procC
            \hfill \quad \text{otherwise}
 \end{cases}
 \hfill
 \left( μ t. G \right)
 \tproj^\merge_\procC \is
 \begin{cases}
    μ t. (G \tproj^\merge_\procC)
        & \text{if } G \tproj^\merge_\procC \neq t \\
    0
        & \text{otherwise}
 \end{cases}
\hfill
 $
\end{footnotesize}
\end{definition}
Intuitively, a projection operator takes the state machine $\semglobalsync(\GG)$ for a global type~$\GG$ and projects each transition label to the respective alphabet of the role, e.g.,
$\msgFromTo{\procA}{\procB}{\val}$
becomes
$\rcv{\procA}{\procB}{\val}$ for role $\procB$.
This can introduce non-determinism that ought to be resolved by a partial merge operator.
Several merge operators have been proposed in the literature.

\begin{definition}[Merge operators]
Let $L₁$ and $L₂$ be local types for a role $\procC$, and $\merge$ be a~merge operator.
We define different cases for the result of $L_1 \merge L_2$\emph{:}
\vspace{1ex}

\begin{minipage}{0.22\textwidth}
    \begin{enumerate}[series=mergecases,leftmargin=!,labelindent=0pt,label=\textnormal{(\arabic*)}]
    \item \label{merge:plainonly}
            \begin{small}
$L_1 \;$ if $L_1 = L_2$
            \end{small}
    \end{enumerate}
\end{minipage}
\begin{minipage}{0.73\textwidth}
    \begin{enumerate}[resume=mergecases,label=\textnormal{(\arabic*)}]
    \item \label{merge:semifullonly}
        \begin{small}
        \hspace{-1ex}
$\left(
        \begin{array}{lr}
                   \ExtCh_{i \in I \setminus J} \rcv{\procB}{}{m_i}.\mathit{L}_{1,i}' & \ExtCh \\
                   \ExtCh_{i \in I∩J} \rcv{\procB}{}{m_i}.(\mathit{L}_{1,i}' \merge \mathit{L}_{2,i}') & \ExtCh \\
                   \ExtCh_{i \in J \setminus I} \rcv{\procB}{}{m_i}.\mathit{L}_{2,i}'
        \end{array}
        \right)$
                $\;$
            if $\begin{cases}
                    L₁ = \ExtCh_{i ∈ I} \rcv{\procB}{}{\val_i.\mathit{L}_{1,i}'} \text{, } \\
                    L₂ = \ExtCh_{i ∈ J} \rcv{\procB}{}{\val_i.\mathit{L}_{2,i}'}
                \end{cases}$
        \end{small}
    \end{enumerate}
\end{minipage}

\vspace{1ex}
\begin{minipage}{0.72\textwidth}
\begin{enumerate}[resume=mergecases,leftmargin=!,labelindent=0pt,label=\textnormal{(\arabic*)}]
\item \label{merge:fullonly}
    \begin{small}
$μt₁.(\mathit{L}_1' \merge \mathit{L}_2'[t₂/t₁])$
            $\,$
            if $L₁ =  μt₁.\mathit{L}_1' \text{ and } L₂ = μt₂.\mathit{L}_2'$
    \end{small}
\end{enumerate}
\end{minipage}

\vspace{1ex}
\noindent Each merge operator is defined by a collection of cases it can apply.
If none of the respective cases applies, the result of the merge is undefined.
The plain merge $\plainmerge$ \emph{\cite{DBLP:conf/sfm/CoppoDPY15}} can only apply Case \ref{merge:plainonly}.
The~semi-full merge $\semifullmerge$ \emph{\cite{DBLP:conf/icdcit/YoshidaG20}} can apply Cases \ref{merge:plainonly} and \ref{merge:semifullonly}.
The full merge $\fullmerge$ \emph{\cite{DBLP:journals/pacmpl/ScalasY19}} can apply all Cases \ref{merge:plainonly}, \ref{merge:semifullonly}, and \ref{merge:fullonly}.
\end{definition}

We will also consider the availability merge operator $\availmerge$ by Majumdar et al.~\cite{DBLP:conf/concur/MajumdarMSZ21} which builds on the full merge operator but generalises Case \ref{merge:semifullonly} to allow sender-driven choice.
We will explain the main differences in \cref{rem:sender-driven-choice}.

\begin{remark}[Correctness of projection]
This would be the correctness criterion for projection:
Let $\GG$ be some global type and let plain merge $\plainmerge$, semi full merge $\semifullmerge$, full merge $\fullmerge$, or availability merge $\availmerge$ be the merge operator $\merge$.
If $\GG \tproj^\merge_\procA$ is defined for each role $\procA$,
then the CSM $\CSMl{\semlocal(\GG \tproj^\merge_\procA)}$ implements $\GG$. \\
We do not actually prove this so we do not state it as lemma.
\emph{But why does this hold?} \newline
The implementability condition is the combination of deadlock freedom and protocol fidelity.
Coppo et al.~\cite{DBLP:conf/sfm/CoppoDPY15} show that \emph{subject reduction} entails protocol fidelity and progress while progress, in turn, entails deadlock freedom.
Subject reduction has been proven for the plain merge operator \cite[Thm.\ 1]{DBLP:conf/sfm/CoppoDPY15} and the
semi-full operator \cite[Thm.~1]{DBLP:conf/icdcit/YoshidaG20}.
Scalas and Yoshida pointed out that several versions of classical projection with the full merge are flawed \cite[Sec.\ 8.1]{DBLP:journals/pacmpl/ScalasY19}.
Hence, we have chosen a full merge operator whose correctness follows from the correctness of the more general availability merge operator.
For the latter, correctness was proven by Majumdar et al.~\cite[Thm.\ 16]{DBLP:conf/concur/MajumdarMSZ21}.
\end{remark}

\begin{example}[Projection without merge / Collapsing erasure]
In the introduction, we considered $\GG_{\TBPOR}$ and the FSM for its semantics in \cref{fig:fsm-semantics-2BPOR}.
We projected (without merge) onto seller $\seller$ to obtain the FSM in \cref{fig:projection-seller-2BPOR}.
In general, we also collapse neutral states with a single $\emptystring$-transition and their only successor.
We call this \emph{collapsing erasure}.
We only need to actually collapse states for the protocol in
\cref{fig:neg-ex-semi-full-merge}.
In all other illustrations, we indicate the interactions the role is not involved with the following notation:
$\erased{\msgFromTo{\procA}{\procB}{l}}$.

\end{example}

\paragraph*{On the Structure of $\pmb{\semglobalsync(}\GG\pmb{)}$}

We now show that the state machine for every local and global type has a certain shape.
This simplifies the visual explanations of the different merge operators.
Intuitively, every such state machine has a tree-like structure where backward transitions only happen at leaves of the tree, are always labelled with $\emptystring$, and only lead to ancestors.
The FSM in~\cref{fig:fsm-semantics-2BPOR} (p.\pageref{fig:fsm-semantics-2BPOR}) illustrates this shape where the root of the tree is at the top.

\begin{definition}[Ancestor-recursive, non-merging, intermediate recursion, etc.]
Let $A = (Q, \Delta, \delta, q_{0}, F)$ be a finite state machine.
We say that $A$ is \emph{ancestor-recursive} if there is a function $\levelfunc \from Q \to \Nat$ such that, for every transition $q \xrightarrow{x} q' \in \delta$, one of the two holds:
\vspace{-1.5ex}
\begin{enumerate}[labelindent=0pt,labelwidth=\widthof{(b)},label=\textnormal{(\alph*)},itemindent=0em,leftmargin=!]
 \item $\levelfunc(q) > \levelfunc(q')$, or
 \item $x = \emptystring$ and there is a run from the initial state $q_0$ (without going through $q$) to $q'$ which can be completed to reach $q$:
$q_0\xrightarrow{\hole} \ldots \xrightarrow{\hole}q_n$ is a run with $q_n = q'$ and $q \neq q_i$ for every $0 \leq i \leq n$, and
 the run can be extended to
 $q_0\xrightarrow{\hole} \ldots \xrightarrow{\hole}q_n \xrightarrow{\hole} \ldots \xrightarrow{\hole}q_{n+m}$ with $q_{n+m} = q$.
Then, the state $q'$ is called \emph{ancestor} of $q$.
\end{enumerate}
\vspace{-1.5ex}
We call the first \textnormal{(a)} kind of transition \emph{forward transition} while the second \textnormal{(b)} kind is a \emph{backward transition}.
The state machine $A$ is said to be free from \emph{intermediate recursion} if every state $q$ with more than one outgoing transition, i.e.,
$
\card{\set{q' \mid q \xrightarrow{\hole} q' \in \delta}} > 1
$,
has only forward transitions.
We say that $A$ is \emph{non-merging} if every state only has one incoming edge with greater level, i.e., for every state $q'$, $\set{q \mid q \xrightarrow{\hole} q' \in \delta \land \levelfunc(q) > \levelfunc(q')} \leq 1$.
The state machine $A$ is \emph{dense} if,
for every $q \xrightarrow{x} q' \in \delta$,
the transition label $x$ is $\emptystring$ implies that $q$ has only one outgoing transition.
Last, the \emph{cone} of $q$ are all states~$q'$ which are reachable from $q$ and have a smaller level than $q$, i.e., $\levelfunc(q) > \levelfunc(q')$.
\end{definition}

\begin{proposition}[Shape of $\semglobalsync(\GG)$ and $\semlocal(L)$]
\label{prop:shape-of-sem-fsms}
Let $\GG$ be some global type and $L$ be some local type.
Then, both $\semglobalsync(\GG)$ and $\semlocal(L)$ are
ancestor-recursive, free from intermediate recursion, non-merging, and dense.
\end{proposition}

For both, the only forward $\emptystring$-transitions occur precisely from binder states while backward transitions happen from variable states to binder states. The illustrations for our examples always have the initial state, which is the state with the greatest level, at the top.
This is why we use greater and higher as well as smaller and lower interchangeably for levels.

\paragraph*{Features of Different Merge Operators by Example}

In this section, we exemplify which features each of the merge operators supports.
We present a sequence of implementable global types.
Despite, some cannot be handled by some (or all) merge operators.
If a global type is not projectable using some merge operator, we say it is \emph{rejected} and it constitutes a \emph{negative} example for this merge operator.
We focus on role $\procC$ when projecting.
Thus, rejected mostly means that there is (at least) no projection onto $\procC$.
If a global type is projectable by some merge operator, we call it a \emph{positive} example.
All examples strive for minimality and follow the idea that roles decide whether to take a left~($l$) or right~($r$) branch of a~choice.

\begin{figure}[t]

\begin{subfigure}[b]{0.38\textwidth}
\centering
\resizebox{0.92\textwidth}{!}{
    \begin{tikzpicture}[sem , node distance=0.5cm and 0.7cm]
    \node[recstate, initial above, initial text = ] (q0) {$q_0$};
    \node[state, below = of q0] (q0p) {$q'_0$};
    \node[rcvstate, below = of q0p, xshift=-14mm,yshift=2mm] (q1) {$q_1$};
    \node[finalstate, below = of q1, xshift=-4mm] (q2) {$q_2$};
    \node[varstate, below = of q1, xshift=4mm] (q3) {$q_3$};
    \node[rcvstate, below = of q0p, xshift=14mm,yshift=2mm] (q4) {$q_4$};
    \node[finalstate, below = of q4, xshift=-4mm] (q5) {$q_5$};
    \node[varstate, below = of q4, xshift=4mm] (q6) {$q_6$};
\path (q0) edge node[right] {$\emptystring$} (q0p);
    \path (q0p) edge node[left,yshift=1.5mm,xshift=1mm] {$\erased{\msgFromTo{\procA}{\procB}{l}}$} (q1);
    \path (q1) edge node[left] {$\msgFromTo{\procB}{\procC}{l}$} (q2);
    \path (q1) edge node[right] {$\msgFromTo{\procB}{\procC}{r}$} (q3);
    \path (q0p) edge node[right,yshift=1.5mm] {$\erased{\msgFromTo{\procA}{\procB}{r}}$} (q4);
    \path (q4) edge node[left] {$\msgFromTo{\procB}{\procC}{l}$} (q5);
    \path (q4) edge node[right] {$\msgFromTo{\procB}{\procC}{r}$} (q6);

    \draw[semarrow] (q3.south) -- node {} ++(0,-0.15) |- ++(-1.65cm,0) |- (q0) node[] {} node[] {};
    \draw[semarrow] (q6.south) -- node {} ++(0,-0.15) |- ++(0.9cm,0) |- (q0) node[] {} node[] {};
\end{tikzpicture}
 }
\caption{Positive example \\ for plain merge}
\label{fig:pos-ex-plain-merge}
\label{fig:pos-ex-plain-merge-projected}
\end{subfigure}
\hfill
\begin{subfigure}[b]{0.20\textwidth}
\centering
\resizebox{0.97\textwidth}{!}{
\begin{tikzpicture}[sem , node distance=0.7cm and 0.7cm]
    \node[recstate, initial above, initial text = ] (q0) {$q_0$};
\node[rcvstate, below = of q0] (q1) {$q_1$};
    \node[finalstate, below = of q1, xshift=-4mm] (q2) {$q_2$};
    \node[varstate, below = of q1, xshift=4mm] (q3) {$q_3$};
\path (q0) edge node[right] {$\emptystring$} (q1);
    \path (q1) edge node[left] {$\rcv{\procB}{\procC}{l}$} (q2);
    \path (q1) edge node[right] {$\rcv{\procB}{\procC}{r}$} (q3);

    \draw[semarrow] (q3.south) -- node {} ++(0,-0.15) |- ++(0.9cm,0) |- (q0) node[] {} node[] {};
\end{tikzpicture}
 }
\caption{After plain merge}
\label{fig:pos-ex-plain-merge-merged}
\end{subfigure}
\hfill
\begin{subfigure}[b]{0.38\textwidth}
\centering
\resizebox{0.97\textwidth}{!}{
\begin{tikzpicture}[sem , node distance=0.6cm and 0.7cm]
    \node[state, initial above, initial text = ] (q0) {$q_0$};
    \node[rcvstate, below = of q0, xshift=-14mm] (q1) {$q_1$};
    \node[finalstate, accepting, below = of q1, xshift=-4mm] (q2) {$q_2$};
\node[rcvstate, below = of q0, xshift=14mm] (q4) {$q_4$};
\node[finalstate, below = of q4, xshift=4mm] (q6) {$q_6$};
\path (q0) edge node[left] {$\erased{\msgFromTo{\procA}{\procB}{l}}$} (q1);
    \path (q1) edge node[left] {$\msgFromTo{\procB}{\procC}{l}$} (q2);
\path (q0) edge node[right] {$\erased{\msgFromTo{\procA}{\procB}{r}}$} (q4);
\path (q4) edge node[right] {$\msgFromTo{\procB}{\procC}{r}$} (q6);
\end{tikzpicture}
 }
\caption{Negative example \\ for plain merge}
\label{fig:neg-ex-plain-merge}
\end{subfigure}

\vspace{-1ex}
\caption{
The FSM on the left represents an implementable global type that is accepted by plain merge.
It implicitly shows the FSM after collapsing erasure:
every interaction $\procC$ is not involved in is given as
$\erased{\msgFromTo{\procA}{\procB}{l}}$.
The FSM in the middle is the result of the plain merge.
The FSM on the right represents an implementable global type that is rejected by plain merge.
It is obtained from the left one by removing one choice option in each branch of the initial choice.
}
\vspace{-3.0ex}
\label{fig:pos-ex-plain-merge-steps}
\end{figure}

\begin{example}[Positive example for plain merge]
The following global type is implementable:
\vspace{-1.5ex}
{ \scriptsize \[
\mu t. +
\begin{cases}
\msgFromTo{\procA}{\procB}{l}. \,
    (
    \msgFromTo{\procB}{\procC}{l}. \, 0
    +
    \msgFromTo{\procB}{\procC}{r}. \, t
    )
\\
\msgFromTo{\procA}{\procB}{r}. \,
    (
    \msgFromTo{\procB}{\procC}{l}. \, 0
    +
    \msgFromTo{\procB}{\procC}{r}. \, t
    )
\end{cases}
.
\] }
\hspace{-1.5ex} The state machine for its semantics is given in \cref{fig:pos-ex-plain-merge}.
After collapsing erasure, there is a non-deterministic choice from $q'_0$ leading to $q_1$ and $q_4$ since $\procC$ is not involved in the initial choice.
The plain merge operator can resolve this non-determinism since both cones of $q_1$ and $q_4$ represent the same subterm.
Technically, there is an isomorphism between the states in both cones which preserves the kind of states as well as the transition labels and the backward transitions from isomorphic recursion states lead to the same binder state.
The result is illustrated in~\cref{fig:pos-ex-plain-merge-merged}.
It is also the FSM of a local type for $\procC$ which is the result of the (syntactic) plain merge:
{ \small $
\mu t. \,
    (
    \rcv{\procB}{}{l}. \, 0
    \, \ExtCh \,
    \rcv{\procB}{}{r}. \, t
    )
$ }.
\end{example}

Our explanation on FSMs allows to check congruence of cones to merge while the definition requires syntactic equality.
If we swap the order of branches
 $\msgFromTo{\procB}{\procC}{l}$
 and
 $\msgFromTo{\procB}{\procC}{r}$
in \cref{fig:pos-ex-plain-merge}
on the right, the syntactic merge rejects.
Still, because both are semantically the same protocol specification, we expect tools to check for such easy fixes.

\begin{example}[Negative example for plain merge]
We consider the following simple implementable global type where the choice by $\procA$ is propagated by $\procB$ to $\procC$:
{ \scriptsize $
+
\begin{cases}
\msgFromTo{\procA}{\procB}{l}. \,
\msgFromTo{\procB}{\procC}{l}. \, 0
\\
\msgFromTo{\procA}{\procB}{r}. \,
\msgFromTo{\procB}{\procC}{r}. \, 0
\end{cases}
$. \\ }
The corresponding state machine is illustrated in
\cref{fig:neg-ex-plain-merge}.
Here, $q_0$ exhibits non-determinism but the plain merge fails because $q_1$ and $q_4$ have different outgoing transition labels. \end{example}

Intuitively, the plain merge operator forbids that any, but the two roles involved in a choice, can have different behaviour after the choice.
It basically forbids propagating a choice.
The semi-full merge overcomes this shortcoming and can handle the previous example.
We present a slightly more complex one to showcase the features it supports.

\begin{figure}[t]

\begin{subfigure}[b]{0.38\textwidth}
\centering
\resizebox{0.92\textwidth}{!}{
    \begin{tikzpicture}[sem , node distance=0.5cm and 0.7cm]
    \node[recstate, initial above, initial text = ] (q0) {$q_0$};
    \node[state, below = of q0] (q0p) {$q'_0$};
    \node[rcvstate, below = of q0p, xshift=-14mm, yshift=2mm] (q1) {$q_1$};
    \node[state, accepting, below = of q1, xshift=-4mm] (q2) {$q_2$};
    \node[state, accepting, below = of q1, xshift=4mm] (q3) {$q_3$};
    \node[rcvstate, below = of q0p, xshift=14mm, yshift=2mm] (q4) {$q_4$};
    \node[state, accepting, below = of q4, xshift=-4mm] (q5) {$q_5$};
    \node[varstate, below = of q4, xshift=4mm] (q6) {$q_6$};
\path (q0) edge node[left] {$\emptystring$} (q0p);
    \path (q0p) edge node[left, yshift=1.5mm, xshift=1mm] {$\erased{\msgFromTo{\procA}{\procB}{l}}$} (q1);
    \path (q1) edge node[left] {$\msgFromTo{\procB}{\procC}{l}$} (q2);
    \path (q1) edge node[right] {$\msgFromTo{\procB}{\procC}{m}$} (q3);
    \path (q0p) edge node[right, yshift=1.5mm] {$\erased{\msgFromTo{\procA}{\procB}{r}}$} (q4);
    \path (q4) edge node[left] {$\msgFromTo{\procB}{\procC}{m}$} (q5);
    \path (q4) edge node[right] {$\msgFromTo{\procB}{\procC}{r}$} (q6);

    \draw[semarrow] (q6.south) -- node {} ++(0,-0.15) |- ++(0.9cm,0) |- (q0) node[] {} node[] {};
\end{tikzpicture}
 }
\caption{Positive example \\for semi-full merge}
\label{fig:pos-ex-semi-full-merge}
\label{fig:pos-ex-semi-full-merge-projected}
\end{subfigure}
\hfill
\begin{subfigure}[b]{0.2\textwidth}
\centering
\resizebox{0.92\textwidth}{!}{
    \begin{tikzpicture}[sem , node distance=0.7cm and 0.7cm]
    \node[recstate, initial above, initial text = ] (q0) {$q_0$};
    \node[rcvstate, below = of q0] (q0p) {$q_{1\mid4}$};
\node[finalstate, below = of q0p, xshift=-10mm] (q2) {$q_2$};
    \node[finalstate, below = of q0p, yshift=-10mm] (q3) {$q_{3\mid5}$};
\node[varstate, below = of q0p, xshift=10mm] (q6) {$q_6$};
\path (q0) edge node[left] {$\emptystring$} (q0p);
\path (q0p) edge node[left] {$\rcv{\procB}{\procC}{l}$} (q2);
    \path (q0p) edge node[sloped] {$\msgFromTo{\procB}{\procC}{m}$} (q3);
\path (q0p) edge node[right] {$\rcv{\procB}{\procC}{r}$} (q6);

    \draw[semarrow] (q6.south) -- node {} ++(0,-0.15) |- ++(0.6cm,0) |- (q0) node[] {} node[] {};
\end{tikzpicture}
 }
\caption{After semi-full merge}
\label{fig:pos-ex-semi-full-merge-merged}
\end{subfigure}
\hfill
\begin{subfigure}[b]{0.38\textwidth}
\centering
\resizebox{0.97\textwidth}{!}{
    \begin{tikzpicture}[sem , node distance=0.7cm and 0.7cm]
    \node[state, initial above, initial text = ] (q0) {$q_0$};
    \node[rcvstate, below = of q0, xshift=-14mm] (q1) {$q_1$};
    \node[finalstate, below = of q1, xshift=-4mm] (q2) {$q_2$};
\node[rcvstate, below = of q0, xshift=14mm] (q4) {$q_4$};
\node[finalstate, below = of q4, xshift=4mm] (q6) {$q_6$};
\path (q0) edge node[left] {$\erased{\msgFromTo{\procA}{\procB}{l}}$} (q1);
    \path (q1) edge node[left] {$\msgFromTo{\procA}{\procC}{l}$} (q2);
\path (q0) edge node[right] {$\erased{\msgFromTo{\procA}{\procB}{r}}$} (q4);
\path (q4) edge node[right] {$\msgFromTo{\procB}{\procC}{r}$} (q6);
\end{tikzpicture}
 }
\caption{Negative example \\ for full merge}
\label{fig:neg-ex-full-merge}
\end{subfigure}

\vspace{-1ex}
\caption{
The FSM on the left represents an implementable global type (and implicitly the collapsing erasure onto $\procC$) that is accepted by semi-full merge.
The FSM in the middle is the result of the semi-full merge.
The FSM on the right is a negative example for the full merge operator.
}
\vspace{-3.0ex}
\end{figure}

\begin{example}[Positive example for semi-full merge]
Let us consider this implementable global type:{ \scriptsize $
\mu t. +
\begin{cases}
\msgFromTo{\procA}{\procB}{l}. \,
    (
    \msgFromTo{\procB}{\procC}{l}. \, 0
    +
    \msgFromTo{\procB}{\procC}{m}. \, 0
    )
\\
\msgFromTo{\procA}{\procB}{r}. \,
    (
    \msgFromTo{\procB}{\procC}{m}. \, 0
    +
    \msgFromTo{\procB}{\procC}{r}. \, t
    )
\end{cases}
$} \hspace{-2.5ex},
illustrated in \cref{fig:pos-ex-semi-full-merge}.
After applying collapsing erasure, there is a non-deterministic choice from $q_0$ leading to $q_1$ and $q_4$ since $\procC$ is not involved in the initial choice,
We apply the semi-full merge for both states.
Both are receive states so Case \ref{merge:semifullonly} applies.
First, we observe that
$\rcv{\procB}{\procC}{l}$
and
$\rcv{\procB}{\procC}{r}$
are unique to one of the two states so both transitions, with the cones of the states they lead to, can be kept.
Second,
there is $\rcv{\procB}{\procC}{m}$ which is common to both states.
We recursively apply the semi-full merge and, with Case \ref{merge:plainonly}, observe that the result $q_{3\mid5}$ is simply a final state.
Overall, we obtain the state machine in
\cref{fig:pos-ex-semi-full-merge-merged}, which is equivalent to the result of the syntactic projection with semi-full merge:
{ \scriptsize $
\mu t. \,
(
    \rcv{\procB}{}{l}. \, 0
    \, \ExtCh \,
    \rcv{\procB}{}{m}. \, 0
    \, \ExtCh \,
    \rcv{\procB}{}{r}. \, t
)
$ }.
\end{example}

\begin{example}[Negative example for semi-full merge and positive example for full merge]
\label{ex:neg-ex-semi-full-merge}
The semi-full merge operator rejects the following implementable global type:

\vspace{-3ex}
{ \scriptsize \[
+
\begin{cases}
\msgFromTo{\procA}{\procB}{l}. \,
    \mu t_1. \,
    (
    \msgFromTo{\procB}{\procC}{l}. \,
    \msgFromTo{\procB}{\procA}{l}. \, t_1
    +
    \msgFromTo{\procB}{\procC}{m}. \,
    \msgFromTo{\procB}{\procA}{m}. \, 0
    )
\\
\msgFromTo{\procA}{\procB}{r}. \,
    \mu t_2. \,
    (
    \msgFromTo{\procB}{\procC}{m}. \,
    \msgFromTo{\procB}{\procA}{m}. \, 0
    +
    \msgFromTo{\procB}{\procC}{r}. \,
    \msgFromTo{\procB}{\procA}{r}. \, t_2
    )
\end{cases}
. \] }
\hspace{-1.5ex} Its FSM and the FSM after collapsing erasure is given in \cref{fig:neg-ex-semi-full-merge,fig:neg-ex-semi-full-merge-projected}.
Intuitively, it would need to recursively merge the parts after both recursion binders in order to merge the branches with receive event $\rcv{\procB}{\procC}{m}$ but it cannot do so.
The full merge can handle this global type.
It can descend beyond $q_1$ and $q_4$ and is able to merge $q'_1$ and $q'_4$.
To obtain $q''_{3\mid5}$, it applies Case \ref{merge:plainonly} while
$q'_{1\mid4}$ is only feasible with Case \ref{merge:semifullonly}.
The result is embedded into the recursive structure to obtain the FSM in \cref{fig:pos-ex-full-merge-merged}. It is equivalent to the (syntactic) result, which renames the recursion variable for one branch:
{ \scriptsize $
\mu t_1. \,
(
\rcv{\procB}{}{l}. \, t_1
\, \ExtCh \,
\rcv{\procB}{}{m}. \, 0
\, \ExtCh \,
\rcv{\procB}{}{r}. \, t_1
)
.$}
\end{example}

\begin{figure}[t]

\begin{subfigure}[b]{0.42\textwidth}
\centering
\resizebox{0.85\textwidth}{!}{
    \begin{tikzpicture}[sem , node distance=0.5cm and 0.7cm]
    \node[state, initial above, initial text = ] (q0) {$q_0$};
    \node[recstate, below = of q0, xshift=-16mm, yshift=1mm] (q1) {$q_1$};
    \node[state, below = of q1, yshift=1mm] (q1p) {$q'_1$};
    \node[state, below = of q1p, xshift=-4mm, yshift=1mm] (q2) {$q_2$};
    \node[varstate, below = of q2] (q2p) {$q'_2$};
    \node[state, below = of q1p, xshift=4mm, yshift=1mm] (q3) {$q_3$};
    \node[finalstate, below = of q3] (q3p) {$q'_3$};
    \node[recstate, below = of q0, xshift=16mm, yshift=1mm] (q4) {$q_4$};
    \node[state, below = of q4, yshift=1mm] (q4p) {$q'_4$};
    \node[state, below = of q4p, xshift=-4mm, yshift=1mm] (q5) {$q_5$};
    \node[finalstate, below = of q5] (q5p) {$q'_5$};
    \node[state, below = of q4p, xshift=4mm, yshift=1mm] (q6) {$q_6$};
    \node[varstate, below = of q6] (q6p) {$q'_6$};
\path (q0) edge node[left, yshift=1mm] {$\msgFromTo{\procA}{\procB}{l}$} (q1);
    \path (q1) edge node[left, yshift=.5mm] {$\emptystring$} (q1p);
    \path (q1p) edge node[left, yshift=.5mm] {$\msgFromTo{\procB}{\procC}{l}$} (q2);
    \path (q1p) edge node[right, yshift=.5mm] {$\msgFromTo{\procB}{\procC}{m}$} (q3);
    \path (q0) edge node[right, yshift=1mm] {$\msgFromTo{\procA}{\procB}{r}$} (q4);
    \path (q4) edge node[left, yshift=.5mm] {$\emptystring$} (q4p);
    \path (q4p) edge node[left, yshift=.5mm] {$\msgFromTo{\procB}{\procC}{m}$} (q5);
    \path (q4p) edge node[right, yshift=.5mm] {$\msgFromTo{\procB}{\procC}{r}$} (q6);

    \path (q2) edge node[left] {$\msgFromTo{\procB}{\procA}{l}$} (q2p);
    \path (q3) edge node[right] {$\msgFromTo{\procB}{\procA}{m}$} (q3p);
    \path (q5) edge node[left] {$\msgFromTo{\procB}{\procA}{m}$} (q5p);
    \path (q6) edge node[right] {$\msgFromTo{\procB}{\procA}{r}$} (q6p);

    \draw[semarrow] (q2p.south) -- node {} ++(0,-0.15) |- ++(-1.1cm,0) |- (q1) node[] {} node[] {};
    \draw[semarrow] (q6p.south) -- node {} ++(0,-0.15) |- ++(1.1cm,0) |- (q4) node[] {} node[] {};
\end{tikzpicture}
 }
\caption{Negative example for semi-full merge \\ and positive example for full merge}
\label{fig:neg-ex-semi-full-merge}
\end{subfigure}
\hfill
\begin{subfigure}[b]{0.33\textwidth}
\centering
\resizebox{0.90\textwidth}{!}{
    \begin{tikzpicture}[sem , node distance=0.7cm and 0.7cm]
    \node[state, initial above, initial text = ] (q0) {$q_0$};
    \node[recstate, below = of q0, xshift=-14mm] (q1) {$q_1$};
    \node[rcvstate, below = of q1] (q1p) {$q'_1$};
    \node[varstate, below = of q1p, xshift=-4mm] (q2) {$q''_2$};
\node[finalstate, below = of q1p, xshift=4mm] (q3) {$q''_3$};
\node[recstate, below = of q0, xshift=14mm] (q4) {$q_4$};
    \node[rcvstate, below = of q4] (q4p) {$q'_4$};
    \node[finalstate, below = of q4p, xshift=-4mm] (q5) {$q''_5$};
\node[varstate, below = of q4p, xshift=4mm] (q6) {$q''_6$};
\path (q0) edge node[left] {$\emptystring$} (q1);
    \path (q1) edge node[left] {$\emptystring$} (q1p);
    \path (q1p) edge node[left] {$\rcv{\procB}{\procC}{l}$} (q2);
    \path (q1p) edge node[right] {$\rcv{\procB}{\procC}{m}$} (q3);
    \path (q0) edge node[right] {$\emptystring$} (q4);
    \path (q4) edge node[left] {$\emptystring$} (q4p);
    \path (q4p) edge node[left] {$\rcv{\procB}{\procC}{m}$} (q5);
    \path (q4p) edge node[right] {$\rcv{\procB}{\procC}{r}$} (q6);

    \draw[semarrow] (q2.south) -- node {} ++(0,-0.15) |- ++(-0.77cm,0) |- (q1) node[] {} node[] {};
    \draw[semarrow] (q6.south) -- node {} ++(0,-0.15) |- ++(0.77cm,0) |- (q4) node[] {} node[] {};
\end{tikzpicture}
 }
\caption{After collapsing erasure \\ \phantom{sth}}
\label{fig:neg-ex-semi-full-merge-projected}
\end{subfigure}
\hfill
\begin{subfigure}[b]{0.23\textwidth}
\centering
\resizebox{0.90\textwidth}{!}{
    \begin{tikzpicture}[sem , node distance=0.7cm and 0.7cm]
\node[recstate, initial above, initial text = ] (q1) {$q_{1\mid4}$};
    \node[rcvstate, below = of q1] (q1p) {$q'_{1\mid4}$};
    \node[varstate, below = of q1p, xshift=-8mm] (q2) {$q''_2$};
\node[finalstate, below = of q1p, yshift=-10mm] (q3) {$q''_{3\mid5}$};
\node[varstate, below = of q1p, xshift=8mm] (q6) {$q''_6$};
\path (q1) edge node[left] {$\emptystring$} (q1p);
    \path (q1p) edge node[left] {$\rcv{\procB}{\procC}{l}$} (q2);
    \path (q1p) edge node[sloped, xshift=0.7mm, yshift=-1mm] {$\rcv{\procB}{\procC}{m}$} (q3);
\path (q1p) edge node[right] {$\rcv{\procB}{\procC}{r}$} (q6);

    \draw[semarrow] (q2.south) -- node {} ++(0,-0.15) |- ++(-0.9cm,0) |- (q1) node[] {} node[] {};
    \draw[semarrow] (q6.south) -- node {} ++(0,-0.15) |- ++(0.9cm,0) |- (q1) node[] {} node[] {};
\end{tikzpicture}
 }
\caption{After full merge \\ \phantom{sth}}
\label{fig:pos-ex-full-merge-merged}
\end{subfigure}

\vspace{-1ex}
\caption{
The FSM on the left represents an implementable global type that is rejected by the semi-full merge.
It is accepted by the full merge:
collapsing erasure yields the FSM in the middle and applying the full merge the FSM on the right.
}
\label{fig:neg-ex-semi-full-merge-steps}
\vspace{-3ex}
\end{figure}

\begin{example}[Negative example for full merge]
We consider a simple implementable global type where $\procA$ propagates its decision to $\procC$ in the top branch while $\procB$ propagates it in the bottom branch:
{ \scriptsize $
+
\begin{cases}
\msgFromTo{\procA}{\procB}{l}. \,
\msgFromTo{\procA}{\procC}{l}. \, 0
\\
\msgFromTo{\procA}{\procB}{r}. \,
\msgFromTo{\procB}{\procC}{r}. \, 0
\end{cases}
$ }\hspace{-2.5ex}.
It is illustrated in \cref{fig:neg-ex-full-merge}.
This cannot be projected onto $\procC$ by the full merge operator for which all receive events need to have the same sender.
\end{example}

\begin{remark}[On sender-driven choice]
\label{rem:sender-driven-choice}
Majumdar et al.~\cite{DBLP:conf/concur/MajumdarMSZ21} proposed a classical projection operator that overcomes this shortcoming.
It can project the previous example.
In general, allowing to receive from different senders has subtle consequences.
Intuitively, messages from different senders could overtake each other in a distributed setting and one cannot rely on the FIFO order provided by the channel of a single sender.
Majumdar et al.\ employ a message availability analysis to ensure that there cannot be any confusion about which branch shall be taken.
Except for the possibility to merge cases where a receiver receives from multiple senders, their merge operator suffers from the same shortcomings as all classical projection operators.
For details, we refer to their work~\cite{DBLP:conf/concur/MajumdarMSZ21}.
\end{remark}

\paragraph*{Shortcomings of Classical Projection/Merge Operators}

We present slight variations of the two buyer protocol that are implementable but rejected by all of the presented projection/merge operators.

\begin{example}
We obtain an implementable variant by omitting both message interactions
$\msgFromTo{\buyerA}{\seller}{\no}$
with which buyer $\buyerA$ notifies seller $\seller$ that they will not buy the item:

\vspace{-3ex}
{
\scriptsize
\[
\mu t. \,
+
\begin{cases}
\msgFromTo{\buyerA}{\seller}{\query}. \,
\msgFromTo{\seller}{\buyerA}{\price}. \,
\bigl(
    \msgFromTo{\buyerA}{\buyerB}{\splitmsg}. \,
        (
        \msgFromTo{\buyerB}{\buyerA}{\yes}. \,
        \msgFromTo{\buyerA}{\seller}{\buy}. \, t
        +
        \msgFromTo{\buyerB}{\buyerA}{\no}. \, t
)
+
    \msgFromTo{\buyerA}{\buyerB}{\cancel}. \, t
\bigr)
\\
\msgFromTo{\buyerA}{\seller}{\done}. \,
\msgFromTo{\buyerA}{\buyerB}{\done}. \, 0
\end{cases}
.
\]
}
\hspace{-1.5ex} This global type cannot be projected onto seller $\seller$.
The merge operator would need to merge a recursion variable with an external choice.
Visually, the merge operator does not allow to unfold the variable $t$ and try to merge again.
However, there is a local type for seller~$\seller$:

\vspace{-3ex}
{
\scriptsize \[
\mu t_1. \,
\ExtCh
\begin{cases}
\rcv{\buyerA}{}{\query}. \,
\mu t_2. \,
\snd{}{\buyerA}{\price}. \,
(
    \rcv{\buyerA}{}{\buy}. \, t_1
\ExtCh
    \rcv{\buyerA}{}{\query}. \, t_2
\ExtCh
    \rcv{\buyerA}{}{\done}. \, 0
)
\\
\rcv{\buyerA}{}{\done}. \, 0
\end{cases}
.
\] }
\hspace{-1.5ex} The local type has two recursion variable binders while the global type only has~one.
Classical projection operators can never yield such a structural change:
the merge operator can only merge states but not introduce new ones or introduce new backward transitions.
\end{example}

\begin{example}[Two Buyer Protocol with Subscription]
\label{ex:2BPWS}
\label{ex:two-buyers-with-subscription}

In this variant, buyer $\buyerA$ first decides whether to subscribe to a yearly discount offer or not --- before purchasing the sequence of items --- and notifies buyer $\buyerB$ if it does so:
{ \scriptsize $
\GG_{\TBPWS}
    \; \is \;
+
\begin{cases}
\msgFromTo{\buyerA}{\seller}{\login}. \,
\GG_{\TBPOR}
\\
\msgFromTo{\buyerA}{\seller}{\subscribe}. \,
\msgFromTo{\buyerA}{\buyerB}{\subscribed}. \,
\GG_{\TBPOR}
\end{cases}
\hspace{-3ex}. $}
The merge operator needs to merge a recursion variable binder $\mu t$ with the external choice $\rcv{\buyerA}{\buyerB}{\subscribed}$.
Still, there is a local type $L_{\buyerB}$ for $\buyerB$ such that $\lang(L_{\buyerB}) = \lang(\GG_{\TBPWS}) \wproj_{\Alphabet_\buyerB}$:

\vspace{-3ex}
{ \scriptsize
\[
L_{\buyerB} \is
\ExtCh
    \begin{cases}
    \rcv{\buyerA}{}{\splitmsg}. \,
        (
        \snd{}{\buyerA}{\yes}. \, L(t_1)
        \IntCh
        \snd{}{\buyerA}{\no}. \, L(t_2)
        )
    \\
    \rcv{\buyerA}{}{\cancel}. \, L(t_3)
    \\
    \rcv{\buyerA}{}{\done}. \, 0
    \\
    \rcv{\buyerA}{}{\subscribed}. \,
    L(t_4)
    \end{cases}
\hspace{-5ex}
\text{ where }
\;
    L(t) \is
    \mu t. \,
    \ExtCh
        \begin{cases}
\rcv{\buyerA}{}{\splitmsg}. \,
            (
            \snd{}{\buyerA}{\yes}. \, t
            \IntCh
            \snd{}{\buyerA}{\no}. \, t
            )
        \\
\rcv{\buyerA}{}{\cancel}. \, t
        \\
\rcv{\buyerA}{}{\done}. \, 0
\end{cases}
.
\]
}
\hspace{-1.5ex} In fact, one can also rely on the fact that buyer $\buyerA$ will comply with the intended protocol.
Then, it suffices to introduce one recursion variable $t$ in the beginning and substitute every $L(\hole)$ with~$t$, yielding a local type $L'_\buyerB$ with $\lang(L_\buyerB) \subseteq \lang(L'_\buyerB)$.
\end{example}

Similarly, classical projection operator cannot handle global types where choices can be disambiguated with semantic properties, e.g., counting modulo a constant.
Scalas and Yoshida~\cite{DBLP:journals/pacmpl/ScalasY19} also identified another shortcoming:
most classical projection operators require all branches of a loop to contain the same set of active roles.
Thus, they cannot project the following global type.
It is implementable and if it was projectable, the result would be equivalent to the local types given in
their example~\cite[Fig.~4~(2)]{DBLP:journals/pacmpl/ScalasY19}.

\begin{example}[Two Buyer Protocol with Inner Recursion]
\label{ex:2BPIR}
This variant allows to recursively negotiate how to split the price (and omits the outer recursion):

{ \vspace{-3ex}
  \scriptsize \[
\GG_{\TBPIR}
    \quad \is \quad
\msgFromTo{\buyerA}{\seller}{\query}. \,
\msgFromTo{\seller}{\buyerA}{\price}. \,
\mu t. \,
+
\begin{cases}
\msgFromTo{\buyerA}{\buyerB}{\splitmsg}. \,
(
    \msgFromTo{\buyerB}{\buyerA}{\yes}. \,
    \msgFromTo{\buyerA}{\seller}{\buy}. \, 0
+
    \msgFromTo{\buyerB}{\buyerA}{\no}. \, t
    )
\\
\msgFromTo{\buyerA}{\buyerB}{\cancel}. \,
\msgFromTo{\buyerA}{\seller}{\no}. \, 0
\end{cases}
.
\] }
\end{example}

These shortcomings have been addressed by some non-classical approaches.
For example, Scalas and Yoshida~\cite{DBLP:journals/pacmpl/ScalasY19} employ model checking while Castagna et al.~\cite{DBLP:journals/corr/abs-1203-0780} characterise implementable global types with an undecidable well-formedness condition and give a sound algorithmically checkable approximation.
It is not known whether the implementability problem for global types, neither with directed or sender-driven choice, is decidable.
We answer this question positively for the more general case of sender-driven~choice.

 \section{Implementability for Global Types from MSTs is Decidable}
\label{sec:standard-implementability-decidable}

In this section, we show decidability of the implementability problem for global types with sender-driven choice, using results from the domain of message sequence charts.
We introduce high-level message sequence charts (HMSCs) and recall an HMSC encoding for global types.
In general, implementability for HMSCs is undecidable but we show that global types, when encoded as HMSCs, belong to a class of HMSCs for which implementability is~decidable.

\subsection{High-level Message Sequence Charts}
\label{sec:hmscs}

Our definitions of (high-level) message sequence charts follow work by
Genest et al.~\cite{DBLP:journals/fuin/GenestKM07} and
Stutz and Zufferey~\cite{DBLP:journals/corr/abs-2209-10328}.
If reasonable, we adapt terminology to the MST setting.

\begin{definition}[Message Sequence Charts]
\label{def:msc}
A \emph{message sequence chart (MSC)} is a $5$-tuple
$M = (\textcolor{colorblind1}{\eventnodes},\textcolor{colorblind2}{p},\textcolor{colorblind3}{f},\textcolor{colorblind4}{l},(\leq_\procA)_{\procA \in \Procs })$
where

\noindent
\begin{minipage}{0.66\textwidth}
\begin{itemize}
\vspace{0.5ex}
\item $\textcolor{colorblind1}{\eventnodes}$ is a set of send $(\SndEvs)$ and receive $(\RcvEvs)$ event nodes such that $\eventnodes = \SndEvs ⊎ \RcvEvs$ (where $\dunion$ denotes disjoint union),\item $\textcolor{colorblind2}{p} \from \eventnodes \to \Procs$ maps each event node to the role acting on it,
\item $\textcolor{colorblind3}{f} \from \SndEvs \to \RcvEvs$ is an injective function linking \\ corresponding send and receive event nodes,
\item $\textcolor{colorblind4}{l} \from \eventnodes \to Σ$ labels every event node with an event, and
\item $(\leq_\procA)_{\procA \in \Procs }$ is a family of total orders for the\\
event nodes of each role: $\leq_\procA \; \subseteq \; \inv{p}(\procA) \times \inv{p}(\procA)$.
\end{itemize}
\end{minipage}\begin{minipage}{0.03\textwidth}
\phantom{s}
\end{minipage}\begin{minipage}{0.31\textwidth}
\begin{footnotesize}
    \centering
    \def\svgwidth{0.83\textwidth}
    \begingroup \makeatletter \providecommand\color[2][]{\errmessage{(Inkscape) Color is used for the text in Inkscape, but the package 'color.sty' is not loaded}\renewcommand\color[2][]{}}\providecommand\transparent[1]{\errmessage{(Inkscape) Transparency is used (non-zero) for the text in Inkscape, but the package 'transparent.sty' is not loaded}\renewcommand\transparent[1]{}}\providecommand\rotatebox[2]{#2}\newcommand*\fsize{\dimexpr\f@size pt\relax}\newcommand*\lineheight[1]{\fontsize{\fsize}{#1\fsize}\selectfont}\ifx\svgwidth\undefined \setlength{\unitlength}{101.66858343bp}\ifx\svgscale\undefined \relax \else \setlength{\unitlength}{\unitlength * \real{\svgscale}}\fi \else \setlength{\unitlength}{\svgwidth}\fi \global\let\svgwidth\undefined \global\let\svgscale\undefined \makeatother \begin{picture}(1,0.41336342)\lineheight{1}\setlength\tabcolsep{0pt}\put(0,0){\includegraphics[width=\unitlength,page=1]{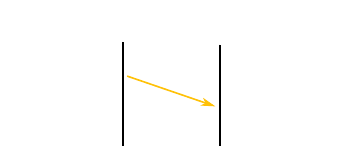}}\put(0.03119107,0.18469745){\color[rgb]{0,0.30196078,0.25098039}\makebox(0,0)[lt]{\lineheight{1.25}\smash{\begin{tabular}[t]{l}$\snd{\procAunc}{\procBunc}{\val}$\end{tabular}}}}\put(0,0){\includegraphics[width=\unitlength,page=2]{figs/bmsc.pdf}}\put(0.65077843,0.08451202){\color[rgb]{0,0.30196078,0.25098039}\makebox(0,0)[lt]{\lineheight{1.25}\smash{\begin{tabular}[t]{l}$\rcv{\procAunc}{\procBunc}{\val}$\end{tabular}}}}\put(0.32458469,0.3304579){\color[rgb]{0,0,0}\makebox(0,0)[lt]{\lineheight{1.25}\smash{\begin{tabular}[t]{c}$\procAunc$\end{tabular}}}}\put(0.5993836,0.32909311){\color[rgb]{0,0,0}\makebox(0,0)[lt]{\lineheight{1.25}\smash{\begin{tabular}[t]{c}$\procBunc$\end{tabular}}}}\end{picture}\endgroup    \captionof{figure}{Highlighting the elements of an MSC $(\textcolor{colorblind1}{\eventnodes},\textcolor{colorblind2}{p},\textcolor{colorblind3}{f},\textcolor{colorblind4}{l},(\leq_{\procAunc})_{\procAunc \in \Procs })$ \cite[Fig.~3]{DBLP:journals/corr/abs-2209-10328}}
  \label{fig:bmsc-color}
    \end{footnotesize}
\vspace{-1.5ex}
\end{minipage}
\vspace{1ex}

\noindent An MSC $M$ induces a partial order $\leq_M$ on $\eventnodes$ that is defined co-inductively\emph{:}

\vspace{-3ex}
{
\scriptsize
\begin{mathpar}

\inferrule*[right=proc]{\label{constr:proc}
e \leq_\procA e'
}{
e \leq_M e'
}

\inferrule*[right=snd-rcv]{\label{constr:snd-rcv}
s \in \SndEvs
}{
s \leq_M f(s)
}

\inferrule*[right=refl]{\label{constr:refl}
}{
e \leq_M e
}

\inferrule*[right=trans]{\label{constr:trans}
e \leq_M e' \\
e' \leq_M e''
}{
e \leq_M e''
}
\end{mathpar}
}
\hspace{-2ex} The labelling function $l$ respects the function $f$:
for every send event node $e$, we have that
$l(e) = \snd{p(e)}{p(f(e))}{\val}$ and $l(f(e)) = \rcv{p(e)}{p(f(e))}{\val}$ for some $\val \in \MsgVals$.
\end{definition}

All MSCs in our work respect FIFO, i.e., there are no $\procA$ and $\procB$ such that there are $e_1, e_2 \in \inv{p}(\procA)$ with $e_1 \neq e_2 $, $l(e_1) = l(e_2)$, $e_1 \leq_\procA e_2$ and $f(e_2) \leq_\procB f(e_1)$ (also called degenerate) and
for every pair of roles $\procA$, $\procB$, and for every two event nodes $e_1 \leq_M e_2$ with $l(e_i) = \snd{\procA}{\procB}{\_}$\phantom{x}for $i \in \set{1,2}$, it holds that $\MsgVals(w_\procA) = \MsgVals(f(w_\procA))$ where $w_\procA$ is the (unique) linearisation of $\inv{p}(\procA)$.
A \emph{basic MSC (BMSC)} has a finite number of nodes $\eventnodes$ and $\BMSCs$ denotes the set of all BMSCs.
When unambiguous, we omit the index $M$ for $\leq_M$ and write $\leq$.
We define $\lneq$ as expected.
The language $\lang(M)$ of an MSC $M$ collects all words $l(w)$ for which $w$ is a linearisation of $\eventnodes$ that is compliant with $\leq_M$.

If one thinks of a BMSC as straight-line code, a high-level message sequence chart adds control flow.
It embeds BMSCs into a graph structure which allows for choice and recursion.

\begin{definition}[High-level Message Sequence Charts] A \emph{high-level message sequence chart (HMSC)} is a $5$-tuple $(V, \edges, \vertexA^I\negmedspace, V^T\negmedspace\!, \mu)$ where
$V$ is a finite set of vertices,
$\edges \subseteq V \times V$ is a set of directed edges,
$\vertexA^I \in V$ is an initial vertex,
$V^T \subseteq V$ is a set of terminal vertices, and
$\mu: V \to \BMSCs$ is a function mapping every vertex to a BMSC.
 A path in an HMSC is a sequence of vertices $\vertexA_1, \ldots$ from $V$ that is connected by edges, i.e.,
 $(\vertexA_i, \vertexA_{i+1}) \in \edges$ for every $i$.
 A path is \emph{maximal} if it is infinite or ends in a vertex from $V^T\negmedspace$.
\end{definition}

Intuitively, the language of an HMSC is the union of all languages of the finite and infinite MSCs generated from maximal paths in the HMSC and is formally defined in\iftoggle{arxiv}
{
\cref{app:hmscs}.
}
{
the technical report \cite{arxiv-version}.
}
Like global types, an HMSC specifies a protocol.
The implementability question was also posed for HMSCs and studied as \emph{safe realisability}.
If the CSM is not required to be deadlock-free, it is called weak realisability.

\begin{definition}[Safe realisability of HMSCs \cite{DBLP:journals/tcs/AlurEY05}]
 An HMSC $H$ is said to be \emph{safely realisable} if there exists a deadlock-free CSM $\CSM{A}$ such that
 $\lang(H) = \lang(\CSM{A})$.
\end{definition}

\smallskip\noindent
{\sffamily\bfseries
Encoding Global Types from MSTs as HMSCs}
\label{sec:mst-to-hmsc}

\noindent
Stutz and Zufferey~\cite[Sec.\ 5.2]{DBLP:journals/corr/abs-2209-10328} provide a formal encoding $H(\hole)$ from global types to HMSCs.
We refer to\iftoggle{arxiv}
{
\cref{app:hmsc-encoding}}
{
the technical report~\cite{arxiv-version}}
for the definition.
We adapt their correctness result to our setting.
In particular, our semantics of $\GG$ use the closure operator $\interswaplang(\hole)$ while they distinguish between a type and execution language.
We also omit the closure operator on the right-hand side because HMSCs are closed with regard to this operator~\cite[Lm.~5]{DBLP:journals/corr/abs-2209-10328}.

\begin{theorem}
\label{thm:correctnessMPSTtoHMSC}
Let $\GG$ be a global type.
Then,
the following holds:
$\lang(\GG) = \lang(H(\GG))$.
\end{theorem}

\subsection{Implementability is Decidable}
\label{sec:imp-decidable}

We introduce a mild assumption for global types.
Intuitively, we require that every run of the protocol can always terminate but does not need to.
Basically, this solely rules out global types that have loops without exit (cf.~\cref{ex:independent-pairs}).
In practice, it is reasonable to assume a mechanism to terminate a protocol for maintenance for instance.
Note that this assumption constitutes a structural property of a protocol and no fairness condition on runs of the protocol.

\begin{assumption*}[$0$-Reachable]
\label{assumption:can-always-finish}
We say a global type $\GG$ is $0$-reachable if every prefix of a word in its language can be completed to a finite word in its language.
Equivalently, we require that the vertex for the syntactic subterm $0$ is reachable from any vertex in~$H(\GG)$. \end{assumption*}

The MSC approach to safe realisability for HMSCs is different from the classical projection approach to implementability.
Given an HMSC, there is a canonical candidate implementation which always implements the HMSC if an implementation exists \cite[Thm.~13]{DBLP:journals/tse/AlurEY03}.
Therefore, approaches center around checking safe realisability of HMSC languages and establishing conditions on HMSCs that entail safe realisability.

\begin{definition}[Canonical candidate implementation~\cite{DBLP:journals/tse/AlurEY03}]
\label{def:canonical-candidate-impl}
Given an HMSC~$H$ and a role~$\procA$,
let $A'_\procA = (Q', \Alphabet_\procA, \delta', q'_{0}, F')$ be a state machine with
$Q' \is \set{q_w
            \mid
           w \in \pref(\lang(H)\wproj_{\Alphabet_\procA})}$,
$F' \is \set{q_w
            \mid
           w \in \lang_{\fin}(H) \wproj_{\Alphabet_\procA}}$,
and
$\delta'(q_w, x, q_{wx})$ for $x \in \AlphAsync$.
The resulting state machine~$A'_\procA$ is not necessarily finite so
$A'_\procA$ is determinised and minimised which yields the FSM $A_\procA$.
We call $\CSM{A}$ the canonical candidate implementation of $H$.
\end{definition}

Intuitively, the intermediate state machine $A'_\procA$ constitutes a tree whose maximal finite paths give $\lang(H) \wproj_{\Alphabet_\procA} \inters \Alphabet_{\procA}^*$.
This set can be infinite and, thus, the construction might not be effective.
We give an effective construction of a deterministic FSM for the same language which was very briefly hinted at by Alur et al.~\cite[Proof of Thm.~3]{DBLP:journals/tcs/AlurEY05}. 

\begin{definition}[Projection by Erasure]
Let $\procA$ be a role and $M = (\eventnodes, p, f, l, (\leq_\procA)_{\procA \in \Procs })
$ be an MSC.
We denote the set of nodes of $\procA$ with
$N_\procA \is \set{n \mid p(n) = \procA}$ and define a two-ary $\operatorname{next}$-relation on $N_\procA$\emph{:}
$\operatorname{next}(n_1, n_2)$ iff
$n_1 \lneq n_2$ and there is no $n'$ with $n_1 \lneq n' \lneq n_2$.
We define the projection by erasure of $M$
on to $\procA$:
$
M \wproj_\procA
    =
(Q_M, \Alphabet_\procA, \delta_M, q_{M,0}, \set{q_{M,f}})
$
with
\vspace{-2ex}
{ \scriptsize
\begin{align*}
Q_M
& \is \set{q_n \mid n \in N_\procA} \dunion \set{q_{M,0}} \dunion \set{q_{M,f}}
\text{ and }
\\
\delta_M & \is
    \set{
        q_{M, 0} \xrightarrow{\emptystring} q_{n_1}
            \mid
        \forall n_2.\, n_1 \leq n_2
    }
    \, \dunion \,
    \set{
        q_{n_1} \xrightarrow{l(n_1)} q_{n_2}
            \mid
        \operatorname{next}(n_1, n_2)
    }
    \, \dunion \,
    \set{
        q_{n_2} \xrightarrow{l(n_2)} q_{M,f}
            \mid
        \forall n_1.\, n_1 \leq n_2
    }
\end{align*}
}
\hspace{-2ex} where $\dunion$ denotes disjoint union.
Let $H = (V, \edges, \vertexA^I\negmedspace, V^T\negmedspace\!, \mu)$ be an HMSC.
We construct the projection by erasure for every vertex and identify them with the vertex, e.g., $Q_\vertexA$ instead of~$Q_{\mu(\vertexA)}$.
We construct an auxiliary FSM
$
(Q'_H, \Alphabet_\procA, \delta'_H, q'_{H,0}, F'_H)
$
with
$Q'_H = \Dunion_{\vertexA \in V} Q_\vertexA$,
$\delta'_H =
    \Dunion_{\vertexA \in V} \delta_\vertexA \dunion
    \set{ q_{\vertexA_1,f} \xrightarrow{\emptystring} q_{\vertexA_2,0} \mid (\vertexA_1, \vertexA_2) \in E}$,
$q'_{H,0} = q_{\vertexA^I\negmedspace,0}$,
and
$F'_H = \Dunion_{\vertexA \in V^F} q_{\vertexA, f}$.
We determinise and minimise
$(Q'_H, \Alphabet_\procA, \delta'_H, q'_{H,0}, F'_H)$
to obtain
$H \wproj_\procA \is (Q_H, \Alphabet_\procA, \delta_H, q_{H,0}, F_H)$,
which we define to be the \emph{projection by erasure of $H$ onto $\procA$}.
The CSM formed from the projections by erasure $\CSMl{H \wproj_{\procA}}$ is called \emph{erasure candidate implementation}.
\end{definition}

\begin{restatable}[Correctness of Projection by Erasure]{lemma}{projectionByErasureCorrect}
\label{lm:projection-by-erasure-correct}
Let $H$ be an HMSC, $\procA$ be a role, and $H \wproj_\procA$ be its projection by erasure.
Then, the following language equality holds:
$\lang(H \wproj_\procA) = \lang(H) \wproj_{\Alphabet_\procA}$.
\end{restatable}

The proof is straightforward and can be found in\iftoggle{arxiv}
{
\cref{proof:projection-by-erasure-correct}.
}
{
the technical report \cite{arxiv-version}.
}
From this result and the construction of the canonical candidate implementation, it follows that the projection by erasure admits the same finite language.

\begin{corollary}
\label{cor:erasure-canonical}
Let $H$ be an HMSC, $\procA$ be a role, $H \wproj_\procA$ be its projection by erasure, and
$A_\procA$ be the canonical candidate implementation.
Then, it holds that
$\lang_{\fin}(H \wproj_\procA)
    =
\lang_{\fin}(A_\procA)$.
\end{corollary}

The projection by erasure can be computed effectively and is deterministic.
Thus, we use it in place of the canonical candidate implementation.
Given a global type, the erasure candidate implementation for its HMSC encoding implements it if it is implementable.

\begin{restatable}{theorem}{erasureCandidateImpl}\label{thm:erasure-candidate-impl}
Let $\GG$ be a global type and $\CSMl{H(\GG) \wproj_\procA}$ be its erasure candidate implementation.
If $\lang_{\fin}(\GG)$ is implementable\xspace \footnote{\emph{Implementability} is lifted to languages as expected.}\negthinspace
, then
$\CSMl{H(\GG) \wproj_\procA}$ is deadlock-free and\linebreak[4] \mbox{$\lang_{\fin}(\CSMl{H(\GG) \wproj_\procA}) = \lang_{\fin}(\GG)$}.\end{restatable}

This result does only account for finite languages so we extend it for infinite sequences.
For both, the proof can be found in\iftoggle{arxiv}
{
\cref{proof:erasure-candidate-impl} and \cref{proof:finite-generalises-to-infinite}.
}
{
the technical report~\cite{arxiv-version}.
}

\begin{restatable}[''Finite implementation`` generalises to infinite language for $0$-reachable global types]{lemma}{finiteGeneralisesToInfinite}
\label{lm:finite-generalises-to-infinite}
Let $\GG$ be a $0$-reachable global type and $\CSM{A}$ be an implementation for~$\lang_{\fin}(\GG)$.Then, it holds that
$\lang_{\inf}(\CSM{A}) = \lang_{\inf}(\GG)$,
and, thus,
$\lang(\CSM{A}) = \lang(\GG)$.
\end{restatable}

\begin{corollary}
\label{cor:erasure-implements-global-type}
Let $\GG$ be a $0$-reachable implementable global type.
Then, the erasure candidate implementation $\CSMl{H(\GG) \wproj_\procA}$ implements $\GG$.
\end{corollary}

So far, we have shown that, if $\GG$ is implementable, the erasure candidate implementation for its HMSC encoding $H(\GG)$ implements $\GG$.
For HMSCs, this is undecidable in general~\cite{DBLP:journals/tcs/Lohrey03}.
We show that, because of their syntactic restrictions on choice, global types fall into the class of globally-cooperative HMSCs for which implementability is decidable.

\begin{definition}[Communication graph \cite{DBLP:journals/jcss/GenestMSZ06}]
Let $M = (\eventnodes, p, f, l, (\leq_\procA)_{\procA \in \Procs})$ be an MSC.
The \emph{communication graph} of $M$ is a directed graph with node $\procA$ for every role $\procA$ that sends or receives a message in $M$ and edges $\procA \to \procB$ if $M$ contains a message from $\procA$ to $\procB$, i.e.,
there is $e \in N$ such that
$p(e) = \procA$ and $p(f(e)) = \procB$.
\end{definition}

It is important that the communication graph of $M$ does not have a node for every role but only the active ones, i.e., the ones that send or receive in $M$.

\begin{definition}[Globally-cooperative HMSCs \cite{DBLP:journals/jcss/GenestMSZ06}]
An HMSC $H = (V,\edges,v^I,V^T\negmedspace,μ)$ is called \emph{globally-cooperative} if for every loop, i.e., $\vertexA_1, \ldots, \vertexA_n$ with $(\vertexA_i, \vertexA_{i+1}) \in \edges$ for every $1 \leq i < n$ and $(\vertexA_n, \vertexA_1) \in \edges$,
the communication graph of $\mu(\vertexA_1)\ldots\mu(\vertexA_n)$ is weakly connected, i.e., all nodes are connected if every edge is considered undirected.
\end{definition}

\noindent
We can check this directly for a global type $\GG$.
It is straight\-forward to define a communication graph for
words from $\AlphSync^*$.
We check it on $\semglobalsync(\GG)$:
for each binder state, we check the communication graph for the shortest trace to every corresponding recursion~state.

\begin{theorem}[Thm.~3.7 \cite{DBLP:journals/tcs/Lohrey03}]
\label{lm:safe-realisability-glob-coop-expspace}
Let $H$ be a globally-cooperative HMSC.
Restricted to its finite language $\lang_{\fin}(H)$, safe realisability is \EXPSPACE-complete.
\end{theorem}

\begin{restatable}{lemma}{implementabilityEntailsGlobCoop}
\label{lm:implementability-entails-glob-coop}
Let $\GG$ be an implementable $0$-reachable global type.
Then, its HMSC encoding $H(\GG)$ is globally-cooperative.
\end{restatable}

\iftoggle{arxiv}
{
The proof can be found in \cref{proof:imp-entails-glob-coop} }
{
The proof can be found in the technical report~\cite{arxiv-version} }
and is far from trivial.
We explain the main intuition for the proof with the following example where we exemplify why the same result does not hold for HMSCs in general.

\begin{figure}[t]

    \begin{subfigure}[b]{0.42\textwidth}
    \centering
    \def\svgwidth{0.50\textwidth}
    \begingroup \makeatletter \providecommand\color[2][]{\errmessage{(Inkscape) Color is used for the text in Inkscape, but the package 'color.sty' is not loaded}\renewcommand\color[2][]{}}\providecommand\transparent[1]{\errmessage{(Inkscape) Transparency is used (non-zero) for the text in Inkscape, but the package 'transparent.sty' is not loaded}\renewcommand\transparent[1]{}}\providecommand\rotatebox[2]{#2}\newcommand*\fsize{\dimexpr\f@size pt\relax}\newcommand*\lineheight[1]{\fontsize{\fsize}{#1\fsize}\selectfont}\ifx\svgwidth\undefined \setlength{\unitlength}{174.51961728bp}\ifx\svgscale\undefined \relax \else \setlength{\unitlength}{\unitlength * \real{\svgscale}}\fi \else \setlength{\unitlength}{\svgwidth}\fi \global\let\svgwidth\undefined \global\let\svgscale\undefined \makeatother \begin{picture}(1,1.025932)\lineheight{1}\setlength\tabcolsep{0pt}\put(0,0){\includegraphics[width=\unitlength,page=1]{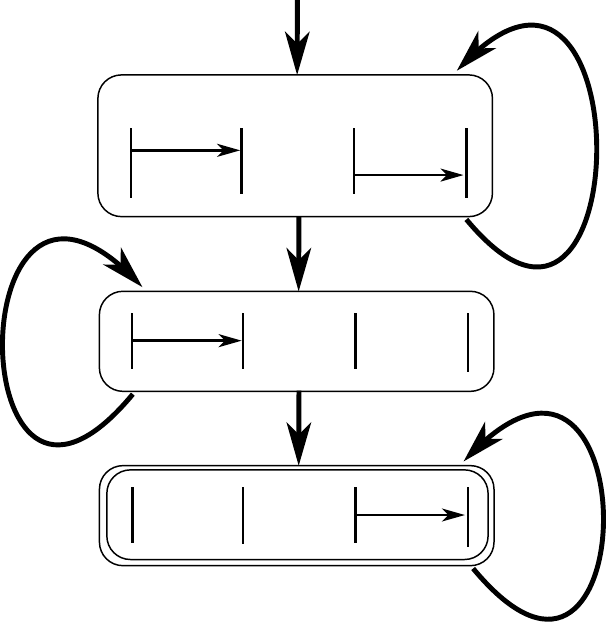}}\put(0.19564343,0.83389756){\color[rgb]{0,0,0}\makebox(0,0)[lt]{\lineheight{1.25}\smash{\begin{tabular}[t]{l}{\footnotesize $\procA$}\end{tabular}}}}\put(0.37454991,0.83321277){\color[rgb]{0,0,0}\makebox(0,0)[lt]{\lineheight{1.25}\smash{\begin{tabular}[t]{l}{\footnotesize $\procB$}\end{tabular}}}}\put(0.55958787,0.83568425){\color[rgb]{0,0,0}\makebox(0,0)[lt]{\lineheight{1.25}\smash{\begin{tabular}[t]{l}{\footnotesize $\procC$}\end{tabular}}}}\put(0.74648787,0.83560048){\color[rgb]{0,0,0}\makebox(0,0)[lt]{\lineheight{1.25}\smash{\begin{tabular}[t]{l}{\footnotesize $\procD$}\end{tabular}}}}\put(0.25709351,0.78742722){\color[rgb]{0,0,0}\makebox(0,0)[lt]{\lineheight{1.25}\smash{\begin{tabular}[t]{l}{\scriptsize $\val$}\end{tabular}}}}\put(0.25702983,0.47227849){\color[rgb]{0,0,0}\makebox(0,0)[lt]{\lineheight{1.25}\smash{\begin{tabular}[t]{l}{\scriptsize $\val$}\end{tabular}}}}\put(0.62090453,0.74491243){\color[rgb]{0,0,0}\makebox(0,0)[lt]{\lineheight{1.25}\smash{\begin{tabular}[t]{l}{\scriptsize $\val$}\end{tabular}}}}\put(0.62571117,0.18504196){\color[rgb]{0,0,0}\makebox(0,0)[lt]{\lineheight{1.25}\smash{\begin{tabular}[t]{l}{\scriptsize $\val$}\end{tabular}}}}\end{picture}\endgroup  \caption{HMSC $H_{\operatorname{ing}}$}
    \label{fig:implementable-hmsc-not-glob-coop}
    \end{subfigure}
    \hfill
    \begin{subfigure}[b]{0.57\textwidth}
    \centering
    \resizebox{0.60\textwidth}{!}{
        \begin{tikzpicture}[sem, node distance=0.7cm and 0.5cm]
\node[state, initial above, initial text = ] (qA0) {$q_{0, \procA}$};
    \node[state, below = of qA0, yshift=1.5mm] (qA1) {$q_{1, \procA}$};
    \node[state, accepting, below = of qA1, yshift=1.5mm] (qA2) {$q_{2, \procA}$};

    \path (qA0) edge node[right, yshift=1mm] {$\snd{\procA}{\procB}{\val}$} (qA1);
    \path (qA1) edge node[right, yshift=1mm] {$\snd{\procA}{\procB}{\val}$} (qA2);
    \path (qA2) edge[loop below] node {$\snd{\procA}{\procB}{\val}$} (qA2);

\node[state, initial above, xshift=15mm, initial text = ] (qB0) {$q_{0, \procB}$};
    \node[state, below = of qB0, yshift=1.5mm] (qB1) {$q_{1, \procB}$};
    \node[state, accepting, below = of qB1, yshift=1.5mm] (qB2) {$q_{2, \procB}$};

    \path (qB0) edge node[right, yshift=1mm] {$\rcv{\procA}{\procB}{\val}$} (qB1);
    \path (qB1) edge node[right, yshift=1mm] {$\rcv{\procA}{\procB}{\val}$} (qB2);
    \path (qB2) edge[loop below] node {$\rcv{\procA}{\procB}{\val}$} (qB2);

\node[state, initial above, xshift=30mm, initial text = ] (qC0) {$q_{0, \procC}$};
\node[state, accepting, below = of qC0, yshift=1.5mm] (qC2) {$q_{1, \procC}$};

    \path (qC0) edge node[right, yshift=1mm] {$\snd{\procC}{\procD}{\val}$} (qC2);
\path (qC2) edge[loop below] node {$\snd{\procC}{\procD}{\val}$} (qC2);

\node[state, initial above, xshift=45mm, initial text = ] (qD0) {$q_{0, \procD}$};
\node[state, accepting, below = of qD0, yshift=1.5mm] (qD2) {$q_{1, \procD}$};

    \path (qD0) edge node[right, yshift=1mm] {$\rcv{\procC}{\procD}{\val}$} (qD2);
\path (qD2) edge[loop below] node {$\rcv{\procC}{\procD}{\val}$} (qD2);
\end{tikzpicture}
     }
    \caption{An implementation for $H_{\operatorname{ing}}$}
    \label{fig:csm-ing}
    \end{subfigure}

    \vspace{-1ex}
    \caption{An implementable HMSC which is not globally-cooperative with its implementation}
    \vspace{-2ex}

\end{figure}

\begin{example}[Implementable HMSC but not globally cooperative]
\label{ex:implementable-not-glob-coop}
HMSC $H_{\operatorname{ing}}$ in \cref{fig:implementable-hmsc-not-glob-coop} is implementable but neither globally-cooperative nor representable with a global type.
In the first loop, $\procA$~sends a message~$\val$ to $\procB$ while $\procC$ sends a message $\val$ to $\procD$ so the communication graph is not weakly connected.
In the second loop, only the interaction between $\procA$ and $\procB$ is specified, while, in the third one, it is only the one between $\procC$ and~$\procD$.
For a variant of the protocol without the second loop, any candidate implementation can always expose an execution with more interactions between $\procA$ and~$\procB$ than the ones between $\procC$ and~$\procD$, due to the lack of synchronisation. Here, the second loop can make up for such executions so any execution has a path in $H_{\operatorname{ing}}$.
The CSM in \cref{fig:csm-ing} implements $H_{\operatorname{ing}}$.
\iftoggle{arxiv}
{
In~\cref{app:explanation-implementable-not-glob-coop},}
{
In the technical report~\cite{arxiv-version},}
we explain in detail why there is a path in $H_{\operatorname{ing}}$ for any trace of the~CSM and how to modify the example not to have final states with outgoing~transitions.
\end{example}

\begin{theorem}
\label{lm:implementability-expspace}
Checking implementability of $0$-reachable global types with sender-driven choice is in \EXPSPACE.
\end{theorem}
\begin{proof}
Let $\GG$ be a $0$-reachable global type with sender-driven choice.
We construct $H(\GG)$ from~$\GG$ and check if it is globally-cooperative.
For this, we apply the coNP-algorithm by Genest et al.~\cite{DBLP:journals/jcss/GenestMSZ06} which is based on guessing a subgraph and checking its communication graph.
If $H(\GG)$ is not globally cooperative, we know from \cref{lm:implementability-entails-glob-coop} that $\GG$ is not implementable.
If $H(\GG)$ is globally cooperative, we check safe realisability for~$H(\GG)$.
By \cref{lm:safe-realisability-glob-coop-expspace}, this is in \EXPSPACE.
If $H(\GG)$ is not safely realisable, it trivially follows that $\GG$ is not implementable.
If $H(\GG)$ is safely realisable, $\GG$ is implemented by the erasure candidate implementation with \cref{thm:erasure-candidate-impl,lm:finite-generalises-to-infinite}.
\end{proof}

\noindent Thus, the implementability problem for global types with sender-driven~choice is decidable.

\begin{corollary}
Let $\GG$ be a $0$-reachable global type with sender-driven choice.
It is decidable whether $\GG$ is implementable and there is an algorithm to obtain its implementation.
\end{corollary}

\begin{remark}[Progress]
The property deadlock freedom is sometimes also studied as \emph{progress} --- in the sense that a system should never get stuck.
For infinite executions, however, a role could starve in a non-final state by waiting for a message that is never sent~\cite[Sec.~3.2]{DBLP:journals/corr/abs-1203-0780}.
Thus, Castagna et al.~\cite{DBLP:journals/corr/abs-1203-0780} consider a stronger notion of progress (Def.~3.3: live session) which requires that every role could eventually reach a final state.
Our results also apply to this stronger notion of progress, which entails that any sent message can eventually be received.
The notion only requires it to be possible but we can ensure that no role starves in a non-final state in two ways.
First, we can impose a (strong) fairness assumption --- as Castagna et al.~\cite{DBLP:journals/corr/abs-1203-0780}.
Second, we can require that every loop branch contains at least all roles that occur in interactions of any path with which the protocol can~finish.
\end{remark}

\subparagraph*{The Odd Case of Infinite Loops Without Exits.}

In theory, one can think of protocols for which the $0$-Reachability-Assumption (p.\pageref{assumption:can-always-finish}) does not hold.
They would simply recurse indefinitely and can never terminate.
This allows interesting behaviour like two sets of roles that do not interact with each other as the following example shows.

\begin{example}
\label{ex:independent-pairs}
 Consider the following global type:
 $ \GG =
 \mu t. \,
 \msgFromTo{\procA}{\procB}{\val}. \,
 \msgFromTo{\procC}{\procD}{\val}. \,
 t
 $.
 This is basically the protocol that consists only of the first loop of $H_{\operatorname{ing}}$ (\cref{ex:implementable-not-glob-coop}).
 It describes an infinite execution with two pairs of roles that independently send and receive messages.
 This can be implemented in an infinite setting but the loop can never be exited due to the lack of synchronisation, breaking protocol fidelity upon termination.
\end{example}

\subparagraph*{Expressiveness of Local Types.}

Local types
also have a distinct expression for termination:~$0$.
Thus, if one considers the FSM of a local type, every final state has no outgoing transition.
Our proposed algorithm might yield FSMs for which this is not the case.
However, the language of such an FSM cannot be represented as local type since both our construction and FSMs for local types are deterministic.
The latter are also ancestor-recursive, free of intermediate recursion, non-merging and dense (\cref{prop:shape-of-sem-fsms}).
For FSMs from our procedure, this is not the case but the ones without final states with outgoing transitions could possibly be transformed to local types, making subtyping techniques applicable.
One could also study subtyping for FSMs as local specifications.
We leave both for future work.

\subparagraph*{On Lower Bounds for Implementability.}

For general globally-cooperative HMSCs, i.e., that are not necessary the encoding of a global type, safe realisability is \EXPSPACE-hard \cite{DBLP:journals/tcs/Lohrey03}.
This hardness result does not carry over for the HMSC encoding $H(\GG)$ of a global type~$\GG$.
The construction exploits that HMSCs do not impose any restrictions on choice.
Global types, however, require every branch to be chosen by a single sender.
 \section{MSC Techniques for MST Verification}
\label{sec:msc-techniques-for-mst}

In the previous section, we generalised results from the MSC literature to show decidability of the implementability problem for global types from MSTs, yielding an \mbox{EXPSPACE}-algorithm.
In this section, we consider further restrictions on HMSCs to obtain algorithms with better complexity for global types.
First, we transfer the algorithms for $\mathcal{I}$-closed HMSCs, which requires an HMSC not to exhibit certain anti-patterns of communication, to global types.
Second, we explain approaches for HMSCs that introduced the idea of choice to HMSCs and a characterisation of implementable MSC languages.
Third, we present a variant of the implementability problem.
It can make unimplementable global types implementable without changing a protocol's structure. From now on, we may use the term implementability for HMSCs instead of safe realisability.

\paragraph*{$\pmb{\mathcal{I}}$-closed Global Types}

For globally-cooperative HMSCs, the implementability problem is \EXPSPACE-complete.
The membership in \EXPSPACE was shown by reducing the problem to implementability of $\mathcal{I}$-closed HMSCs~\cite[Thm.~3.7]{DBLP:journals/tcs/Lohrey03}.
These require the language of an HMSC to be closed with regard to an independence relation $\mathcal{I}$, where, intuitively, two interactions are independent if there is no role which is involved in both.
Implementability for $\mathcal{I}$-closed HMSCs is \PSPACE-complete~\cite[Thm.~3.6]{DBLP:journals/tcs/Lohrey03}.
As for the \EXPSPACE-hardness for globally-cooperative HMSCs, the \PSPACE-hardness exploits features that cannot be modelled with global types and there might be algorithms with better worst-case complexity.

We adapt the definitions~\cite{DBLP:journals/tcs/Lohrey03} to the MST setting.
These consider atomic BMSCs, which are BMSCs that cannot be split further.
With the HMSC encoding for global types, it is straightforward that atomic BMSCs correspond to individual interactions for global types.
Thus, we define the independence relation $\mathcal{I}$ on the alphabet~$\AlphSync$.

\begin{definition}[Independence relation $\mathcal{I}$]
\label{def:independence-relation}
We define the independence relation $\mathcal{I}$ on
$\AlphSync$\emph{:}
\vspace{-2ex}
\[
 \mathcal{I}
    \is
 \set{(\msgFromTo{\procA}{\procB}{\val},
       \msgFromTo{\procC}{\procD}{\val'})
       \mid
       \set{\procA, \procB} \inters \set{\procC, \procD} = \emptyset)}.
\]
We lift this to words, \ie
$
 \set{(u.\,
       x_1.\, x_2.\,
       w,
       u.\,
       x_2.\, x_1.\,
       w)
       \mid
       u, w \in \AlphSync^*
       \text{ and }
       (x_1, x_2) \in \mathcal{I}
     }
$,
and obtain an equivalence relation $\equiv_{\mathcal{I}}$
as its transitive and reflexive closure.
We define its closure for language $L \subseteq \AlphSync^*$\emph{:}
$
 \indeprellang(L)
    \is
 \set{
    u \in \AlphSync^* \mid
    \exists
    w \in L \text{ with }
    u \equiv_\mathcal{I} w
 }
$.
\end{definition}

\begin{definition}[$\mathcal{I}$-closed global types]
Let $\GG$ be a global type $\GG$.
We say $\GG$ is $\mathcal{I}$-closed if
$\lang_{\fin}(\semglobalsync(\GG)) = \indeprellang(\lang_{\fin}(\semglobalsync(\GG)))$.
\end{definition}

Note that $\mathcal{I}$-closedness is defined on the state machine $\semglobalsync(\GG)$ of $\GG$ with alphabet $\AlphSync$ and not on its semantics $\lang(\GG)$ with alphabet $\AlphAsync$.

\begin{example}
The global type $\GG_{\TBPOR}$ is $\mathcal{I}$-closed.
Buyer $\buyerA$ is involved in every interaction.
Thus, for every two consecutive interactions, there is a role that is involved in both.
\end{example}

\begin{algorithm}[Checking if $\GG$ is $\mathcal{I}$-closed]
\label{alg:checking-I-closedness}
Let $\GG$ be a global type $\GG$.
We construct the state machine $\semglobalsync(\GG)$.
We need to check every consecutive occurrence of elements from $\AlphSync$ for words from $\lang(\semglobalsync(\GG))$.
For binder states,
incoming and outgoing transition labels are always $\emptystring$.
This is why we slightly modify the state machine but preserve its language.
We remove all variable states and rebend their only incoming transition to the state their only outgoing transition leads~to.
In addition, we merge binder states with their only successor.
For every state $q$ of this modified state machine, we consider the labels $x, y \in \AlphSync$ of every combination of incoming and outgoing transition of $q$.
We check if $x \equiv_{\mathcal{I}} y$.
If this is true for all $x$ and $y$, we return \emph{true}.
If~not, we return \emph{false}.
\end{algorithm}

\vspace{-1ex}
\begin{restatable}{lemma}{correctnessCheckingIClosedness}\label{lm:correctness-checking-I-closedness}
A global type $\GG$ is $\mathcal{I}$-closed iff \cref{alg:checking-I-closedness} returns \emph{true}.
\end{restatable}
\vspace{-1ex}

The proof can be found in\iftoggle{arxiv}
{
\cref{proof:correctness-checking-I-closedness}.
}
{
the technical report~\cite{arxiv-version}.
}
This shows that the presented algorithm can be used to check $\mathcal{I}$-closedness.
The algorithm considers every state and all combinations of transitions leading to and from~it.

\begin{proposition}
For global type $\GG$, checking if $\GG$ is $\mathcal{I}$-closed is in $\bigO{\card{\GG}^2}$.
\end{proposition}

The tree-like shape of $\semglobalsync(\GG)$ might suggest that this check can be done in linear time.
However, this example shows that recursion can lead to a quadratic number of~checks.

\begin{example}
Let us consider the following global type for some $n \in \Nat$.

\begin{minipage}{0.41\textwidth}
\hspace{-5ex}
{
\scriptsize
$
 \mu t.\,
 +
 \begin{cases}
  \msgFromTo{\procA}{\procB_0}{\val_0}.\,
  \msgFromTo{\procB_0}{\procC_0}{\val_0}.\,
  \msgFromTo{\procC_0}{\procD_0}{\val_0}.\, 0 \\
  \msgFromTo{\procA}{\procB_1}{\val_1}.\,
  \msgFromTo{\procB_1}{\procC_1}{\val_1}.\,
  \msgFromTo{\procC_1}{\procD_1}{\val_1}.\, t \\
  \quad { \vdots} \\
  \msgFromTo{\procA}{\procB_n}{\val_n}.\,
  \msgFromTo{\procB_n}{\procC_n}{\val_n}.\,
  \msgFromTo{\procC_n}{\procD_n}{\val_n}.\, t
 \end{cases}
$
}
\end{minipage}
\begin{minipage}{0.54\textwidth}
\vspace{1ex}
It is obvious that
$
  (\msgFromTo{\procA}{\procB_i}{\val_i},
  \msgFromTo{\procB_i}{\procC_i}{\val_i})
  \notin \mathcal{I}
$
and
$
  (\msgFromTo{\procB_i}{\procC_i}{\val_i},
  \msgFromTo{\procC_i}{\procD_i}{\val_i})
  \notin \mathcal{I}
$
for every~$i$.
Because of the recursion, we need to check
if
\mbox{$
  (\msgFromTo{\procC_i}{\procD_i}{\val_i},
  \msgFromTo{\procA}{\procB_j}{\val_j})
$}
is in
$\mathcal{I}$ for every $0 \neq i \neq j$.
This might lead to a quadratic number of checks.
\end{minipage}
\end{example}

\noindent If a global type $\GG$ is $\mathcal{I}$-closed, we can apply the results for its $\mathcal{I}$-closed HMSC encoding~$H(\GG)$, for which checking implementability is in \mbox{PSPACE}.
With \cref{cor:erasure-implements-global-type}, the projection by erasure implements~$\GG$.

\begin{corollary}
Checking implementability of $0$-reachable, $\mathcal{I}$-closed global types with sender-driven choice is in \mbox{PSPACE}.
\end{corollary}

\begin{example}\label{ex:implementable-but-not-I-closed}
This implementable global type is not $\mathcal{I}$-closed:
{
$
\msgFromTo{\procA}{\procB}{\val}.\,
\msgFromTo{\procC}{\procD}{\val}.\, 0 \;.
$
}
\end{example}

\paragraph*{Detecting Non-local Choice in HMSCs}

For HMSCs, there are no restrictions on branching.
Similar to choice for global types, the idea of imposing restrictions on choice was studied for HMSCs~\cite{DBLP:conf/tacas/Ben-AbdallahL97,DBLP:conf/fase/Muccini03,DBLP:conf/fase/MooijGR05,DBLP:conf/sdl/Helouet01,DBLP:journals/jcss/GenestMSZ06}.
We refer to \cite{DBLP:conf/concur/MajumdarMSZ21} for an overview.
Here, we focus on results that seem most promising for developing algorithms to check implementability of global types with better worst-case complexity.
The work by Dan et al.~\cite{DBLP:conf/sefm/DanHC10} centers around the idea of non-local choice.
Intuitively, non-local choice yields scenarios that make it impossible to implement the language.
In fact, if a language is not implementable, there is some non-local choice.
Thus, checking implementability amounts to checking non-local choice freedom.
For this definition, they showed insufficiency of Baker's condition~\cite{DBLP:conf/sigsoft/BakerBJKTMB05} and reformulated the closure conditions for safe realisability by Alur et al.~\cite{DBLP:journals/tse/AlurEY03}.
In particular, they provide a definition that is based on projected words of a language in contrast to explicit choice.
While it is straightforward to check their definition for finite collections of $k$ BMSCs with $n$ events in $\bigO{k^2 \cdot \card{\Procs} + n \cdot \card{\Procs}}$,
it is unclear how to check their condition for languages with infinitely many elements.
The design of such a check is far from trivial as their definition does not give any insight about local behaviour and their algorithm heavily relies on the finite nature of finite collections of BMSCs.

\paragraph*{Payload Implementability}

A deadlock-free CSM implements a global type if their languages are precisely the same.
In the HMSC domain, a variant of the implementability problem has been studied.
Intuitively, it allows to add fresh data to the payload of an existing message and protocol fidelity allows to omit the additional payload data.
This allows to add synchronisation messages to existing interactions and can make unimplementable global types implementable while preserving the structure of the protocol.
It can also be used if a global type is rejected by a projection operator or the run time of the previous algorithms is not acceptable.

\begin{definition}[Payload implementability]
Let $L$ be a language with message alphabet~$\MsgVals_1$.
We say that $L$ is \emph{payload implementable} if
there is a message alphabet $\MsgVals_2$
for a deadlock-free CSM~$\CSM{A}$ with
$A_\procA$ over
$\set{ \snd{\procA}{\procB}{\val}, \rcv{\procB}{\procA}{\val} \mid \procB \in \Procs,\; \val \in \MsgVals_1 \times \MsgVals_2}$
such that its language is the same when projecting onto the message alphabet $\MsgVals_1$, i.e.,
$
 \interswaplang(L)
    =
 \lang(\CSM{A}) \wproj_{\MsgVals_1}
$,
where
$
    (\snd{\procA}{\procB}{(\val_1, \val_2)})
        \wproj_{\MsgVals_1}
    \is
    \snd{\procA}{\procB}{\val_1}
    \text{ and }
    (\rcv{\procA}{\procB}{(\val_1, \val_2)})
        \wproj_{\MsgVals_1}
    \is
    \rcv{\procA}{\procB}{\val_1}
$ and is lifted to words and languages as expected.
\end{definition}

The finite language $\lang_{\fin}(H)$ of a \emph{local} HMSC $H$ is always payload implementable with a deadlock-free CSM of linear size.

\begin{definition}[Local HMSCs \cite{DBLP:journals/jcss/GenestMSZ06}]
Let
$H = (V, \edges, \vertexA^I\negmedspace, V^T\negmedspace\!, \mu)$
be an HMSC.
We call $H$ \emph{local} if
$\mu(\vertexA^I)$ has a unique minimal event and
there is a function $\operatorname{root} \from V \to \Procs$
such that
for every $(\vertexA, \vertexB) \in \edges$, it holds that $\mu(\vertexB)$ has a unique minimal event $e$ and $e$ belongs to $\operatorname{root}(\vertexA)$, i.e.,
for
$\mu(\vertexB) = (\eventnodes, p, f, l, (\leq_\procA)_{\procA \in \Procs }) $, we have that
$p(e) = \operatorname{root}(\vertexA)$ and $e \leq e'$ for every $e' \in \eventnodes$.
\end{definition}

\begin{proposition}[Prop.~21~\cite{DBLP:journals/jcss/GenestMSZ06}]
For any local HMSC $H\negthinspace$, $\lang_{\fin}(H)$ is payload implementable.
\end{proposition}

With \cref{lm:finite-generalises-to-infinite}, we can use the implementation of a local $H(\GG)$ for a $0$-reachable global type~$\GG$.

\begin{corollary}
Let $\GG$ be a $0$-reachable for which $H(\GG)$ is local.
Then, $\GG$ can be implemented with a CSM of linear size.
\end{corollary}

The algorithm to construct a deadlock-free CSM \cite[Sec.~5.2]{DBLP:journals/jcss/GenestMSZ06} suggests that the BMSCs for such HMSCs need to be maximal -- in the sense that any vertex with a single successor is collapsed with its successor.
If this was not the case, the result would claim that the language of the following global type is payload implementable:
{ \scriptsize $
 \mu t. \,
 +
 \begin{cases}
 \msgFromTo{\procA}{\procB}{\val_1}. \,
 \msgFromTo{\procC}{\procD}{\val_2}. \, t
 \\
 \msgFromTo{\procA}{\procB}{\val_3}. \, 0
 \end{cases}
$. }
However, is is easy to see that it is not payload implementable since there is no interaction between~$\procA$, which decides whether to stay in the loop or not, and $\procC$.
Thus, we cannot simply check whether $H(\GG)$ is local.
In fact, it would always be.
Instead, we first need to minimise it and then check whether it is local.
If we collapse the two consecutive vertices with independent pairs of roles in this example, the HMSC is not local.
The representation of the HMSC matters which shows that local as property is rather a syntactic than a semantic~notion.

\begin{algorithm}[Checking if $H(\GG)$ is local -- directly on $\semglobalsync(\GG)$]
\label{alg:checking-local}
Let $\GG$ be a global type $\GG$.
We consider the finite trace~$w'$ of every longest branch-free, loop-free and non-initial run in the state machine $\semglobalsync(\GG)$.
We split the (synchronous) interactions into asynchronous events: $w = \SyncToAsync(w') = w_1 \ldots w_n$.
We need to check if there is $u \interswap w$ with $u = u_1 \ldots u_n$ such that $u_1 \neq w_1$.
For this, we can construct an MSC for $w$ \cite[Sec.~3.1]{DBLP:journals/fuin/GenestKM07} and check if there is a single minimal event.
This works because MSCs are closed under $\interswap$~\cite[Lm.~5]{DBLP:journals/corr/abs-2209-10328}.
If the MSC of every trace $w'$ has a single minimal event, we return \emph{true}.
If~not, we return \emph{false}.
\end{algorithm}

It is straightforward that this mimics the corresponding check for the HMSC $H(\GG)$ and, with similar modifications as for \cref{alg:checking-I-closedness}, the check can be done in $\bigO{\card{\GG}}$.

\begin{proposition}For a global type $\GG$, \cref{alg:checking-local} returns \emph{true} iff $H(\GG)$ is local.
\end{proposition}

Ben{-}Abdallah and Leue~\cite{DBLP:conf/tacas/Ben-AbdallahL97} introduced local-choice HMSCs, which are as expressive as local HMSCs.
Their condition also uses a $\operatorname{root}$-function and minimal events but quantifies over paths.
Every local HMSC is a local-choice HMSC and
every local-choice HMSC can be translated to a local HMSC that accepts the same language with a quadratic blow-up~\cite{DBLP:journals/jcss/GenestMSZ06}.
It is straightforward to adapt the \cref{alg:checking-local} to check if a global type is local-choice. If this is the case, we translate the protocol and use the implementation for the translated protocol.

 \section{Implementability with Intra-role Reordering}
\label{intraswap-implementability-undecidable}

In this section, we introduce a generalisation of the implementability problem that relaxes the total event order for each role. We prove that this generalisation is undecidable in general.

\paragraph*{A Case for More Reordering}

\noindent
From the perspective of a single role, each word in its language consists of a sequence of send and receive events.
Choice in global types happens by sending (and not by receiving).
Because of this, one can argue that a role should be able to receive messages from different senders in any order between sending two messages. In practice, receiving a message can induce a task with non-trivial computation that our model does not account for.
Therefore, such a reordering for a sequence of receive events can have outsized performance benefits.
In addition, there are global types that can be implemented with regard to this generalised relation even if no (standard) implementation exists.

\begin{example}[Example for intra-role reordering]
\label{ex:task-coordinator}
Let us consider a global type where a central coordinator $\procA$ distributes independent tasks to different roles in rounds:

\vspace{-3ex}
{ \scriptsize
\[
 \GG_{\operatorname{TC}}
    \is
 \mu t.\,
 +
 \begin{cases}
 \msgFromTo{\procA}{\procB_1}{\operatorname{task}}.\,
 \ldots \,
 \msgFromTo{\procA}{\procB_n}{\operatorname{task}}.\,
 \msgFromTo{\procB_1}{\procA}{\operatorname{result}}.\,
 \ldots \,
 \msgFromTo{\procB_n}{\procA}{\operatorname{result}}.\, t
 \\
 \msgFromTo{\procA}{\procB_1}{\operatorname{done}}.\,
 \ldots \,
 \msgFromTo{\procA}{\procB_n}{\operatorname{done}}.\, 0
 \end{cases}
 .
\]
}
\hspace{-1.5ex} Since all tasks in each round are independent, $\procA$ can benefit from receiving the results in the order they arrive instead of busy-waiting.
\end{example}

We generalise the indistinguishability relation $\interswap$ accordingly. 

\begin{definition}[Intra-role indistinguishability relation $\intraswap$]
We define a family of \emph{intra-role indistinguishability relations}
${\intraswap_i} \subseteq \Sigma^* \times \Sigma^*$, for $i\geq 0$
as follows.
For all $w, u \in\Sigma^*$, $w \interswap_i u$ entails $w \intraswap_i u$.
For $i=1$, we define:
if $\procB ≠ \procC$, then
{ \small
$
 w.\rcv{\procB}{\procA}{\val}.\rcv{\procC}{\procA}{\val'}.u
 \; \intraswap_{1} \;
 w.\rcv{\procC}{\procA}{\val'}.\rcv{\procB}{\procA}{\val}.u .
$
}
Based on this, we define $\intraswap$ analogously to $\interswap$.
Let $w$, $w'$, $w''$ be words s.t.~$w \intraswap_1 w'$ and $w' \intraswap_i w''$ for some~$i$.
Then $w \intraswap_{i+1} w''$.
We define $w \intraswap u$ if $w \intraswap_n u$ for some $n$.
It is straightforward that $\intraswap$ is an equivalence relation.
Define $u \preceq_\intraswap v$ if there is $w\in\Sigma^*$ such that $u.w \intraswap v$.
Observe that $u \interswap v$ iff
$u \preceq_\intraswap v$ and $v \preceq_\intraswap u$.
We extend $\intraswap$ to infinite words and languages as for $\interswap$.
\end{definition}

\begin{definition}[Implementability w.r.t.\ $\intraswap$]
\label{def:intra-role-implementability}
A global type $\GG$ is \emph{implementable with regard to $\intraswap$} if there exists a deadlock-free CSM $\CSM{A}$ such that
(i)~$\lang(\GG)
            \subseteq
        \intraswaplang(\lang(\CSM{A}))$ and
 (ii) $\intraswaplang(\lang(\GG))
            =
        \intraswaplang(\lang(\CSM{A}))$.
We say that $\CSM{A}$ $\intraswap$-implements~$\GG$.\end{definition}

In this section, we emphasise the indistinguishability relation, e.g., $\intraswap$-implementable.
We could have also followed the definition of $\interswap$-implementability and required
$\intraswaplang(\lang(\GG)) = \lang(\CSM{A})$.
This, however, requires the CSM to be closed under $\intraswap$.
In general, this might not be possible with finitely many states.
In particular, if there is a loop without any send events for a role, the labels in the loop would introduce an infinite~closure if we require that
$\intraswaplang(\lang(\GG)) \wproj_{\Alphabet_\procA} = \lang(A_\procA)$.

\begin{example}
We consider a variant of $\GG_{\operatorname{TC}}$
 from \cref{ex:task-coordinator} with $n = 2$ where $\procB_1$ and~$\procB_2$ send a log message to $\procC$ after receiving the task and before sending the result back:

\vspace{-3ex}
{ \scriptsize
\[
 \GG_{\operatorname{TCLog}}
    \is
 \mu t.\,
 +
 \begin{cases}
 \msgFromTo{\procA}{\procB_1}{\operatorname{task}}.\,
 \msgFromTo{\procA}{\procB_2}{\operatorname{task}}.\,
 \msgFromTo{\procB_1}{\procC}{\operatorname{log}}.\,
 \msgFromTo{\procB_2}{\procC}{\operatorname{log}}.\,
 \msgFromTo{\procB_1}{\procA}{\operatorname{result}}.\,
 \msgFromTo{\procB_2}{\procA}{\operatorname{result}}.\, t
 \\
 \msgFromTo{\procA}{\procB_1}{\operatorname{done}}.\,
 \msgFromTo{\procA}{\procB_2}{\operatorname{done}}.\, 0
 \end{cases}
 .
\]
}
\hspace{-1.6ex} There is no FSM for $\procC$ that precisely accepts
$\intraswaplang(\lang(\GG_{\operatorname{TCLog}})) \wproj_{\Alphabet_\procC}$.
If we rely on the fact that $\procB_1$ and $\procB_2$ send the same number of $\operatorname{log}$-messages to $\procC$, we can use an FSM $A_\procC$ with a single state (both initial and final) with two transitions: one for the $\operatorname{log}$-message from $\procB_1$ and $\procB_2$ each, that lead back to the only state.
For this, it holds that
$\intraswaplang(\lang(\GG_{\operatorname{TCLog}})) \wproj_{\Alphabet_\procC} \subseteq \lang(A_\procC)$.
If we cannot rely on this, the FSM would need to keep track of the difference, which can be unbounded and thus not recognisable by an FSM.
\end{example}

This is why we chose a more permissive definition which is required to cover at least as much as specified in the global type (i) and the $\intraswap$-closure of both are the same (ii).

It is trivial that any $\interswap$-implementation for a global type does also $\intraswap$-implement it.

\begin{proposition}
Let $\GG$ be a global type that is $\interswap$-implemented by the CSM~$\CSM{A}$.
Then, $\CSM{A}$ also $\intraswap$-implements $\GG$.
\end{proposition}

For instance,
the erasure candidate implementation is a $\interswap$-implementation as well as a $\intraswap$-implementation for  the task coordination protocol
$\GG_{\operatorname{TC}}$
from \cref{ex:task-coordinator}.
Still, \mbox{$\intraswap$-implementability} gives more freedom and allows to consider all possible combinations of arrivals of results.
In addition, $\intraswap$-implementability renders some global types implementable which would not be~otherwise.
For instance, those with a role that would need to receive different sequences, related by $\intraswap$ though, in different branches it cannot distinguish (yet).

\begin{example}[$\intraswap$-implementable but not $\interswap$-implementable]
Let us consider the following global type:
{
\small
$
(
 \msgFromTo{\procA}{\procB}{l}.\,
 \msgFromTo{\procA}{\procC}{m}.\,
 \msgFromTo{\procB}{\procC}{m}.\, 0
)
+
(
 \msgFromTo{\procA}{\procB}{r}.\,
 \msgFromTo{\procB}{\procC}{m}.\,
 \msgFromTo{\procA}{\procC}{m}.\, 0
)
$.
}
This cannot be $\interswap$-implemented because $\procC$ would need to know about the choice to receive the messages from $\procA$ and $\procB$ in the correct order.
However, it is $\intraswap$-implementable.
The FSMs for $\procA$ and $\procB$ can be obtained with projection by erasure.
For $\procC$, we can have an FSM that only accepts
 $\rcv{\procA}{\procC}{m}.\,
 \rcv{\procB}{\procC}{m}$
 but also an FSM which accepts
 $\rcv{\procB}{\procC}{m}.\,
 \rcv{\procA}{\procC}{m}$
 in addition.
Note that $\procC$ does not learn the choice in the second FSM even if it branches.
Hence, it would not be implementable if it sent different messages in both branches later on.
However, it could still learn by receiving and, afterwards, send different messages.
\end{example}

\paragraph*{Implementability with Intra-role Reordering is Undecidable}
\label{sec:pc-undec}

\noindent
Unfortunately, checking implementability with regard to $\intraswap$ for global types (with directed choice) is undecidable.
Intuitively, the reordering allows roles to drift arbitrarily far apart as the execution progresses which makes it hard to learn which choices were~made.

We reduce the \emph{Post Correspondence Problem} (PCP) \cite{NO-DBLP:Post1946AVO} to the problem of checking implementability with regard to $\intraswap$.
An instance of PCP over an alphabet~$\Delta$, $\card{\Delta} > 1$, is given by two finite lists $(u_1, u_2, \ldots, u_n)$ and $(v_1,v_2,\ldots,v_n)$ of finite words over~$\Delta$, also called tile sets.
A solution to the instance is a sequence of indices $(i_{j})_{1\leq j\leq k}$ with $k\geq 1$ and $1\leq i_{j}\leq n$ for all $1\leq j \leq k$,
such that $u_{i_{1}}\ldots u_{i_{k}}=v_{i_{1}}\ldots v_{i_{k}}$.
To be precise, we present a reduction from the modified~PCP (MPCP) \cite[Sec.~5.2]{DBLP:books/daglib/0086373}, which is also undecidable.
It simply requires that a match starts with a specific pair --- in our case we choose the pair with index $1$.
It is possible to directly reduce from PCP but the reduction from MPCP is more concise.
Intuitively, we require that the solution starts with the first pair so there exists no trivial solution and choosing a single pair is more concise than all possible ones.
Our encoding is the following global type where $x \in \set{u, v}$,
$[x_i]$ denotes a sequence of message interactions with message $x_{i}[1], \ldots, x_{i}[k]$ each for
$x_i$ of length $k$,
message $\textit{c-}x$ indicates choosing tile set~$x$, and
message $\textit{ack-}x$ indicates acknowledging the tile set $x$:

\begin{figure}[t]

\centering
\def\svgwidth{0.75\textwidth}
\begingroup \makeatletter \providecommand\color[2][]{\errmessage{(Inkscape) Color is used for the text in Inkscape, but the package 'color.sty' is not loaded}\renewcommand\color[2][]{}}\providecommand\transparent[1]{\errmessage{(Inkscape) Transparency is used (non-zero) for the text in Inkscape, but the package 'transparent.sty' is not loaded}\renewcommand\transparent[1]{}}\providecommand\rotatebox[2]{#2}\newcommand*\fsize{\dimexpr\f@size pt\relax}\newcommand*\lineheight[1]{\fontsize{\fsize}{#1\fsize}\selectfont}\ifx\svgwidth\undefined \setlength{\unitlength}{774.58717166bp}\ifx\svgscale\undefined \relax \else \setlength{\unitlength}{\unitlength * \real{\svgscale}}\fi \else \setlength{\unitlength}{\svgwidth}\fi \global\let\svgwidth\undefined \global\let\svgscale\undefined \makeatother \begin{picture}(1,0.34530324)\lineheight{1}\setlength\tabcolsep{0pt}\put(0,0){\includegraphics[width=\unitlength,page=1]{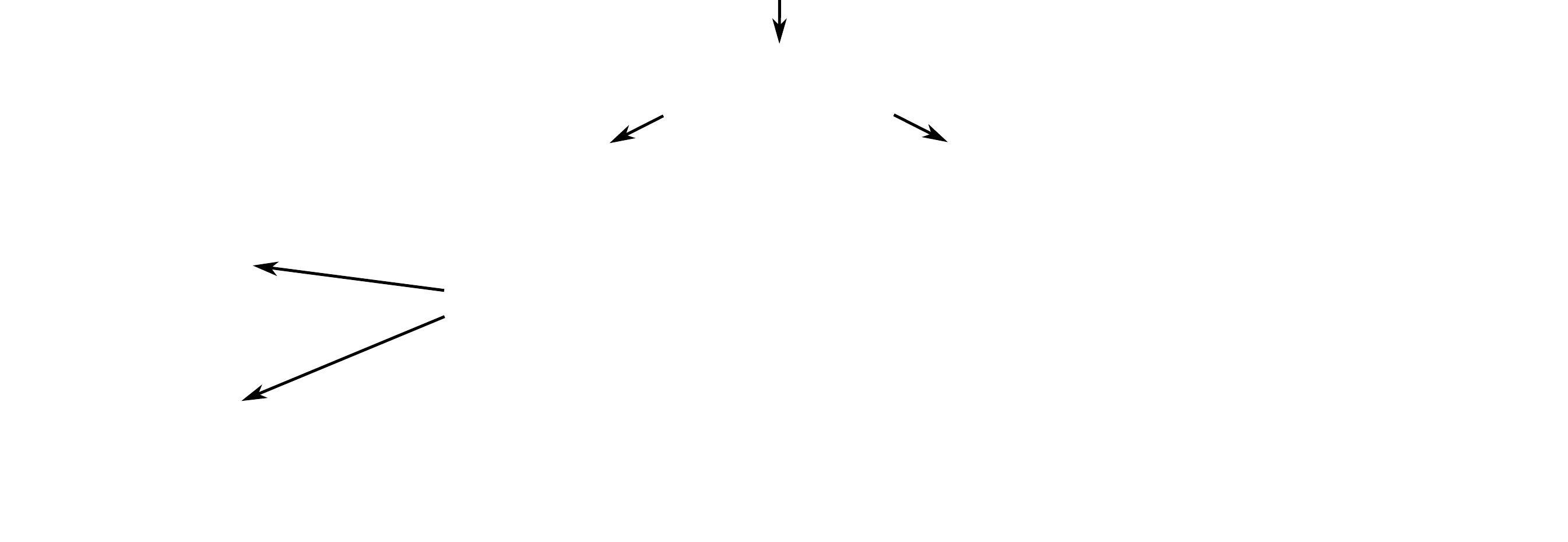}}\put(0.42748298,0.29847817){\color[rgb]{0,0,0}\makebox(0,0)[lt]{\lineheight{1.25}\smash{\begin{tabular}[t]{l}{\scriptsize $\procA$}\end{tabular}}}}\put(0,0){\includegraphics[width=\unitlength,page=2]{figs/hmsc-mpcp-encoding.pdf}}\put(0.2600398,0.27861764){\color[rgb]{0,0,0}\makebox(0,0)[lt]{\lineheight{1.25}\smash{\begin{tabular}[t]{l}{\tiny $\textit{c-}u$}\end{tabular}}}}\put(0.3156623,0.22173455){\color[rgb]{0,0,0}\makebox(0,0)[lt]{\lineheight{1.25}\smash{\begin{tabular}[t]{l}{\tiny $[u_1]$}\end{tabular}}}}\put(0.26876969,0.25623815){\color[rgb]{0,0,0}\makebox(0,0)[lt]{\lineheight{1.25}\smash{\begin{tabular}[t]{l}{\tiny $1$}\end{tabular}}}}\put(0.33316407,0.24034016){\color[rgb]{0,0,0}\makebox(0,0)[lt]{\lineheight{1.25}\smash{\begin{tabular}[t]{l}{\tiny $1$}\end{tabular}}}}\put(0,0){\includegraphics[width=\unitlength,page=3]{figs/hmsc-mpcp-encoding.pdf}}\put(0.08737807,0.03310352){\color[rgb]{0,0,0}\makebox(0,0)[lt]{\lineheight{1.25}\smash{\begin{tabular}[t]{l}{\tiny $[u_1]$}\end{tabular}}}}\put(0.04048539,0.06729812){\color[rgb]{0,0,0}\makebox(0,0)[lt]{\lineheight{1.25}\smash{\begin{tabular}[t]{l}{\tiny $1$}\end{tabular}}}}\put(0.10487985,0.05164879){\color[rgb]{0,0,0}\makebox(0,0)[lt]{\lineheight{1.25}\smash{\begin{tabular}[t]{l}{\tiny $1$}\end{tabular}}}}\put(0,0){\includegraphics[width=\unitlength,page=4]{figs/hmsc-mpcp-encoding.pdf}}\put(0.10563762,0.17445457){\color[rgb]{0,0,0}\makebox(0,0)[lt]{\lineheight{1.25}\smash{\begin{tabular}[t]{l}{\tiny \doneshort}\end{tabular}}}}\put(0.02975806,0.14078956){\color[rgb]{0,0,0}\makebox(0,0)[lt]{\lineheight{1.25}\smash{\begin{tabular}[t]{l}{\tiny $\textit{ack-}u$}\end{tabular}}}}\put(0,0){\includegraphics[width=\unitlength,page=5]{figs/hmsc-mpcp-encoding.pdf}}\put(0.3146029,0.03222686){\color[rgb]{0,0,0}\makebox(0,0)[lt]{\lineheight{1.25}\smash{\begin{tabular}[t]{l}{\tiny $[u_n]$}\end{tabular}}}}\put(0.26871023,0.06682158){\color[rgb]{0,0,0}\makebox(0,0)[lt]{\lineheight{1.25}\smash{\begin{tabular}[t]{l}{\tiny $n$}\end{tabular}}}}\put(0.17904244,0.05279385){\color[rgb]{0,0,0}\makebox(0,0)[lt]{\lineheight{1.25}\smash{\begin{tabular}[t]{l}\textit{...}\end{tabular}}}}\put(0.33310464,0.05148574){\color[rgb]{0,0,0}\makebox(0,0)[lt]{\lineheight{1.25}\smash{\begin{tabular}[t]{l}{\tiny $n$}\end{tabular}}}}\put(0.4904877,0.29866708){\color[rgb]{0,0,0}\makebox(0,0)[lt]{\lineheight{1.25}\smash{\begin{tabular}[t]{l}{\scriptsize $\procB$}\end{tabular}}}}\put(0.55517307,0.29732859){\color[rgb]{0,0,0}\makebox(0,0)[lt]{\lineheight{1.25}\smash{\begin{tabular}[t]{l}{\scriptsize $\procC$}\end{tabular}}}}\put(0,0){\includegraphics[width=\unitlength,page=6]{figs/hmsc-mpcp-encoding.pdf}}\put(0.6364881,0.27948344){\color[rgb]{0,0,0}\makebox(0,0)[lt]{\lineheight{1.25}\smash{\begin{tabular}[t]{l}{\tiny $\textit{c-}v$}\end{tabular}}}}\put(0.69219289,0.22227469){\color[rgb]{0,0,0}\makebox(0,0)[lt]{\lineheight{1.25}\smash{\begin{tabular}[t]{l}{\tiny $[v_1]$}\end{tabular}}}}\put(0.6449834,0.2562057){\color[rgb]{0,0,0}\makebox(0,0)[lt]{\lineheight{1.25}\smash{\begin{tabular}[t]{l}{\tiny $1$}\end{tabular}}}}\put(0.70937696,0.24087031){\color[rgb]{0,0,0}\makebox(0,0)[lt]{\lineheight{1.25}\smash{\begin{tabular}[t]{l}{\tiny $1$}\end{tabular}}}}\put(0,0){\includegraphics[width=\unitlength,page=7]{figs/hmsc-mpcp-encoding.pdf}}\put(0.9244207,0.03334701){\color[rgb]{0,0,0}\makebox(0,0)[lt]{\lineheight{1.25}\smash{\begin{tabular}[t]{l}{\tiny $[v_n]$}\end{tabular}}}}\put(0.87902821,0.06749055){\color[rgb]{0,0,0}\makebox(0,0)[lt]{\lineheight{1.25}\smash{\begin{tabular}[t]{l}{\tiny $n$}\end{tabular}}}}\put(0.94342258,0.05159228){\color[rgb]{0,0,0}\makebox(0,0)[lt]{\lineheight{1.25}\smash{\begin{tabular}[t]{l}{\tiny $n$}\end{tabular}}}}\put(0,0){\includegraphics[width=\unitlength,page=8]{figs/hmsc-mpcp-encoding.pdf}}\put(0.69219289,0.03289537){\color[rgb]{0,0,0}\makebox(0,0)[lt]{\lineheight{1.25}\smash{\begin{tabular}[t]{l}{\tiny $[v_1]$}\end{tabular}}}}\put(0.64492371,0.06679008){\color[rgb]{0,0,0}\makebox(0,0)[lt]{\lineheight{1.25}\smash{\begin{tabular}[t]{l}{\tiny $1$}\end{tabular}}}}\put(0.81802558,0.0506373){\color[rgb]{0,0,0}\rotatebox{-180}{\makebox(0,0)[lt]{\lineheight{1.25}\smash{\begin{tabular}[t]{l}\textit{...}\end{tabular}}}}}\put(0.70931766,0.05089171){\color[rgb]{0,0,0}\makebox(0,0)[lt]{\lineheight{1.25}\smash{\begin{tabular}[t]{l}{\tiny $1$}\end{tabular}}}}\put(0,0){\includegraphics[width=\unitlength,page=9]{figs/hmsc-mpcp-encoding.pdf}}\put(0.04086324,0.18843561){\color[rgb]{0,0,0}\makebox(0,0)[lt]{\lineheight{1.25}\smash{\begin{tabular}[t]{l}{\tiny \doneshort}\end{tabular}}}}\put(0,0){\includegraphics[width=\unitlength,page=10]{figs/hmsc-mpcp-encoding.pdf}}\put(0.10563762,0.15122597){\color[rgb]{0,0,0}\makebox(0,0)[lt]{\lineheight{1.25}\smash{\begin{tabular}[t]{l}{\tiny \doneshort}\end{tabular}}}}\put(0,0){\includegraphics[width=\unitlength,page=11]{figs/hmsc-mpcp-encoding.pdf}}\put(0.94454514,0.17216115){\color[rgb]{0,0,0}\makebox(0,0)[lt]{\lineheight{1.25}\smash{\begin{tabular}[t]{l}{\tiny \doneshort}\end{tabular}}}}\put(0.86797851,0.13936641){\color[rgb]{0,0,0}\makebox(0,0)[lt]{\lineheight{1.25}\smash{\begin{tabular}[t]{l}{\tiny $\textit{ack-}v$}\end{tabular}}}}\put(0,0){\includegraphics[width=\unitlength,page=12]{figs/hmsc-mpcp-encoding.pdf}}\put(0.87977077,0.18614219){\color[rgb]{0,0,0}\makebox(0,0)[lt]{\lineheight{1.25}\smash{\begin{tabular}[t]{l}{\tiny \doneshort}\end{tabular}}}}\put(0,0){\includegraphics[width=\unitlength,page=13]{figs/hmsc-mpcp-encoding.pdf}}\put(0.94454514,0.14893255){\color[rgb]{0,0,0}\makebox(0,0)[lt]{\lineheight{1.25}\smash{\begin{tabular}[t]{l}{\tiny \doneshort}\end{tabular}}}}\put(0,0){\includegraphics[width=\unitlength,page=14]{figs/hmsc-mpcp-encoding.pdf}}\end{picture}\endgroup  

\caption{HMSC encoding $H(\GG_{\MPCP})$ of the MPCP encoding}
\label{fig:hmsc-mpcp-encoding}
\vspace{-2ex}
\end{figure}

\vspace{-2ex}
{ \scriptsize \[
\GG_{\MPCP}
    \quad \is \quad
+
\begin{cases}
G(u, \msgFromTo{\procC}{\procA}{\textit{ack-}u}. \, 0)
\\
G(v, \msgFromTo{\procC}{\procA}{\textit{ack-}v}. \, 0)
\end{cases}
\text{ with }
\]
}
\vspace{-4ex}
{
\scriptsize \[
G(x, X)
    \quad \is \quad
    \msgFromTo{\procA}{\procB}{\textit{c-}x}. \,
    \msgFromTo{\procA}{\procB}{1}. \,
    \msgFromTo{\procA}{\procC}{1}. \,
    \msgFromTo{\procB}{\procC}{[x_1]}. \,
\mu t. \,
+
    \begin{cases}
    \msgFromTo{\procA}{\procB}{1}. \,
    \msgFromTo{\procA}{\procC}{1}. \,
    \msgFromTo{\procB}{\procC}{[x_1]}. \, t
    \vspace{-1ex}
    \\
    \quad \vdots
    \\
    \msgFromTo{\procA}{\procB}{n}. \,
    \msgFromTo{\procA}{\procC}{n}. \,
    \msgFromTo{\procB}{\procC}{[x_n]}. \, t
    \\
    \msgFromTo{\procA}{\procB}{\doneshort}. \,
    \msgFromTo{\procA}{\procC}{\doneshort}. \,
    \msgFromTo{\procB}{\procC}{\doneshort}. \,
    X
\end{cases}
.
\]
}
\hspace{-1.5ex} The HMSC encoding $H(\GG_{\MPCP})$ is illustrated in~\cref{fig:hmsc-mpcp-encoding}.
Intuitively, $\procC$~eventually needs to know which branch was taken in order to match $\textit{ack-}x$ with $\textit{c-}x$ from the beginning.
However, it can only know if there is no solution to the MPCP instance.
In the full proof in\iftoggle{arxiv}
{
\cref{proof:undec-intraswap},}
{
technical report~\cite{arxiv-version},}
we show that $\GG_{\MPCP}$ is $\intraswap$-implementable
iff
the MPCP instance has no~solution.

\begin{restatable}{theorem}{implementabilityIntraswapReorderingUndecidable}
\label{thm:implementability-intraswap-reordering-undecidable}
Checking implementability with regard to $\intraswap$ for $0$-reachable global types with directed choice is undecidable.
\end{restatable}

This result carries over to HMSCs if we consider safe realisability with regard to $\intraswap$.

\begin{definition}[Safe realisability with regard to $\intraswap$]
 An HMSC $H$ is said to be \emph{safely realisable with regard to $\intraswap$} if there exists a deadlock-free CSM $\CSM{A}$ such that the following holds:
 (i)~$\lang(H)
            \subseteq
        \intraswaplang(\lang(\CSM{A}))$ and
 (ii) $\intraswaplang(\lang(H))
            =
        \intraswaplang(\lang(\CSM{A}))$.
\end{definition}

\begin{corollary}
Checking safe realisability with regard to $\intraswap$ for HMSCs is undecidable.
\end{corollary}

It is obvious that a terminal vertex is reachable from every vertex in $H(\GG_{\MPCP})$.
In fact, the HMSC encoding for $\GG_{\MPCP}$ also satisfies a number of channel restrictions.
The HMSC $H(\GG_{\MPCP})$ is existentially $1$-bounded, $1$-synchronisable and half-duplex~\cite{DBLP:journals/corr/abs-2209-10328}.
For details on these channel restrictions, we refer to work by Stutz and Zufferey \cite[Sec.~3.1]{DBLP:journals/corr/abs-2209-10328}.

The MPCP encoding only works since receive events can be reordered unboundedly in an execution.
If we amended the definition of $\intraswap$ to give each receive event a budget that depletes with every reordering, this encoding would not be possible.
We leave a detailed analysis for future work.
 \section{Related Work}
\label{sec:related}

In this section, we solely cover related work which was not discussed before.

\smallskip\noindent
{\sffamily\bfseries
Multiparty Session Types.}
Session types originate in process algebra and were first introduced by Honda et al.~\cite{DBLP:conf/concur/Honda93} for binary sessions.
For systems with more than two roles, they have been extended to multiparty session types~\cite{DBLP:conf/popl/HondaYC08}.
We explained MST frameworks with classical projection operators.
Other approaches do not focus on projection but only apply ideas from MST without the need for global types~\cite{DBLP:journals/pacmpl/ScalasY19,DBLP:conf/cav/LangeY19}.

\smallskip\noindent
{\sffamily\bfseries
Completeness and Sender-driven Choice.}
Our decidability result applies to global types with sender-driven choice.
To the best of our knowledge, the work by Castagna et al.~\cite{DBLP:journals/corr/abs-1203-0780} is the only one to attempt completeness for global types with sender-driven choice.
However, their definition of completeness  is ''less demanding then other ones``~\cite[Abs.]{DBLP:journals/corr/abs-1203-0780}.
For one global type, they accept different implementations that generate different sets of traces~\cite[Def.\ 4.1 and Sec.\ 5.3]{DBLP:journals/corr/abs-1203-0780}.
Their conditions, given as inference rules, are not effective and their algorithmically checkable conditions can only exploit local information to disambiguate choices.
In contrast, Majumdar et al.~\cite{DBLP:conf/concur/MajumdarMSZ21} employ a global availability analysis but, as classical projection operator, it suffers from the shortcomings presented in this work.
For a detailed overview of MST frameworks with sender-driven choice, we refer to their work~\cite{DBLP:conf/concur/MajumdarMSZ21}.
The global types by Castellani et al.~\cite{DBLP:journals/corr/abs-2203-12876} specify send and receive events independently and allow to receive from different senders.
Dagnino et al.~\cite{DBLP:conf/coordination/DagninoGD21} consider similar global types but each term requires to send to a single receiver and to receive from a single sender upon branching though.

\smallskip\noindent
{\sffamily\bfseries
On the Synchronous Implementability Problem.}
We could not find a reference that shows decidability of the implementability problem in a synchronous setting, i.e., without channels.
Before giving a proof sketch, let us remark that there are global types that can be implemented synchronously but not asynchronously, e.g.,
{
\scriptsize
$
\msgFromTo{\procA}{\procB}{l}.\,
 \msgFromTo{\procC}{\procB}{l}.\, 0
+
\msgFromTo{\procA}{\procB}{r}.\,
 \msgFromTo{\procC}{\procB}{r}.\, 0
$
}
because $\procB$ can force the right choice by $\procC$.
We sketch how to prove decidability of the synchronous implementability problem for global types (with sender-driven choice).
One defines the synchronous semantics of CSMs and HMSCs as expected.
For global types, one uses the independence relation $\mathcal{I}$ (Def.~\ref{def:independence-relation}), which defines reasonable reorderings for synchronous events in a distributed setting, similar to the indistinguishability relation $\interswap$.
It is straightforward that the HMSC encoding $H(\hole)$ for global types~\cite{DBLP:journals/corr/abs-2209-10328} also works for the synchronous setting (cf.~Thm.~\ref{thm:correctnessMPSTtoHMSC}).
Thus, every implementation for $H(\GG)$ is also an implementation for $\GG$.
For the asynchronous setting, we used \cite[Thm.~13]{DBLP:journals/tse/AlurEY03}, which shows that the canonical candidate implementation implements an HMSC if it is implementable.
Alur et al.~\cite[Sec.\ 8]{DBLP:journals/tse/AlurEY03} also considered the synchronous setting.
They observe that both Theorem 5 and 8, basis for Theorem 13, stay valid under these modified conditions.
Together with our results, the erasure candidate implementation implements a global type if it is implementable.
Because of the synchronous semantics, its state space is finite and can be model-checked against the global type, yielding a~\mbox{PSPACE}-procedure for the synchronous implementability problem, and thus decidability.
Closest are works by Jongmans and Yoshida~\cite{DBLP:conf/esop/JongmansY20} and Glabbeek et al.~\cite{DBLP:conf/lics/GlabbeekHH21}.
Jongmans and Yoshida consider quite restrictive synchronous semantics for global types \cite[Ex.\ 3]{DBLP:conf/esop/JongmansY20} that does not allow the natural reorderings in a distributed setting, as enabled by $\mathcal{I}$, e.g.,
$\msgFromTo{\procA}{\procB}{\val}.\,
 \msgFromTo{\procC}{\procD}{\val}.\, 0$ (Ex.\ \ref{ex:implementable-but-not-I-closed})
is considered unimplementable.
Glabbeek et al.~\cite{DBLP:conf/lics/GlabbeekHH21} present a projection operator that is complete for various notions of lock-freedom, a typical liveness property, and investigate how much fairness is required for~those.

\smallskip\noindent
{\sffamily\bfseries
Subtyping and MST Extensions.}
In this work, we do not distinguish between local types and implementations but use local types directly as implementations.
Intuitively, subtyping studies how to give freedom in the implementation while preserving the correctness properties.
The intra-role indistinguishability relation~$\intraswap$, which allows to reorder receive events for a role, resembles subtyping to some extent, e.g., the work by Cutner et al.~\cite{DBLP:conf/ppopp/CutnerYV22}.
A detailed investigation of this relation is beyond the scope of this work.
For details on subtyping, we refer to work by
Chen et al.~\cite{DBLP:conf/ppdp/ChenDY14,DBLP:journals/lmcs/ChenDSY17},
Lange and Yoshida~\cite{DBLP:conf/fossacs/LangeY17}, and Bravetti et al.~\cite{DBLP:journals/tcs/BravettiCZ18}.
Various extensions to make MST verification applicable to more scenarios were studied:
for instance delegation~\cite{DBLP:conf/esop/HondaVK98,DBLP:conf/popl/HondaYC08,DBLP:journals/tcs/CastellaniDGH20},
dependent session types~\cite{DBLP:conf/ppdp/ToninhoCP11,DBLP:journals/corr/abs-1208-6483,DBLP:conf/fossacs/ToninhoY18},
parametrised session types~\cite{dblp:journals/scp/charalambidesda16,DBLP:journals/corr/abs-1208-6483},
gradual session types~\cite{DBLP:journals/jfp/IgarashiTTVW19}, or dynamic self-adaption \cite{DBLP:conf/ecoop/00020DG21}.
Context-free session types~\cite{dblp:conf/icfp/thiemannv16,dblp:journals/toplas/keizerbp22} provide a more expressive way to specify protocols. Research on fault-tolerant MSTs \cite{DBLP:journals/pacmpl/VieringHEZ21,DBLP:conf/concur/BarwellSY022} investigates ways to weaken the strict assumptions about reliable channels.

\smallskip\noindent
{\sffamily\bfseries
Communicating State Machines.}
The connection of MSTs and CSMs was studied soon after MSTs had been proposed~\cite{DBLP:conf/esop/DenielouY12}. CSMs are known to be Turing-powerful~\cite{DBLP:journals/jacm/BrandZ83}.
Decidable classes have been obtained for different semantics, e.g.,
half-duplex communication for two roles~\cite{DBLP:journals/iandc/CeceF05},
input-bounded~\cite{DBLP:conf/concur/BolligFS20}, and unreliable/lossy channels~\cite{DBLP:conf/cav/AbdullaBJ98}, as well for restricted communication topology~\cite{DBLP:journals/acta/PengP92, DBLP:conf/tacas/TorreMP08}.
Similar restrictions for CSMs are existential boundedness~\cite{DBLP:journals/fuin/GenestKM07} and synchronisability~\cite{DBLP:conf/cav/BouajjaniEJQ18,DBLP:conf/fossacs/GiustoLL20}.
It was shown that global types can only express existentially $1$-bounded, $1$-synchronisable and half-duplex communication~\cite{DBLP:journals/corr/abs-2209-10328} while Bollig et al.~\cite{DBLP:conf/concur/BolligGFLLS21} established a connection between synchronisability and MSO logic.

\smallskip\noindent
{\sffamily\bfseries
High-level Message Sequence Charts}
Globally-cooperative HMSCs were independently introduced by Morin~\cite{DBLP:conf/stacs/Morin02} as c-HMSCs.
Their communication graph is weakly connected.
The class of bounded HMSCs \cite{DBLP:conf/concur/AlurY99} requires it to be strongly connected.
Historically, it was introduced before the class of globally-cooperative HMSCs and, after the latter has been introduced, safe realisability for bounded HMSCs was also shown to be EXPSPACE-complete~\cite{DBLP:journals/tcs/Lohrey03}.
This class was independently introduced as regular HMSCs by Muscholl and Peled \cite{DBLP:conf/mfcs/MuschollP99}.
Both terms are justified:
the language generated by a \emph{regular} HMSC is regular and
every \emph{bounded} HMSC can be implemented with universally bounded channels.
In fact, a HMSC is bounded if and only if it is a globally-cooperative and it has universally bounded channels~\cite[Prop.~4]{DBLP:journals/jcss/GenestMSZ06}.
 \section{Conclusion}

We have proven decidability of the implementability problem for global types with generalised choice from MSTs --- under the mild assumption that protocols can (but do not need to) terminate.
To point at the origin for incompleteness of classical projection operators, we gave a visual explanation of the projection with various merge operators on finite state machines, which define the semantics of global and local types.
To prove decidability, we formally related the implementability problem for global types with the safe realisability problem for HMSCs.
While safe realisability is undecidable in general, we showed that implementable global types do always belong to the class of globally-cooperative HMSCs.
There are global types that are outside of this class but the syntax of global types allowed us to prove that those cannot be implemented. Another key was the extension of the HMSC results to infinite executions.
We also gave a comprehensive overview of MSC techniques and adapted some to the MST setting.
Furthermore, we introduced a performance-oriented generalisation of the implementability problem  which, however, we proved to be undecidable in general.
 
\phantomsection\label{paper-last-page}

\clearpage
\appendix

\section{Definitions for \cref{sec:mst}: \newline Multiparty Session Types}

\subsection{Semantics of Communicating State Machines} \label{app:semantics-csm}

With $\channels = \set{\channel{\procA}{\procB} \mid \procA,\procB\in \Procs, \procA\neq \procB}$, we denote the set of channels.
The set of global states of a CSM is given by $\prod_{\procA \in \Procs} Q_\procA$.
Given a global state $q$, $q_\procA$ denotes the state of $\procA$ in $q$.
A~\emph{configuration} of a CSM $\CSM{A}$ is a pair $(q, \xi)$, where $q$ is a global state and
$\xi : \channels \rightarrow \MsgVals^\infty$ is a mapping of each channel to its current content.
The initial configuration $(q_0, \xi_\emptystring)$ consists of  a global state $q_0$ where the state of each role is the initial state  $q_{0,\procA}$ of $A_\procA$ and a mapping~$\xi_\emptystring$, which maps each channel to the empty word~$\emptystring$.
A~configuration $(q, \xi)$ is said to be \emph{final} iff each individual local state $q_\procA$ is final for every $\procA$ and $\xi$ is $\xi_{\emptystring}$.

The global transition relation $\rightarrow$ is defined as follows:
\vspace{-1ex}
\begin{itemize}
\item
$(q,\xi) \xrightarrow{\snd{\procA}{\procB}{\val}} (q',\xi')$ if
$(q_\procA, \snd{\procA}{\procB}{\val}, q'_\procA)\in\delta_\procA$,
$q_\procC = q'_\procC$ for every role $\procC \neq \procA$,
$\xi'(\channel{\procA}{\procB}) =  \xi(\channel{\procA}{\procB})\cdot\val$ and $\xi'(c) = \xi(c)$ for every other channel $c\in \channels$.

\item
$(q,\xi) \xrightarrow{\rcv{\procA}{\procB}{\val}} (q',\xi')$ if
$(q_\procB, \rcv{\procA}{\procB}{\val}, q'_\procB)\in\delta_\procB$,
$q_\procC = q'_\procC$ for every role $\procC \neq \procB$,
$\xi(\channel{\procA}{\procB}) = \val\cdot \xi'(\channel{\procA}{\procB})$
and $\xi'(c) = \xi(c)$ for every other channel $c\in \channels$.

\item
$(q,\xi) \xrightarrow{\emptystring} (q',\xi)$ if
$(q_\procA, \emptystring, q'_\procA)\in\delta_\procA$ for some role
$\procA$, and
$q_\procB = q'_\procB$ for every role $\procB \neq \procA$.
\end{itemize}
\vspace{-1ex}
\noindent
A run of the CSM always starts with an initial configuration $(q_0, \xi_0)$, and is a finite or infinite sequence
$(q_0, \xi_0) \xrightarrow{w_0} (q_1, \xi_1) \xrightarrow{w_1} \ldots$
for which  $(q_i,\xi_i) \xrightarrow{w_i} (q_{i+1},\xi_{i+1})$.
The word $w_0w_1 \ldots \in\Sigma^\infty$ is said to be the \emph{trace} of the run.
A run is called maximal if it is infinite or finite and ends in a final configuration. As before, the trace of a maximal run is maximal.
The language $\lang(\CSM{A})$ of the CSM $\CSM{A}$ consists of its set of maximal traces.
A deadlock of $\CSM{A}$ is a reachable configuration without outgoing transitions that is not final.

\subsection{Semantics for Local Types}
\label{app:semantics-local-types}

\begin{definition}[Semantics for local types]
\label{def:language-local-mst}
Given a local type $L$ for role $\procA$, we index syntactic subterms as for the semantics of global types.
We construct a state machine $\semlocal(L) = (Q, \Sigma_\procA, δ, q₀, F)$ where
\vspace{-1ex}
\begin{itemize}
\item $Q$ is the set of all indexed syntactic subterms in $L$,
\item $δ$ is the smallest set containing \\
            $(
            [\IntCh_{i ∈ I} \snd{}{\procB_i}{\val_i.[L_i, k_i]}, k], \snd{\procA}{\procB_i}{\val_i},
            [L_i, k_i]
            )$ and
            $(
            [\ExtCh_{i ∈ I} \rcv{\procB_i}{}{\val_i.[L_i, k_i]}, k],
            \rcv{\procB_i}{\procA}{\val_i},
            [L_i, k_i])$
            for each $i ∈ I$,
as well as
            $([μ t. [L', k'_2], k'_1], ε, [L', k'_2])$
            and $([t, k'_3], ε, [μ t. [L', k'_2], k'_1])$,
\item $q₀ = [L, 1]$ and
$F = \set{[0, k] \mid k \text{ is an index for subterm }0}$.
\end{itemize}
\vspace{-1ex}
We define the semantics of $L$ as language of this automaton:
$\lang(L) = \lang(\semlocal(L))$.
\end{definition}
 \section{Additional Explanation for Different Merge Operators on FSMs from \cref{sec:from-global-to-local}}
\label{sec:additional-explanation}

\subparagraph*{Visual Explanation of the Parametric Projection Operator: Collapsing Erasure} 

Here, we describe \emph{collapsing erasure} more formally.
Let $\GG$ be some global type and $\procC$ be the role onto which we project.
We apply the parametric projection operator to the state machine $\semglobalsync(\GG)$.
It projects each transition label onto the respective event for role~$\procC$:
every forward transition label
$\msgFromTo{\procA}{\procB}{\val}$
turns to
$\snd{\procA}{\procB}{\val}$
if $\procC = \procA$,
$\rcv{\procA}{\procB}{\val}$
if $\procC = \procB$,
and $\emptystring$ otherwise.
Then, it collapses neutral states with a single successor: $q_{1\mid2}$ replaces two states $q_1$ and $q_2$ if
$q_1 \xrightarrow{\emptystring} q_2$ is the only forward transition for $q_1$.
In case there is only a backward transition from $q_1$ to $q_2$, the state $q_{1\mid2}$ is also final.
This accounts for loops a role is not part of.

We call this procedure collapsing erasure as it erases interactions that do not belong to a role and collapses some states.
It is common to all the presented merge operators.
This procedure yields a state machine over $\Alphabet_\procC$.
It is straightforward that it is still ancestor-recursive, free from intermediate recursion and non-merging.
However, it might not be dense.
In fact, it is not dense if $\procC$ is not involved in some choice with more than one branch.

\subparagraph*{Parametric Merge in the Visual Explanation}
The parametric projection operator applies the merge operator for these cases. Visually, these correspond precisely to the remaining neutral states (since all neutral states with a single successor have been collapsed). For instance, we have a neutral state $q_1$ with $q_1 \xrightarrow{\emptystring} q_2$ and $q_1 \xrightarrow{\emptystring} q_3$ for $q_2 \neq q_3$.
Through the parametric projection operator, the merge operator may be indirectly called recursively.
Thus, we explain the merge operators for two states (and their cones) in general.
No information is propagated when the merge operator recurses and recursion variables are never unfolded.
Thus, we can ignore backward transitions and consider the cones of $q_2$ and~$q_3$. Intuitively, we iteratively apply the merge operator from lower to higher levels. However, we might need descend again when merge is applied recursively.
Similar to the syntactic version, we do only explain the $2$-ary case but the reasoning easily lifts to the $n$-ary case.

\subparagraph*{Visual Explanation of Plain Merge} 

The plain merge is not applied recursively.
Thus, we consider $q_1$ with $q_1 \xrightarrow{\emptystring} q_2$ and $q_1 \xrightarrow{\emptystring} q_3$ for $q_2 \neq q_3$ such that $q_1$ has the lowest level for which this holds.
Hence, we can assume that each cone of $q_2$ and $q_3$ does not contain neutral states. Then, the plain merge is only defined if there is an isomorphism between the states of both cones that satisfy the following conditions:
\vspace{-3.5ex}
\begin{itemize}
 \item it preserves the transition labels and hence the kind of states, and
 \item if a state has a backward transition to a state outside of the cone, its isomorphic state has a transition to the same state
\end{itemize}
\vspace{-2.5ex}
If defined, the result is $q_1$ with its cone (and $q_2$ with its cone is removed).

\subparagraph*{Visual Explanation of Semi-full Merge}The semi-full merge applies itself recursively.
Thus, we consider two states $q_2 \neq q_3$ that shall be merged.
As before, we can assume that each cone of $q_2$ and $q_3$ does not contain neutral states. In addition to plain merge, the semi-full merge allows to merge receive states.
For these, we introduce a new receive state $q_{2\mid3}$ from which all new transitions start.
For all possible transitions from $q_2$ and $q_3$, we check if
there is a transition with the same label from the other state.
For the ones not in common, we simply add the respective transition with the state it leads to and its respective cone.
For the ones in common, we recursively check if the two states, which both transitions lead to, can be merged.
If not, the semi-full merge is undefined.
If so, we add the original transition to the state of the respective merge and keep its cone.

\subparagraph*{Visual Explanation of Full Merge}
Intuitively, the full merge simply applies the idea of the semi-full merge to another case.
For the semi-full merge, one can recursively apply the merge operator when a reception was common between two states to merge.
The full merge operator allows to descend for recursion variable binders.
 \section{Formalisation for \cref{sec:standard-implementability-decidable}: \newline Implementability for Global Types from MSTs is Decidable}

\subsection{Definitions for \cref{sec:hmscs}}
\label{app:hmscs}

\begin{definition}[Concatenation of MSCs] \label{def:concatenation-msc}
Let $M_i = (\eventnodes_i, p_i, f_i, l_i, (\leq^i_p)_{\procA\in𝓟})$ for $i \in \set{1,2}$ where $M_1$ is a BMSC and $M_2$ is an MSC
with disjoint sets of events, i.e., $\eventnodes_1\cap \eventnodes_2 = \emptyset$.
We define their \emph{concatenation} $M_1\cdot M_2$ as the MSC
$M = (\eventnodes, p, f, l, (\leq_p)_{\procA\in𝓟})$ where:
\vspace{-1ex}
\begin{itemize}
 \item $\eventnodes
 \quad \is \quad
 \eventnodes₁ \; \union \; \eventnodes₂$,
 \item $
            \text{for }
            \zeta \in \set{p, f, l}:
            \quad
            \zeta(e) \is
            \begin{cases}
              \zeta(e)         & \text{if } e \in \eventnodes₁ \\
              \zeta(e)         & \text{if } e \in \eventnodes₂
            \end{cases}
       $, and
 \item $\forall \procA \in 𝓟: \quad$
$
                \leq_{\procA}
                    \quad \is \quad
                \leq^1_{\procA}
                \; \union \;
                \leq^2_{\procA}
                \; \union \;
                \set{(e₁,e₂) \mid  \, e₁\in \eventnodes₁ \, \land \, e₂\in \eventnodes₂ \, \land  p(e_1) = p(e_2)=\procA }.
       $
\end{itemize}
\end{definition}

\begin{definition}[Language of an HMSC] \label{def:lang-hmsc}
Let $H = (V, \edges, \vertexA^I\negmedspace, V^T\negmedspace\!, \mu)$ be an HMSC.
The language of $H$ is defined as
\begin{small}
\begin{align*}
\lang(H) \is
& \;
\set{
    w
    \mid
    w \in \lang(\mu(\vertexA_1) \mu(\vertexA_2) \ldots \mu(\vertexA_n))
    \text{ with }
    \vertexA_1 = v^I
    \land
    \forall \, 0 \leq i < n: \,
    (v_i, v_{i+1}) \in \edges
    \land
    v_n \in V^T
}
\\
&
\;
\union \;
\set{
    w \mid
    w \in \lang(\mu(\vertexA_1) \mu(\vertexA_2) \ldots )
    \text{ with }
    \vertexA_1 = v^I
    \land
    \forall \, i \geq 0: \,
    (v_i, v_{i+1}) \in \edges
} \;.
\end{align*}
\end{small}
\end{definition}

\subsection{HMSC Encoding for Global Types}
\label{app:hmsc-encoding}

\begin{definition}[Encoding global types as HMSCs~\cite{DBLP:journals/corr/abs-2209-10328}]
\label{def:hmsc-encoding}
In the translation, the following notation is used:
$M_\emptyset$ is the empty BMSC ($\eventnodes = ∅$) and
$M( \msgFromTo{\procA}{\procB}{\val} )$ is the BMSC with two event nodes: $e₁$, $e₂$ such that
    $f(e₁) = e₂$,
    $l(e₁) = \snd{\procA}{\procB}{\val}$, and
    $l(e₂) = \rcv{\procA}{\procB}{\val} \,$.

Let $\GG$ be a global type, we construct an HMSC
$H(\GG) = (V,\edges,v^I,V^T,μ)$~with

\smallskip
\begin{small}
$
\begin{array}{llll}
V     = \; & \set{G' \; \mid \; G' \text{ is a subterm of } \GG } \; ∪ \;
\\ &
        \set{ (\sum_{i ∈ I} \msgFromTo{\procA}{\procB_i}{\val_i}.G_i, j) \; \mid \;
\sum_{i ∈ I} \msgFromTo{\procA}{\procB_i}{\val_i}.G_i  \text{ occurs in } \GG ∧ j∈ I }
        \phantom{some }
    \vspace{1.5ex}
        \\
\edges = \; & \set{ (μt.G', G') \; \mid \; μt.G'  \text{ occurs in } \GG }  \; ∪ \;  \set{ (t, μt.G') \; \mid \; t, μt.G'  \text{ occurs in } \GG }
        \vspace{0.5ex}
            \\
           &  ∪ \; \set{ (\sum_{i ∈ I} \msgFromTo{\procA}{\procB_i}{\val_i}.G_i, (\sum_{i ∈ I} \msgFromTo{\procA}{\procB_i}{\val_i}.G_i,j))  \; \mid \;
(\sum_{i ∈ I} \msgFromTo{\procA}{\procB_i}{\val_i}.G_i,j) ∈ V}
        \vspace{0.5ex}
          \\
          & ∪ \; \set{ ( (\sum_{i ∈ I} \msgFromTo{\procA}{\procB_i}{\val_i}.G_i, j), G_j) \; \mid \;
(\sum_{i ∈ I} \msgFromTo{\procA}{\procB_i}{\val_i}.G_i, j) ∈ V }
    \vspace{1.5ex}
          \\
v^I  = \; & \GG \qquad \;
V^T  = \;  \set{0} \qquad \;
μ(v) = \;
\begin{cases}
    M( \msgFromTo{\procA}{\procB_i}{\val_j})   & \text{if } v = (\sum_{i ∈ I} \msgFromTo{\procA}{\procB_i}{\val_i}.G_i\}, j) \\
    M_\emptyset                     & \text{otherwise}
\end{cases}
\end{array}
$
\end{small}
\end{definition}

\subsection{Proof of \cref{lm:projection-by-erasure-correct}: Projection by Erasure is Correct}
\label{proof:projection-by-erasure-correct}

\projectionByErasureCorrect*

\begin{proof}
Let $H = (V, \edges, \vertexA^I\negmedspace, V^T\negmedspace\!, \mu)$ be an HMSC.
For every $\vertexA \in V$, it is straightforward that
the construction of $\mu(v) \wproj_\procA$ yields
$\lang(\mu(v)) \wproj_{\Alphabet_\procA}
    =
 \lang(\mu(v) \wproj_\procA)$ $(1)$.
We recall that $\interswap$ does not reorder events by the same role:
$w \interswap w'$ for $w \in \Alphabet_\procA$ iff $w = w'$ $(2)$.

The following reasoning proves the claim where the first equivalence follows from the construction of the transition relation of $H \wproj_\procA$:
\vspace{-2ex}
\begin{align*}
    & w \in \lang(H \wproj_\procA) \\
\Leftrightarrow \quad
    & w = w_1 \ldots
    \text{, there is a path } \vertexA_1, \ldots \text{ in } H
    \text{ and }
    w_i \in \lang(\mu(\vertexA_i) \wproj_\procA)
    \text{ for every } i \\
\overset{(1)}{\Leftrightarrow} \quad
    & w = w_1 \ldots
    \text{, there is a path } \vertexA_1, \ldots \text{ in } H
    \text{ and }
    w_i \in \lang(\mu(\vertexA_i)) \wproj_{\Alphabet_\procA}
    \text{ for every } i \\
\overset{(2)}{\Leftrightarrow} \quad
    & w \in \lang(H) \wproj_{\Alphabet_\procA}
\end{align*}
\end{proof}

\subsection{Proof of \cref{thm:erasure-candidate-impl}: \newline
Erasure Candidate Implementation is Sufficient}
\label{proof:erasure-candidate-impl}

{
\renewcommand{\footnote}[1]{}
\erasureCandidateImpl*
}

\begin{proof}
We first use the correctness of the global type encoding (\cref{thm:correctnessMPSTtoHMSC}) to observe that $\lang_{\fin}(\GG) = \lang_{\fin}(H(\GG))$.
Theorem 13 by Alur et al.~\cite{DBLP:journals/tse/AlurEY03} states that the canonical candidate implementation implements $\lang_{\fin}(H(\GG))$ if it is implementable.
\Cref{cor:erasure-canonical} and the fact that the FSM for each role is deterministic by construction allows us to replace every $A_\procA$ from the canonical candidate implementation with the projection by erasure $H(\GG) \wproj_\procA$ for every role $\procA$ which proves the claim.
\end{proof}

\subsection{Proof of \cref{lm:finite-generalises-to-infinite}:
''Finite Implementation`` Generalises to Infinite Case for $\pmb{0}$-reachable Global Types}
\label{proof:finite-generalises-to-infinite}

{
\finiteGeneralisesToInfinite*
}

\begin{proof}
By assumption, we know that $\CSM{A}$ is deadlock-free and
$\lang_{\fin}(\CSM{A}) = \lang_{\fin}(\GG)$.We prove the claim by showing both inclusions.

\textbf{First}, we show that
$\lang_{\inf}(\CSM{A}) \subseteq \lang_{\inf}(\GG)$.
For this direction, let $w$ be a word in $\lang_{\inf}(\CSM{A})$.
We need to show that there is a run $\rho$ in $\semglobalsync(\GG)$ such that $w \preceq_\interswap^\omega \SyncToAsync(\trace(\rho))$.
Since $\GG$ is $0$-reachable, we know that for every $u \in \pref(w)$, it holds that $u \in \pref(\lang_{\fin}(\GG))$.
Thus, there exists a finite run $\rho$ (that does not necessarily end in a final state) and $u'$ such that $u.u' \interswap \trace(\rho)$.
We call $\rho$ a witness run.
Intuitively, we need to argue that every such witness run for $u$ can be extended when appending the next event $x$ from $w$ to obtain $ux$.
In general, this does not hold for every choice of witness run.
However, because of monotonicity, any run (or rather a prefix of it) for an extension $ux$ can also be used as witness run for $u$.
Thus, we make use of the idea of prophecy variables \cite{DBLP:conf/lics/AbadiL88} and assume an oracle which picks the correct witness run for every prefix $u$.
This oracle does not restrict the next possible events in any way.
From here, we apply the same idea as Majumdar et al.\ for the proof of Lemma~41~\cite{DBLP:conf/concur/MajumdarMSZ21}.
We construct a tree~$\mathcal{T}$ such that each node represents a run~$\rho$ of some finite prefix $w'$ of $w$.
The root's label is the empty run.
For every node labelled with $\rho$, the children's labels extend $\rho$ by a single transition.
The tree $\mathcal{T}$ is finitely branching by construction of $\semglobalsync(\GG)$.
With König's Lemma, we obtain an infinite path in $\mathcal{T}$ and thus an infinite run~$\rho$ in $\semglobal(\GG)$ with $w \preceq^\omega_\interswap \SyncToAsync(\trace(\rho))$.
From this, it follows that
$w \in \lang_{\inf}(\GG)$.

\textbf{Second}, we show that
$\lang_{\inf}(\GG) \subseteq \lang_{\inf}(\CSM{A})$.
Let $w$ be a word in $\lang_{\inf}(\GG)$.
Eventually, we will apply the same reasoning with König's lemma to obtain an infinite run in $\CSM{A}$ for $w$.
Inspired from the first statement of Lemma 25 by Majumdar et al.~\cite{DBLP:conf/concur/MajumdarMSZ21}, we show:
\vspace{-1ex}
\begin{enumerate}[(i)]
 \item for every prefix $w' \in \pref(w)$, there is a run $\rho'$ in $\CSM{A}$ such that $w' \preceq \trace(\rho')$, and\label{claim:prefix-run}
 \item for every extension $w'x$ where $x$ is the next event in $w$, the run $\rho'$ can be extended. \label{claim:extendable}
\end{enumerate}
\vspace{-1ex}
We prove Claim \ref{claim:prefix-run} first.
We first observe that, since $\GG$ is $0$-reachable, there is an extension $w''$ of $w'$ with $w'' \in \lang(\GG)$.
By construction, we know that there is a run $\rho''$ in $\CSM{A}$ for~$w''$.
For $\rho'$, we can simply take the prefix of $\rho''$ that matches~$w'$.
This proves Claim \ref{claim:prefix-run}.

Now, let us prove Claim \ref{claim:extendable}.
Similar to the first case, we will use prophecy variables and an oracle to pick the correct witness run that we can extend.
Again, because of monotonicity, any run (or rather a prefix of it) for an extension $w'x$ can also be used as witness run for $w'$.
As before, we make use of the idea of prophecy variables \cite{DBLP:conf/lics/AbadiL88}, assume an oracle which picks the correct witness run for every prefix $w'$, and this oracle does not restrict the roles in any way.
From this, Claim \ref{claim:extendable} follows.

From here, we (again) use the same reasoning as Majumdar et al.\ for the proof of Lemma~41~\cite{DBLP:conf/concur/MajumdarMSZ21}.
We construct a tree $\mathcal{T}$ such that each node represents a run $\rho$ of some finite prefix $w'$ of~$w$.
The root's label is the empty run.
For every node labelled with $\rho$, the children's labels extend~$\rho$ by a single transition.
The tree $\mathcal{T}$ is finitely branching by construction of $A_\procA$ for every role~$\procA$.
With König's Lemma, we obtain an infinite path in~$\mathcal{T}$ and, thus, an infinite run~$\rho$ in $\CSM{A}$ with $w \preceq^\omega_\interswap \trace(\rho)$.
From this, it follows that
$w \in \lang(\CSM{A})$.
\end{proof}

\subsection{Formalisation for \cref{lm:implementability-entails-glob-coop}: \newline Implementability entails Globally Cooperative}
\label{proof:imp-entails-glob-coop}

\begin{definition}[Matching Sends and Receptions] In a word $w = e_1 \ldots \in \Alphabet^\infty$,
a send event $e_i = \snd{\procA}{\procB}{\val}$ is said to be \emph{matched} by a receive event $e_j = \rcv{\procA}{\procB}{\val}$ if $i < j$ and
$\MsgVals((e_1 \ldots e_i) \wproj_{\snd{\procA}{\procB}{\_}})$
=
$\MsgVals((e_1 \ldots e_j) \wproj_{\rcv{\procA}{\procB}{\_}})$.
\end{definition}

\implementabilityEntailsGlobCoop*

\begin{proof}
We prove our claim by contraposition:
assume there is a loop $\vertexA_1, \ldots, \vertexA_n$ such that the communication graph of $\mu(\vertexA_1)\ldots\mu(\vertexA_n)$ is not weakly connected.
By construction of $H(\GG)$, we know that every vertex is reachable so
there is a path
$
\vertexB_1 \ldots \vertexB_m
\vertexA_1 \ldots \vertexA_n
$
in $H(\GG)$
for some~$m$ and vertices $\vertexB_1$ to $\vertexB_m$ such that $\vertexB_1 = v^I$.
Because $\GG$ is $0$-reachable,
this path can be completed to end in a terminal vertex to obtain
$
\vertexB_1 \ldots \vertexB_m
\vertexA_1 \ldots \vertexA_n
\vertexB_{m+1} \ldots \vertexB_{m+k}
$
for some $k$ and vertices $\vertexB_{m+1}$ to $\vertexB_{m+k}$ such that $\vertexB_{k+m} \in V^T$.
By the syntax of global types and the construction of $H(\GG)$, there is a role $\procA$ that is the (only) sender in $\vertexA_1$ and $\vertexB_{m+1}$.

Without loss of generality, let $\mathcal{S}_1$ and $\mathcal{S}_2$ be the two sets of (active) roles whose communication graphs of
$\vertexA_1 \ldots \vertexA_n$
are weakly connected and their union consists of all active roles.
Similar reasoning applies if there are more than two sets.

We want to consider specific linearisations from the language of the BMSC of each subpath.
Intuitively, these simply follow the order prescribed by the global type and do not exploit the partial order of BMSC or the closure of the semantics for global types.
For this, we say that $w_1$ is the \emph{canonical word for path} $\vertexB_1, \ldots \vertexB_m$ if
$w_1 \in \set{w'_1 \ldots w'_m \mid w'_i \in \lang(\mu(\vertexB_i)) \text{ for } 1 \leq i \leq m}$.
Analogously, let $w_2$ be the canonical word for $\vertexA_1 \ldots \vertexA_n$ and $w_3$ be the canonical word for $\vertexB_{m+1} \ldots \vertexB_{m+k}$.
Without loss of generality,
$\mathcal{S}_1$ contains the sender of the first element in $w_2$ and $w_3$ --- basically the role which decides when to exit the loop for the considered loop branch.
Let $\CSMl{H(\GG) \wproj_\procA}$ be the erasure candidate implementation.
By its definition and the correctness of $H(\GG)$, it holds that:
$
 \lang(\GG) = \lang(H(\GG)) .
$
With the equivalence of the canonical candidate implementation
(\cref{cor:erasure-canonical}),
the reasoning for Lemma 3.2 by Lohrey~\cite{DBLP:journals/tcs/Lohrey03}, and the fact that it generalises to infinite executions (\cref{lm:finite-generalises-to-infinite}),
the erasure candidate implementation admits at least the language specified by $H(\GG)$:
\vspace{-2ex}
\[
 \lang(H(\GG))
 \subseteq
 \lang(\CSMl{H(\GG) \wproj_\procA}) .
\]
Thus, it holds that
$
\lang(\GG)
=
\lang(\CSMl{H(\GG) \wproj_\procA})
$
if $\GG$ is implementable.
Therefore, we know that
$w_1\,.\,w_2\,.\,w_3 \in \lang(\CSMl{H(\GG) \wproj_\procA})$.

From the construction of $H(\GG)$ and the construction of $w_i$ for $i \in \set{1,2,3}$,
it also holds that
$w_1\,.\,(w_2)^h\,.\,w_3 \in \lang(H(\GG)) \subseteq \lang(\CSMl{H(\GG) \wproj_\procA})$ for any $h > 0$.

By construction of $\mathcal{S}_1$ and $\mathcal{S}_2$, no two roles from both sets communicate with each other in~$w_2$: there are no $\procC \in \mathcal{S}_1$ and $\procD \in \mathcal{S}_2$ such that
$\snd{\procC}{\procD}{\val}$ is in $w_2$
or
$\snd{\procD}{\procC}{\val}$ is in $w_2$
(and consequently
$\rcv{\procD}{\procC}{\val}$ is in $w_2$
or
$\rcv{\procC}{\procD}{\val}$ is in $w_2$)
for any $\val$.

From the previous two observations, it follows that
\vspace{-2ex}
\[
w_1\,.\,w_2\,.\,(w_2 \wproj_{\Alphabet_{\mathcal{S}_1}})^h.\,w_3 \in \lang(\CSMl{H(\GG) \wproj_\procA})
\]
for any $h$
where
$\Alphabet_{\mathcal{S}_1}
    =
 \Union_{\procC \in \mathcal{S}_1} \Alphabet_{\procC}$.
Intuitively, this means that the set of roles with the role to decide when to exit the loop can continue longer in the loop than the roles in $\mathcal{S}_2$.

With $\lang(\GG) = \lang(H(\GG))$, it suffices to show the following to find a contradiction:
$w_1\,.\,w_2\,.\,(w_2 \wproj_{\Alphabet_{\mathcal{S}_1}}).\,w_3 \notin \lang(H(\GG))$.

Towards a contradiction, we assume the membership holds.
By determinacy of $H(\GG)$, we need to find a path $\vertexA'_1 \ldots \vertexA'_{m'}$, that starts at the beginning of the loop, i.e., $\vertexA'_1 = \vertexA_1$,
with canonical word $w_4$ such that
$w_2 \wproj_{\Alphabet_{\mathcal{S}_1}}.\,w_3 \interswap w_4$.

We show such a path cannot exist and that we would need to diverge during the loop.

We denote $w_2\,.\,w_3$ with $x \is x_1 \ldots x_{l}$ and $w_2 \wproj_{\Alphabet_{\mathcal{S}_1}}.\,w_3$ with $x' \is x'_1 \ldots x'_{l'}$.
We know that $x'$ is a subsequence of $x$.
Let $x_1 \ldots x_j = x'_1 \ldots x'_j$ denote the maximal prefix on which both agree.
Since $\mathcal{S}_2$ is not empty, we know that $j$ can be at most $\card{w_2 \wproj_{\Alphabet_{\mathcal{S}_1}}\negmedspace}$.
(Intuitively, $j$ cannot be so big that it reaches $w_3$ because there will be mismatches due to $w_2 \wproj_{\Alphabet_{\mathcal{S}_2}}$ before.)
We also claim that the next event $x_{j+1}$ cannot be a receive event.
If it was, there was a matching send event in $x_1 \ldots x_j$ (which is equal to $x'_1 \ldots x'_j$ by construction).
Such a matching send event exists by construction of $x$ from a path in $H(\GG)$.
By definition of $\wproj_{\hole}$, the matching receive event must be $x'_{j+1}$ which would contradict the maximality of $j$.
Thus, $x_{j+1}$ must be a send event.

By determinacy of $H(\GG)$ and $j \leq \card{w_2 \wproj_{\Alphabet_{\mathcal{S}_1}}\negmedspace}$, we know that $x_1 \ldots x_j = x'_1 \ldots x'_j$ share a path $\vertexA_1 \ldots \vertexA_{n'}$ which is a part of the loop, i.e., $x_1 \ldots x_j \in \lang(\mu(\vertexA_1) \cdots \mu(\vertexA_{n'}))$ with $n < n'$.
For $M( \msgFromTo{\procA}{\procB}{\val} )$ --- the BMSC with solely this interaction from~\cref{def:hmsc-encoding}, we say that $\procA$ is its \emph{sender}.
The syntax of global types prescribes that choice is deterministic and the sender in a choice is unique.
This is preserved for $H(\GG)$: for every vertex, all its successors have the same sender.
Therefore, the path for $x'$ can only diverge, but also needs to diverge, from the loop $\vertexA_1 \ldots \vertexA_n$ after the common prefix $\vertexA_1 \ldots \vertexA_{n'}$ with a different send event but with the same sender.
Let $\vertexA_{l}$ be next vertex after $\vertexA_1 \ldots \vertexA_{n'}$ on the loop $\vertexA_1, \ldots, \vertexA_n$ for which $\mu(\vertexA_{l})$ is not $M_\emptystring$ --- the BMSC with an empty set of event nodes from \cref{def:hmsc-encoding}.
Note that $x_{j+1}$ belongs to $\vertexA_l$:
$x_{j+1} \in \pref(\lang(\mu(\vertexA_l)))$.

We do another case analysis whether $x_{j+1}$ belongs to $\mathcal{S}_1$ or not, i.e., if $x_{j+1} \in \Alphabet_{\mathcal{S}_1}$.

If $x_{j+1} \notin \Alphabet_{\mathcal{S}_1}$, there cannot be a path that continues for $x'_{j+1} \in \Alphabet_{\mathcal{S}_1}$ as the sender for $\mu(\vertexA_{l})$ is not in $\mathcal{S}_1$.
If $x_{j+1} \in \Alphabet_{\mathcal{S}_1}$, the choice of $j$ was not maximal which yields a contradiction.
\end{proof}

\subsection{Further Explanation for \cref{ex:implementable-not-glob-coop}}
\label{app:explanation-implementable-not-glob-coop}

Here, we show that any trace of the CSM in \cref{fig:csm-ing} is specified by the HMSC in \cref{fig:implementable-hmsc-not-glob-coop}.
Let us consider a finite execution of the CSM for which we want to find a path in the HMSC.
Let us assume there are $i$ interactions between $\procA$ and $\procB$ and $j$ interactions between $\procC$ and $\procD$.
In our asynchronous setting, these interactions are split and can be interleaved.
From the CSM, it is easy to see that both $i$ and $j$ are each at least $2$.
The simplest path goes through the first loop once and accounts for $i - 1$ iterations in the second loop and $j - 1$ iterations in the third one.
A more involved path could account for $\min(i,j)-1$ iterations of the first loop, as many as possible, and
$i - \min(i,j) + 1$ iterations of the second loop as well as
$j - \min(i,j) + 1$ iterations of the third loop.
The key that both paths are valid possibilities is that the interactions of $\procA$ and $\procB$ in the first and second loop are indistinguishable, i.e., the executions can be reordered with $\interswap$ such that both is possible.
The syntactic restriction on choice does prevent this for global types (and this protocol cannot be represented with a global type).
Intuitively, one cannot make up for a different number of loop iterations, which are the consequence of missing synchronisation, in global types because the ''loop exit``-message will be distinct (compared to staying in the loop) and anything specified afterwards cannot be reordered with $\interswap$ in front of~it.
It~is straightforward to adapt the protocol so final states do not have outgoing transitions.
We add another vertex with a BMSC at the bottom, which has the same structure as the top one but with another message $l$ instead of $m$.
We add an edge from the previous terminal vertex to the new vertex and make the new one the only terminal vertex.
With this, $\procA$ and $\procC$ can eventually decide not to send $m$ anymore and indicate their choice with the distinct message $l$ to the other two roles.
 \section{Proof for \cref{lm:correctness-checking-I-closedness}: \newline Correctness of \cref{alg:checking-I-closedness} to check $\pmb{\mathcal{I}}$-closedness of Global Types}
\label{proof:correctness-checking-I-closedness}

\correctnessCheckingIClosedness*

\begin{proof}
It is obvious that the language is preserved by the changes to the state machine.
(We basically turned an unambiguous state machine into a deterministic one.)

For soundness, we assume that \cref{alg:checking-I-closedness} returns \emph{true} and let $w$ be a word in $\indeprellang(\lang(\semglobalsync(\GG)))$.
By definition, there is a run with trace $w'$ in $\semglobalsync(\GG)$ such that $w' \equiv_{\mathcal{I}} w$.
The conditions in \cref{alg:checking-I-closedness} ensure that $w = w'$ because no two adjacent elements in $w'$ can be reordered with $\equiv_{\mathcal{I}}$.
Therefore, $w \in \lang(\semglobalsync(\GG))$ which proves the claim.

For completeness, we assume that \cref{alg:checking-I-closedness} returns \emph{false} and show that there is $w \in \indeprellang(\lang_{\fin}(\GG))$ such that
$w \notin \lang_{\fin}(\semglobalsync(\GG))$.
Without loss of generality, let $q_2$ be the state for which an incoming label $x$ and outgoing label $y$ can be reordered,
i.e.,
$x \equiv_{\mathcal{I}} y$,
and let $q_1$ be the state from which the transition with label $x$ originates: $q_1 \xrightarrow{x} q_2 \in \delta_{\semglobalsync(\GG)}$.
We consider a word $w'$ which is the trace of a maximal run that passes $q$ and the transitions labelled with $x$ and~$y$.
By construction, it holds that $w' \in \lang_{\fin}(\semglobalsync(\GG))$.
We swap $x$ and $y$ in $w'$ to obtain~$w$.
We denote
$x$ with $\msgFromTo{\procA}{\procB}{\val}$ and
$y$ with $\msgFromTo{\procC}{\procD}{\val'}$
such that
$\set{\procA, \procB} \inters \set{\procC, \procD} \neq \emptyset$.
From the syntactic restrictions of global types, we know that any transition label from $q_1$ has sender $\procA$ while every transition label from $q_2$ has sender $\procC$.
Because of this and determinacy of the state machine, there is no run in $\semglobalsync(\GG)$ with trace $w'$.
Thus, $w \notin \lang_{\fin}(\semglobalsync(\GG))$ which concludes the proof.
\end{proof}
 \section{Proof for \cref{thm:implementability-intraswap-reordering-undecidable}: Implementability with regard to \newline Intra-role Reordering for Global Types from MSTs is Undecidable}
\label{proof:undec-intraswap}

\implementabilityIntraswapReorderingUndecidable*

\begin{proof}
Let $\set{(u_1,u_2,\ldots,u_n),(v_1,v_2,\ldots,v_n)}$ be an instance of MPCP where $1$ is the special index which each solution needs to start with.
We construct a global type where, for a word $w = a_1a_2\cdots a_m \in \Delta\negmedspace^*$,  a message labelled $[w]$ denotes a sequence of individual message interactions with message $a_1$, $a_2$, \ldots, $a_m$, each of size $1$.
We define a parametric global type where $x \in \set{u, v}$: 

\vspace{-3ex}
{
\scriptsize \[
G(x, X)
    \quad \is \quad
    \msgFromTo{\procA}{\procB}{\textit{c-}x}. \,
    \msgFromTo{\procA}{\procB}{1}. \,
    \msgFromTo{\procA}{\procC}{1}. \,
    \msgFromTo{\procB}{\procC}{[x_1]}. \,
\mu t. \,
+
    \begin{cases}
    \msgFromTo{\procA}{\procB}{1}. \,
    \msgFromTo{\procA}{\procC}{1}. \,
    \msgFromTo{\procB}{\procC}{[x_1]}. \, t
    \\
    \cdots
    \\
    \msgFromTo{\procA}{\procB}{n}. \,
    \msgFromTo{\procA}{\procC}{n}. \,
    \msgFromTo{\procB}{\procC}{[x_n]}. \, t
    \\
    \msgFromTo{\procA}{\procB}{\doneshort}. \,
    \msgFromTo{\procA}{\procC}{\doneshort}. \,
    \msgFromTo{\procB}{\procC}{\doneshort}. \,
    X
\end{cases}
\]
}

\noindent where $\textit{c-}x$ indicates \emph{choosing} tile set $x$.
We use the message $\textit{ack-}x$ to indicate \emph{acknowledgement} of tile set $x$.
With this, we obtain our encoding:

\vspace{-2ex}
{ \scriptsize \[
\GG_{\MPCP}
    \quad \is \quad
+
\begin{cases}
G(u, \msgFromTo{\procC}{\procA}{\textit{ack-}u}. \, 0)
\\
G(v, \msgFromTo{\procC}{\procA}{\textit{ack-}v}. \, 0)
\end{cases}
. \] }

\begin{figure}[t]
\centering
\def\svgwidth{0.75\textwidth}
\begingroup \makeatletter \providecommand\color[2][]{\errmessage{(Inkscape) Color is used for the text in Inkscape, but the package 'color.sty' is not loaded}\renewcommand\color[2][]{}}\providecommand\transparent[1]{\errmessage{(Inkscape) Transparency is used (non-zero) for the text in Inkscape, but the package 'transparent.sty' is not loaded}\renewcommand\transparent[1]{}}\providecommand\rotatebox[2]{#2}\newcommand*\fsize{\dimexpr\f@size pt\relax}\newcommand*\lineheight[1]{\fontsize{\fsize}{#1\fsize}\selectfont}\ifx\svgwidth\undefined \setlength{\unitlength}{774.58717166bp}\ifx\svgscale\undefined \relax \else \setlength{\unitlength}{\unitlength * \real{\svgscale}}\fi \else \setlength{\unitlength}{\svgwidth}\fi \global\let\svgwidth\undefined \global\let\svgscale\undefined \makeatother \begin{picture}(1,0.34530324)\lineheight{1}\setlength\tabcolsep{0pt}\put(0,0){\includegraphics[width=\unitlength,page=1]{figs/hmsc-mpcp-encoding.pdf}}\put(0.42748298,0.29847817){\color[rgb]{0,0,0}\makebox(0,0)[lt]{\lineheight{1.25}\smash{\begin{tabular}[t]{l}{\scriptsize $\procA$}\end{tabular}}}}\put(0,0){\includegraphics[width=\unitlength,page=2]{figs/hmsc-mpcp-encoding.pdf}}\put(0.2600398,0.27861764){\color[rgb]{0,0,0}\makebox(0,0)[lt]{\lineheight{1.25}\smash{\begin{tabular}[t]{l}{\tiny $\textit{c-}u$}\end{tabular}}}}\put(0.3156623,0.22173455){\color[rgb]{0,0,0}\makebox(0,0)[lt]{\lineheight{1.25}\smash{\begin{tabular}[t]{l}{\tiny $[u_1]$}\end{tabular}}}}\put(0.26876969,0.25623815){\color[rgb]{0,0,0}\makebox(0,0)[lt]{\lineheight{1.25}\smash{\begin{tabular}[t]{l}{\tiny $1$}\end{tabular}}}}\put(0.33316407,0.24034016){\color[rgb]{0,0,0}\makebox(0,0)[lt]{\lineheight{1.25}\smash{\begin{tabular}[t]{l}{\tiny $1$}\end{tabular}}}}\put(0,0){\includegraphics[width=\unitlength,page=3]{figs/hmsc-mpcp-encoding.pdf}}\put(0.08737807,0.03310352){\color[rgb]{0,0,0}\makebox(0,0)[lt]{\lineheight{1.25}\smash{\begin{tabular}[t]{l}{\tiny $[u_1]$}\end{tabular}}}}\put(0.04048539,0.06729812){\color[rgb]{0,0,0}\makebox(0,0)[lt]{\lineheight{1.25}\smash{\begin{tabular}[t]{l}{\tiny $1$}\end{tabular}}}}\put(0.10487985,0.05164879){\color[rgb]{0,0,0}\makebox(0,0)[lt]{\lineheight{1.25}\smash{\begin{tabular}[t]{l}{\tiny $1$}\end{tabular}}}}\put(0,0){\includegraphics[width=\unitlength,page=4]{figs/hmsc-mpcp-encoding.pdf}}\put(0.10563762,0.17445457){\color[rgb]{0,0,0}\makebox(0,0)[lt]{\lineheight{1.25}\smash{\begin{tabular}[t]{l}{\tiny \doneshort}\end{tabular}}}}\put(0.02975806,0.14078956){\color[rgb]{0,0,0}\makebox(0,0)[lt]{\lineheight{1.25}\smash{\begin{tabular}[t]{l}{\tiny $\textit{ack-}u$}\end{tabular}}}}\put(0,0){\includegraphics[width=\unitlength,page=5]{figs/hmsc-mpcp-encoding.pdf}}\put(0.3146029,0.03222686){\color[rgb]{0,0,0}\makebox(0,0)[lt]{\lineheight{1.25}\smash{\begin{tabular}[t]{l}{\tiny $[u_n]$}\end{tabular}}}}\put(0.26871023,0.06682158){\color[rgb]{0,0,0}\makebox(0,0)[lt]{\lineheight{1.25}\smash{\begin{tabular}[t]{l}{\tiny $n$}\end{tabular}}}}\put(0.17904244,0.05279385){\color[rgb]{0,0,0}\makebox(0,0)[lt]{\lineheight{1.25}\smash{\begin{tabular}[t]{l}\textit{...}\end{tabular}}}}\put(0.33310464,0.05148574){\color[rgb]{0,0,0}\makebox(0,0)[lt]{\lineheight{1.25}\smash{\begin{tabular}[t]{l}{\tiny $n$}\end{tabular}}}}\put(0.4904877,0.29866708){\color[rgb]{0,0,0}\makebox(0,0)[lt]{\lineheight{1.25}\smash{\begin{tabular}[t]{l}{\scriptsize $\procB$}\end{tabular}}}}\put(0.55517307,0.29732859){\color[rgb]{0,0,0}\makebox(0,0)[lt]{\lineheight{1.25}\smash{\begin{tabular}[t]{l}{\scriptsize $\procC$}\end{tabular}}}}\put(0,0){\includegraphics[width=\unitlength,page=6]{figs/hmsc-mpcp-encoding.pdf}}\put(0.6364881,0.27948344){\color[rgb]{0,0,0}\makebox(0,0)[lt]{\lineheight{1.25}\smash{\begin{tabular}[t]{l}{\tiny $\textit{c-}v$}\end{tabular}}}}\put(0.69219289,0.22227469){\color[rgb]{0,0,0}\makebox(0,0)[lt]{\lineheight{1.25}\smash{\begin{tabular}[t]{l}{\tiny $[v_1]$}\end{tabular}}}}\put(0.6449834,0.2562057){\color[rgb]{0,0,0}\makebox(0,0)[lt]{\lineheight{1.25}\smash{\begin{tabular}[t]{l}{\tiny $1$}\end{tabular}}}}\put(0.70937696,0.24087031){\color[rgb]{0,0,0}\makebox(0,0)[lt]{\lineheight{1.25}\smash{\begin{tabular}[t]{l}{\tiny $1$}\end{tabular}}}}\put(0,0){\includegraphics[width=\unitlength,page=7]{figs/hmsc-mpcp-encoding.pdf}}\put(0.9244207,0.03334701){\color[rgb]{0,0,0}\makebox(0,0)[lt]{\lineheight{1.25}\smash{\begin{tabular}[t]{l}{\tiny $[v_n]$}\end{tabular}}}}\put(0.87902821,0.06749055){\color[rgb]{0,0,0}\makebox(0,0)[lt]{\lineheight{1.25}\smash{\begin{tabular}[t]{l}{\tiny $n$}\end{tabular}}}}\put(0.94342258,0.05159228){\color[rgb]{0,0,0}\makebox(0,0)[lt]{\lineheight{1.25}\smash{\begin{tabular}[t]{l}{\tiny $n$}\end{tabular}}}}\put(0,0){\includegraphics[width=\unitlength,page=8]{figs/hmsc-mpcp-encoding.pdf}}\put(0.69219289,0.03289537){\color[rgb]{0,0,0}\makebox(0,0)[lt]{\lineheight{1.25}\smash{\begin{tabular}[t]{l}{\tiny $[v_1]$}\end{tabular}}}}\put(0.64492371,0.06679008){\color[rgb]{0,0,0}\makebox(0,0)[lt]{\lineheight{1.25}\smash{\begin{tabular}[t]{l}{\tiny $1$}\end{tabular}}}}\put(0.81802558,0.0506373){\color[rgb]{0,0,0}\rotatebox{-180}{\makebox(0,0)[lt]{\lineheight{1.25}\smash{\begin{tabular}[t]{l}\textit{...}\end{tabular}}}}}\put(0.70931766,0.05089171){\color[rgb]{0,0,0}\makebox(0,0)[lt]{\lineheight{1.25}\smash{\begin{tabular}[t]{l}{\tiny $1$}\end{tabular}}}}\put(0,0){\includegraphics[width=\unitlength,page=9]{figs/hmsc-mpcp-encoding.pdf}}\put(0.04086324,0.18843561){\color[rgb]{0,0,0}\makebox(0,0)[lt]{\lineheight{1.25}\smash{\begin{tabular}[t]{l}{\tiny \doneshort}\end{tabular}}}}\put(0,0){\includegraphics[width=\unitlength,page=10]{figs/hmsc-mpcp-encoding.pdf}}\put(0.10563762,0.15122597){\color[rgb]{0,0,0}\makebox(0,0)[lt]{\lineheight{1.25}\smash{\begin{tabular}[t]{l}{\tiny \doneshort}\end{tabular}}}}\put(0,0){\includegraphics[width=\unitlength,page=11]{figs/hmsc-mpcp-encoding.pdf}}\put(0.94454514,0.17216115){\color[rgb]{0,0,0}\makebox(0,0)[lt]{\lineheight{1.25}\smash{\begin{tabular}[t]{l}{\tiny \doneshort}\end{tabular}}}}\put(0.86797851,0.13936641){\color[rgb]{0,0,0}\makebox(0,0)[lt]{\lineheight{1.25}\smash{\begin{tabular}[t]{l}{\tiny $\textit{ack-}v$}\end{tabular}}}}\put(0,0){\includegraphics[width=\unitlength,page=12]{figs/hmsc-mpcp-encoding.pdf}}\put(0.87977077,0.18614219){\color[rgb]{0,0,0}\makebox(0,0)[lt]{\lineheight{1.25}\smash{\begin{tabular}[t]{l}{\tiny \doneshort}\end{tabular}}}}\put(0,0){\includegraphics[width=\unitlength,page=13]{figs/hmsc-mpcp-encoding.pdf}}\put(0.94454514,0.14893255){\color[rgb]{0,0,0}\makebox(0,0)[lt]{\lineheight{1.25}\smash{\begin{tabular}[t]{l}{\tiny \doneshort}\end{tabular}}}}\put(0,0){\includegraphics[width=\unitlength,page=14]{figs/hmsc-mpcp-encoding.pdf}}\end{picture}\endgroup  \caption{HMSC encoding $H(\GG_{\MPCP})$ of the MPCP encoding (same as in \cref{fig:hmsc-mpcp-encoding})}
\label{fig:hmsc-mpcp-encoding-2}
\vspace{-2.5ex}
\end{figure}

\Cref{fig:hmsc-mpcp-encoding-2} illustrates its HMSC encoding
$H(\GG_{\MPCP})$.

It suffices to show the following equivalences:
\vspace{-1ex}
\begin{align*}
&
\quad
\GG_{\MPCP} \text{ is } \intraswap \text{-implementable} \\
\Leftrightarrow_1
&
\quad
\intraswaplang(\lang(G(u, 0))) \wproj_{\Alphabet_\procC}
    \inters
\intraswaplang(\lang(G(v, 0))) \wproj_{\Alphabet_\procC}
    =
\emptyset
\\
\Leftrightarrow_2
&
\quad
\text{MPCP instance has no solution}
\end{align*}

We prove $\Rightarrow_1$ by contraposition.
Let $w \in
\intraswaplang(\lang(G(u, 0))) \wproj_{\Alphabet_\procC}
    \inters
\intraswaplang(\lang(G(v, 0))) \wproj_{\Alphabet_\procC}$.
For $x \in \set{u, v}$,
let $w_x \in
\intraswaplang(\lang(G(x, 0)))$
such that $w_x \wproj_{\Alphabet_\procC} = w$.
By construction of $\GG_{\MPCP}$, we know that
$w_x\,.\,\snd{\procC}{\procA}{\textit{ack-}x}\,.\,\rcv{\procC}{\procA}{\textit{ack-}x}
\in
\intraswaplang(\lang(\GG_{\MPCP}))$.

Suppose that CSM $\CSM{A}$ $\intraswap$-implements $\GG_{\MPCP}$.
Then, it holds that
\vspace{-2ex}
\[
w_x\,.\,\snd{\procC}{\procA}{\textit{ack-}x}\,.\,\rcv{\procC}{\procA}{\textit{ack-}x}
\in
\intraswaplang(\lang(\CSM{A}))
\]
by (ii) from \cref{def:intra-role-implementability}.
We also know that
$w_x\,.\,\snd{\procC}{\procA}{\textit{ack-}y}\,.\,\rcv{\procC}{\procA}{\textit{ack-}y}
\notin
\intraswaplang(\lang(\GG_{\MPCP}))$ for $x \neq y$ where $x, y \in \set{u,v}$.
By the choice of $w_u$ and $w_v$, it holds that
$w_u \wproj_{\Alphabet_\procC} = w = w_v \wproj_{\Alphabet_\procC}$.
Therefore, $\procC$ needs to be in the same state of~$A_\procC$ after processing
$w_u \wproj_{\Alphabet_\procC}$ or $w_v \wproj_{\Alphabet_\procC}$ and it can either send both $\textit{ack-}u$ and $\textit{ack-}v$, only one of them or none of them to $\procA$.
Thus, either one of the following is true:
\vspace{-1ex}
\begin{itemize}
 \item [a)] (sending both) $w_x\,.\,\snd{\procC}{\procA}{\textit{ack-}y} \in \pref(\intraswaplang(\lang(\CSM{A})))$ for $x \neq y$ where $x,y \in \set{u,v}$, or
 \item [b)] (sending $u$ without loss of generality) $w_v\,.\,\snd{\procC}{\procA}{\textit{ack-}u} \notin \pref(\intraswaplang(\lang(\CSM{A})))$, or
 \item [c)] (sending none) $w_x\,.\,\snd{\procC}{\procA}{\textit{ack-}x} \notin \pref(\intraswaplang(\lang(\CSM{A})))$ for $x \in \set{u,v}$.
\end{itemize}
\vspace{-1ex}
All cases lead to deadlocks in $\CSM{A}$.
For a) and for b) if $\textit{c-}v$ was chosen in the beginning, $\procA$ cannot receive the sent message as it disagrees with its choice from the beginning $\textit{c-}x$.
In all other cases, $\procA$ waits for a message while no message will ever be sent.
Having deadlocks contradicts the assumption that $\CSM{A}$ $\intraswap$-implements~$\GG$ (and there cannot be any CSM that $\intraswap$-implements $\GG$).

We prove $\Leftarrow_1$ next.
The language $\intraswaplang(\lang(\GG_{\MPCP}))$ is obviously non-empty.
Therefore, let $w' \in \intraswaplang(\lang(\GG_{\MPCP}))$.
We split $w'$ to obtain:
\vspace{-2ex}
\[
w' = w \,.\,\snd{\procC}{\procA}{\textit{ack-}x}\,.\,\rcv{\procC}{\procA}{\textit{ack-}x} \text{ for some $w$ and $x \in \set{u,v}$}.
\]
By construction of $\GG_{\MPCP}$, we know that
\vspace{-2ex}
\[
w
    \in
\intraswaplang(\lang(G(u, 0)))
    \; \union \;
\intraswaplang(\lang(G(v, 0))).
\]
By assumption, it follows that exactly one of the following holds:
\vspace{-2ex}
\[
    w \wproj_{\Alphabet_\procC}
        \in
    \intraswaplang(\lang(G(u, 0))) \wproj_{\Alphabet_\procC}
    \quad
    \text{ or }
    \quad
    w \wproj_{\Alphabet_\procC}
        \in
    \intraswaplang(\lang(G(v, 0))) \wproj_{\Alphabet_\procC}.
\]

We give a $\intraswap$-implementation for $\GG_{\MPCP}$.
It is straightforward to construct FSMs for both $\procA$ and $\procB$.
They are involved in the initial decision and $\intraswap$ does not affect their projected languages.
Thus, the projection by erasure can be applied to obtain FSMs $A_\procA$ and $A_\procB$.
We construct an FSM $A_\procC$ for $\procC$ with control state
$\texttt{i} \in \set{1, \ldots, n}$,
$\texttt{j} \in \set{1, \ldots, \max(\card{u_i} \mid i \in \set{1,\ldots,n})}$,
$\texttt{d} \in \set{0,1,2}$, and
$\texttt{x} \in \set{u, v}$,
where $\card{w}$ denotes the length of a word.
The FSM is constructed in a way such that
\vspace{-2ex}
\begin{align*}
    & w \wproj_{\Alphabet_\procC}
        \in
    \intraswaplang(\lang(G(u, 0))) \wproj_{\Alphabet_\procC}
\quad \text{ if and only if } \quad
\texttt{d} \text{ is } 2 \text{ and }
\texttt{x} \text{ is } u
\qquad \text{ as well as }
\\
    & w \wproj_{\Alphabet_\procC}
        \in
    \intraswaplang(\lang(G(v, 0))) \wproj_{\Alphabet_\procC}
\quad \text{ if and only if } \quad
\texttt{d} \text{ is } 2 \text{ and }
\texttt{x} \text{ is } v.
\end{align*}

We first explain that this characterisation suffices to show that $\CSM{A}$ $\intraswap$-implements~$\GG$.
The control state $\texttt{d}$ counts the number of received $\doneshort$-messages.
Thus, there will be no more messages to $\procC$ in any channel once $\texttt{d}$ is $2$ by construction of $\GG_{\MPCP}$.
Once in a state for which $\texttt{d}$ is~$2$, $\procC$ sends message $\textit{ack-}u$ to $\procA$ if $\texttt{x}$ is $u$ and message $\textit{ack-}v$ if $\texttt{x}$ is $v$.
With the characterisation, this message $\textit{ack-}x$ matches the message $\textit{c-}x$ sent from $\procA$ to $\procB$ in the beginning and, thus, $\procA$ will be able to receive it and conclude the execution.

Now, we will explain how to construct the FSM $A_\procC$.
Intuitively, $\procC$ keeps a tile number, which it tries to match against, and stores this in $\texttt{i}$.
It is initially set to $0$ to indicate no tile has been chosen yet.
The index $\texttt{j}$ denotes the position of the letter it needs to match in tile $u_{\texttt{i}}$ next and, thus, is  initialised to $1$.
The variable $\texttt{d}$ indicates the number of $\doneshort$-messages received so far, so initially $\texttt{d}$ is $0$.
With this, $\procC$ knows when it needs to send $\textit{ack-}x$.
The FSM for $\procC$ tries to match the received messages against the tiles of $u$, so $\texttt{x}$ is initialised to~$u$.
If this matching fails at some point, $\texttt{x}$ is set to $v$ as it learned that $v$ was chosen initially by~$\procA$.

\noindent In any of the following cases: if a received message is a $\doneshort$-message, $\texttt{d}$ is solely increased by~$1$:
\vspace{-1ex}
\begin{itemize}
 \item If $\texttt{x}$ is $u$ and $\texttt{i}$ is $0$, $\procC$ receives a message $z$ from $\procA$ and sets $\texttt{i}$ to $z$ (technically the integer represented by $z$).
 \item If $\texttt{x}$ is $u$ and $\texttt{i}$ is not $0$, $\procC$ receives a message $z$ from $\procB$.
    \vspace{-1ex}
 \begin{itemize}
  \item If $z$ is the same as $u_{\texttt{i}}[\texttt{j}]$, we increment $\texttt{j}$ by $1$ and \\ check if $\texttt{j} > \card{u_{\texttt{i}}}$ and, if so, set $\texttt{i}$ to $0$ and $\texttt{j}$ to $1$
  \item If not, we set $\texttt{x}$ to $v$
 \end{itemize}
    \vspace{-1ex}
 \item Once $\texttt{x}$ is $v$, $\procC$ can simply receive all remaining messages in any order. \end{itemize}
\vspace{-1ex}
The described FSM can be used for $\procC$ because it reliably checks whether a presented sequence of indices and words belongs to tile set $u$ or $v$.
It can do so because
$\intraswaplang(\lang(G(u, 0))) \wproj_{\Alphabet_\procC}
    \inters
\intraswaplang(\lang(G(v, 0))) \wproj_{\Alphabet_\procC} = \emptyset$ by assumption.

We prove $\Rightarrow_2$ by contraposition.
Suppose the MPCP instance has a solution.
Let $i_1$, $\ldots$, $i_k$ be a non-empty sequence of indices such that
$
 u_{i_1} u_{i_2} \cdots u_{i_{k}}
 =
 v_{i_1} v_{i_2} \cdots v_{i_{k}}
$
and $i_1 = 1$.
It~is easy to see that
\vspace{-2ex}
{ \small
\[
 w_x \is
 \rcv{\procA}{\procC}{i_1}
 \rcv{\procB}{\procC}{[x_{i_1}]}.\,
 \cdots \,.\,
 \rcv{\procA}{\procC}{i_k}.\,
 \rcv{\procB}{\procC}{[x_{i_k}]}.\,
 \rcv{\procA}{\procC}{\doneshort}.\,
 \rcv{\procB}{\procC}{\doneshort}
\in \lang(G(x, 0)) \wproj_{\Alphabet_\procC}
\text{ for } x \in \set{u, v}.
\]
}
By definition of $\intraswap$, we can re-arrange the previous sequences such that

\vspace{-3ex}
{ \small
\[
 \rcv{\procA}{\procC}{i_1}.\,
 \cdots \,.\,
 \rcv{\procA}{\procC}{i_k}.\,
 \rcv{\procB}{\procC}{[x_{i_1}]}.\,
 \cdots \,.\,
 \rcv{\procB}{\procC}{[x_{i_k}]}.\,
 \rcv{\procA}{\procC}{\doneshort}.\,
 \rcv{\procB}{\procC}{\doneshort}
 \in \intraswaplang(\lang(G(x, 0))) \wproj_{\Alphabet_\procC}
\text{ for } x \in \set{u, v}.
\]
}
Because $i_1, \ldots, i_k$ is a solution to the instance of MPCP, it holds that
\vspace{-2ex}
\[
 \rcv{\procB}{\procC}{[u_{i_1}]}.\,
 \cdots \,.\,
 \rcv{\procB}{\procC}{[u_{i_k}]}
 =
 \rcv{\procB}{\procC}{[v_{i_1}]}.\,
 \cdots \,.\,
 \rcv{\procB}{\procC}{[v_{i_k}]}
\]
and, thus,
\vspace{-2ex}
\[
 \rcv{\procA}{\procC}{i_1}.\,
 \cdots \,.\,
 \rcv{\procA}{\procC}{i_k}.\,
 \rcv{\procB}{\procC}{[u_{i_1}]}.\,
 \cdots \,.\,
 \rcv{\procB}{\procC}{[u_{i_k}]}.\,
 \rcv{\procA}{\procC}{\doneshort}.\,
 \rcv{\procB}{\procC}{\doneshort}
 \text{ is in }
 \intraswaplang(\lang(G(v, 0))) \wproj_{\Alphabet_\procC}.
\]
This shows that
$
 \intraswaplang(\lang(G(u, 0))) \wproj_{\Alphabet_\procC}
    \inters
 \intraswaplang(\lang(G(v, 0))) \wproj_{\Alphabet_\procC}
    \neq
 \emptyset
$.

Lastly, we prove $\Leftarrow_2$.
We know that the MPCP instance has no solution.
Thus, there cannot be a non-empty sequence of indices $i_1$, $i_2$, \ldots, $i_k$ such that
$
 u_{i_1} u_{i_2} \cdots u_{i_{k}}
 =
 v_{i_1} v_{i_2} \cdots v_{i_{k}}
$
and $i_1 = 1$.
For any possible word $w_u \in
\intraswaplang(\lang(G(u, 0))) \wproj_{\Alphabet_\procC}$
and word $w_v \in
\intraswaplang(\lang(G(v, 0))) \wproj_{\Alphabet_\procC}$.

We consider the sequence of receive events $w_x \wproj_{\rcv{\procA}{\procC}{\_}}$ with sender $\procA$ and the sequence of messages $w_x \wproj_{\rcv{\procB}{\procC}{\_}}$ from $\procB$ for $x \in \set{u,v}$.
The intra-role indistinguishability relation $\intraswap$ allows to reorder events of both but for a non-empty intersection of both sets, we would still need to find a word $w_u$ and $w_v$ such that
\vspace{-2ex}
\[
    w_u \wproj_{\rcv{\procA}{\procC}{\_}}
        =
    w_v \wproj_{\rcv{\procA}{\procC}{\_}}
    \quad \text{ and } \quad
    w_u \wproj_{\rcv{\procB}{\procC}{\_}}
        =
    w_v \wproj_{\rcv{\procB}{\procC}{\_}}.
\]
However, $G(x,0)$ for $x \in \set{u,v}$ is constructed in a way that this is only possible if the MPCP instance has a solution.
Therefore, the intersection is empty which proves our claim.
\end{proof}

\end{document}